%% file: samplepaper.tex
\documentclass[runningheads]{llncs}
\input{packages.tex}

\input{customcontrolsequences.tex}
\usepackage[T1]{fontenc}
\usepackage{graphicx}
\begin{document}
\title{Popularity on the 3D-Euclidean Stable Roommates
}
%
%
\author{Steven Ge
	\orcidID{0000-0001-5073-748X}  
	\and
	Toshiya Itoh
	\orcidID{0000-0002-1149-7046}}
\authorrunning{S. Ge and T. Itoh}
%
\institute{Tokyo Institute of Technology, Meguro, Japan\\
	\email{ge.s.aa@m.titech.ac.jp}\\
	\email{titoh@c.titech.ac.jp}}
\maketitle              
\input{sections/abstract}
\input{sections/introduction}
\input{sections/preliminaries}
\input{sections/strictpopularity}
\input{sections/conclusion}
\bibliographystyle{splncs04}
\bibliography{references} 
\end{document}

%% file: packages.tex
\usepackage{graphicx}
\usepackage{amsmath}
\usepackage[capitalise]{cleveref}%
\usepackage{float}
\usepackage{amsfonts}
\usepackage{tikz}
\usetikzlibrary{automata}
\usepackage{tikzscale}
\usepackage{mathrsfs}
\usepackage{amssymb}
\usepackage{adjustbox}
\usepackage{breqn}

\usepackage[shortlabels]{enumitem}

\usepackage{flafter}
\usepackage{mathtools}
\usetikzlibrary{arrows,decorations.pathreplacing,backgrounds,calc,positioning, fit, shapes.geometric}
\allowdisplaybreaks
\crefname{case}{Case}{Cases}
\usepackage[ruled,vlined,linesnumbered]{algorithm2e}
\usepackage{comment}
\usepackage{caption}
\usepackage{subcaption}
\usepackage{multirow}

%% file: sections/abstract.tex
\begin{abstract}
	We study the 3D-Euclidean Multidimensional Stable Roommates problem, which asks whether a given set $V$ of $s\cdot n$ agents with a location in 3-dimensional Euclidean space can be partitioned into $n$ disjoint subsets $\pi = \{R_1 ,\dots , R_n\}$ with $|R_i| = s$ for each $R_i \in \pi$ such that $\pi$ is (strictly) popular, where $s$ is the room size. A partitioning is popular if there does not exist another partitioning in which more agents are better off than worse off. Computing a popular partition in a stable roommates game is NP-hard, even if the preferences are strict. The preference of an agent solely depends on the distance to its roommates. An agent prefers to be in a room where the sum of the distances to its roommates is small. We show that determining the existence of a strictly popular outcome in a 3D-Euclidean Multidimensional Stable Roommates game with room size $3$ is co-NP-hard.
	
	\keywords{Stable marriage problem \and Stable roommates problem \and Stable matching \and Popularity \and Coalition formation \and Co-NP-hard}
\end{abstract}

%% file: sections/introduction.tex
\section{Introduction}
The formation of stable coalitions in multi-agent systems is a computational problem that has several variations with different conditions on the coalition size. Arkin et al. \cite{ARKIN2009219} introduced a restricted variant of the stable roommates problem called the geometric stable roommates: each agent is associated with a point in $d$-dimensional metric space $\mathbb{R}^d$ and an agent prefers to be matched with an agent located close to it over an agent that is located far from it.


Chen and Roy \cite{chen2022multidimensional} investigated a restricted variant of the geometric stable roommates, called the Euclidean $d$-Dimensional Stable Roommates (Euclid-$d$-SR) problem, where every agent is associated with a point in 2-dimensional Euclidean space. The agents are partitioned into rooms of size $d$. Chen and Roy  \cite{chen2022multidimensional} found that determining the existence of a core stable outcome for $d \geq 3$ is NP-complete for the Euclid-$d$-SR problem.


Various notions that define the stability or optimality of a partitioning of agents exist. The notion of popularity was introduced by G{\"a}rdenfors \cite{Grdenfors1975MatchMA} in 1975. Popular matchings have been an exciting area of research \cite{Manlove2013AlgorithmicsOM}. Cseh \cite{cseh2017popular} has recently provided a survey on popular matchings. 

Computing a popular partition in a stable roommates game is NP-hard, even if the preferences are strict \cite{Cseh2021,doi:10.1137/1.9781611975482.173,10.1145/3442354}.
One method of obtaining tractability results from intractable coalition formation problems is to put restrictions on the agents’ preferences for which the associated computational problems become tractable. 
This approach has been successfully used in the stable roommates problem \cite{10.1007/978-3-540-77105-0_48,BARTHOLDI1986165,10.1007/978-3-642-13073-1_10,Bredereck2020,CHUNG2000206,10.1145/3434427}. Euclidean preferences is the same type of approach.


We study an adapted version of the Euclid-$d$-SR problem named the 3D-Euclidean Multidimensional Stable Roommates problem, where we consider 3-dimensional space instead of 2-dimensional, with the notion of popularity. We show that determining the existence of a strictly popular outcome in a 3D-Euclidean Multidimensional Stable Roommates game with room size $3$ is co-NP-hard.

Our hardness proofs are inspired by the work by Chen and Roy \cite{chen2022multidimensional} in combination with the work by Brandt and Bullinger \cite{10.1613/jair.1.13470}.

%% file: sections/preliminaries.tex
\section{Preliminaries}\label{prelim}
We use the notation and definitions related to the Euclidean $d$-Dimensional Stable Roommates problem described in the work of Chen and Roy \cite{chen2021multi}. Additionally, we also use the notation and definitions related to popularity described in the work of Brandt and Bullinger \cite{10.1613/jair.1.13470}. For our hardness results, we construct polynomial time reductions from the Planar and Cubic Exact Cover by 3-Sets (PC-X3C) problem. In order to ensure that this paper is self-contained, we shall include the relevant existing definitions in this section. 

For $m \in \mathbb{N}^+$, let $[m] = \{1, \dots, m\}$. Let $\epsilon$ denote a small fractional value with $0 < \epsilon < 0.001$ and let $d = \sqrt{ 1 - \left(\frac{\epsilon}{2}\right)^2}$.

\subsection{3D-Euclidean Multidimensional Stable Roommates}
A 3D-Euclidean Multidimensional Stable Roommates (3D-EuclidSR) game $G$ is a triple $(V, E, s)$, where the set $V = [s \cdot n]$ represents the set of agents, embedding $E: V \rightarrow \mathbb{R}^3$ denotes the position of an agent in 3-dimensional Euclidean space, and $s \in \mathbb{N}^+$ is the room size. For an agent $a \in V$, let $E(a)_x$, $E(a)_y$, and $E(a)_z$ denote the $x$, $y$, and $z$ coordinate of agent $a$ in the embedding $E$ respectively.

An $s$-sized subset of $V$ is called a room. An outcome $\pi = \left\{C_1, \dots, C_k\right\}$ of $G$ is a partitioning of the agents $V$ into $k$ rooms. Let $\pi(a)$ denote the room in $\pi$ that contains agent $a \in V$.

For $p, q \in V$, let $\delta(p,q) = \sqrt{\sum\limits_{i \in \{x,y,z\}}(E(p)_i - E(q)_i)^2}$ denote the Euclidean distance between agents $p$ and $q$. Let $a \in V$ be an agent and $R \subseteq V$ be a room such that $a \in R$, we overload the notation by defining $\delta(a, R) = \sum\limits_{a' \in R}\delta(a, a')$. 

Let $a \in V$ be an agent and $R, T \subseteq V$ be two rooms such that $a \in R$ and $a \in T$. We write $R \succsim_a T$ say that agent $a$  weakly prefers $R$ over $T$ if $a$ likes being in room $R$ at least as much as being in room $T$. Additionally, we write $R \succ_a T$ and say that agent $a$ strictly prefers room $R$ over room $T$ if $R \succsim_a T$ and $T \not\succsim_a R$. If agent $a$ weakly prefers $R$ over $T$ and $T$ over $R$, we write $R \sim_a T$ and say that $a$ is indifferent between $R$ and $T$.

The preference of an agent solely depend on its distance to its roommates. Particularly, an agent prefers being in a room with roommates that are positioned close to itself. Formally, let $a \in V$ be an agent and $R,T \subseteq V$ be two rooms such that $a \in R$ and $a \in T$. We have that $R \succsim_a T \iff \delta(a, R) \leq \delta(a, T).$

\subsection{Planar and Cubic Exact Cover by 3-Sets}
Let $X = [m]$, where $m \in \mathbb{N}^+$ and $m \mod 3 = 0$, and let $C = \left\{A_1, \dots, A_q\right\}$ be a collection of 3-element subsets of $X$. An instance of the X3C problem is a tuple $I = (X,C)$ and asks: does there exist $S \subseteq C$ such that $S$ partitions $X$? We call such an $S$ a solution of $I$.

The associated graph of $I$, denoted by $G(I)$, is a bipartite graph $(U \cup W, E)$ where $U = \{u_i | i \in X\}$, $W = \{w_j | A_j \in C\}$, and edge $\{u_i, w_j\} \in E$ if and only if $i \in A_j$. We call vertices in $U$ and $W$ the element-vertices and set-vertices, respectively. 

Instance $I$ is an instance of the PC-X3C problem, if each element $x \in X$ occurs in exactly three sets of $C$ and the associated graph is planar. Note that for a valid instance of the PC-X3C problem, we have that $|X| \geq 6$. Additionally, we have that $|X| = |C|$.

\subsection{Popularity}
Let $\pi$ and $\pi'$ be arbitrary outcomes. Let $N(\pi,\pi') = \left\{a \in N | \pi(a) \succ_a \pi'(a)\right\}$ and $\phi(\pi,\pi') = |N(\pi, \pi')| - |N(\pi',\pi)|$. We call $\phi(\pi,\pi')$ the popularity margin of $\pi$ and $\pi'$.

An outcome $\pi$ is strictly more popular than outcome $\pi'$ if $\phi(\pi,\pi') > 0$. An outcome $\pi$ is popular if for any outcome $\pi'$ we have $\phi(\pi,\pi') \geq 0$. An outcome $\pi$ is strictly popular if for any other outcome $\pi'\neq \pi$ we have $\phi(\pi,\pi')>0$.


%% file: sections/strictpopularity.tex
\section{Strict Popularity}\label{strictpop}
In this section, we show that determining the existence of a strictly popular outcome is co-NP-hard, when the room size is fixed to 3. We shall construct a 3D-EuclidSR game $G = (V,E,3)$ from an X3C instance $I = (X,C)$ such that there exists a solution $S \subseteq C$ that partitions $X$ if and only if no strictly popular outcome exists for $G$. The 3D-EuclidSR game $G$ consists of 3 parts: the bottom layer, the top layer, and the ascending layer. We shall construct separate sets of agents $V_b,V_t,V_a$, and embeddings $E^b,E^t,E^a$ for the bottom layer, the top layer, and the ascending layer, respectively. Note that $V = V_b \cup V_t \cup V_a$ and for each $a \in V$ we have $E = \begin{cases}
	E^b(a) & \text{, if } a \in V_b \\
	E^t(a) & \text{, if } a \in V_t \\
	E^a(a) & \text{, if } a \in V_a
\end{cases}$.

\subsection{Construction}
\input{sections/construction-bottomlayer}
\input{sections/construction-toplayer}
\input{sections/construction-ascendinglayer}

\subsection{Instance size}
\input{sections/construction-size}

\subsection{Permanent popular outcome.}
\input{sections/ppoutcome}

\subsection{Reduced outcome.}
\input{sections/reducedoutcome}

\subsection{Hardness}
\input{sections/strictpophardness}

%% file: sections/construction-bottomlayer.tex
\subsubsection{Bottom layer.}\label{strictpopconstr}
For the bottom layer, we use the fact that the associated graph $G(I)$ is planar and cubic. By Valiant \cite{6312176}, $G(I)$ admits a specific planar embedding in $\mathbb{Z}^2$, called orthogonal drawing, which maps each vertex $v \in U \cup W$ to an integer grid point and each edge to a chain of non-overlapping horizontal and vertical segments along the grid (except at the endpoints). A horizontal segment is a collection of consecutive horizontal edges of the grid $\mathbb{Z}^2$. A vertical segment is defined analogously.

We use the following more specific restricted orthogonal drawing $E^{G(I)}$ for $G(I)$:
\begin{proposition}\label{propgridgraph}
	(\cite{10.1137/S0097539794262847}) A planar graph with maximum vertex degree three can be embedded, in polynomial time, in the grid $\mathbb{Z}^2$ such that its vertices are at the integer grid points and its edges are drawn using at most one horizontal and one vertical segment in the grid.
\end{proposition}
The intersection point of the horizontal and vertical segments is called a bending point. 

The planar graph in \cref{propgridgraph} can be constructed without any empty vertical nor empty horizontal grid lines that intersect the graph. By empty, we mean that the grid line does not intersect any vertices. If there is an empty vertical or empty horizontal grid line, the graph can be compressed by shifting the vertices to populate the empty space. See \cref{propgridgraphcompact} for an illustration. Note that a grid line cannot only intersect with an edge of the graph, as otherwise it would have at least 2 bending points, thus have more than one horizontal or one vertical segment. We shall assume that our orthogonal graph drawing $E^{G(I)}$ does not have any empty vertical nor empty horizontal grid lines that intersect the graph.

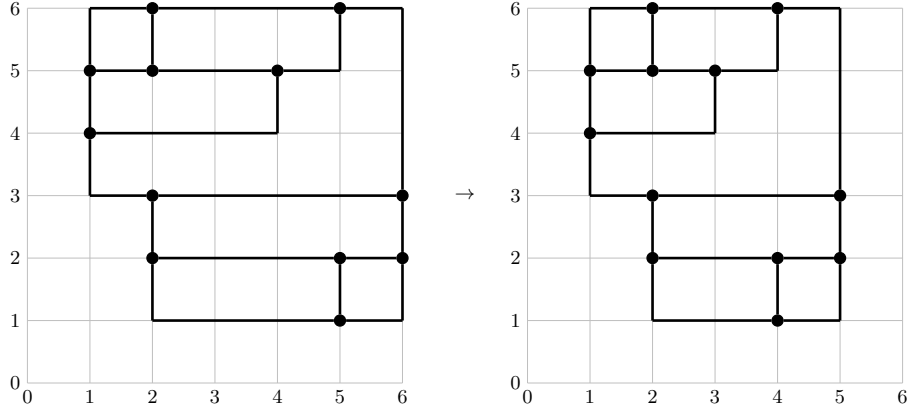
\begin{figure}
	\input{sections/figures/gridgraph0}
	\caption{Example 3-regular orthogonal graph drawing, where each edge has at most one vertical and horizontal segment. The vertical grid line $x=3$, which intersects the graph, is initially empty.}
	\label{propgridgraphcompact}
\end{figure}

In the orthogonal drawing $E^{G(I)}$, the edges of the grid $\mathbb{Z}^2$ are set to length $10$. Since in $E^{G(I)}$ there are no empty vertical nor empty horizontal grid lines that intersect the graph, the height and width of the graph in the embedding is at most $10 \cdot (|X| + |C|)$. See \cref{gridgraphfrompcx3c} for an example orthogonal drawing.

\begin{figure}
	\input{sections/figures/gridgraph1}
	\caption{Example orthogonal drawing $E^{G(I)}$, where $I = (X, C)$, $X = \{1, 2, 3, 4, 5, 6\}$, and $C = \{\{1, 2, 3\}, \{1, 2, 4\}, \{1, 3, 5\}, \{2, 4, 6\}, \{3, 5, 6\}, \{4, 5, 6\}\}$.}
	\label{gridgraphfrompcx3c}
\end{figure}
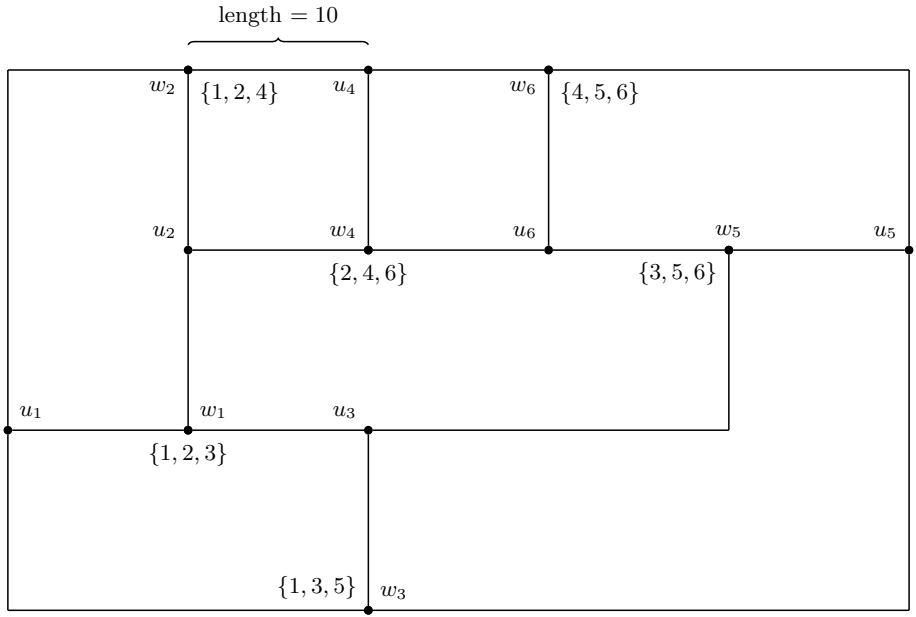

For each element-vertex $u_i \in U$, we create an agent $u_i$ for $V_b$, which we call an element-agent, with the same $xy$-coordinates as the element-vertex $u_i$ and $z$-coordinate $0$ for embedding $E^b$.

Additionally, for each edge $\{u_i, w_j\}$ in $G(I)$, if $\{u_i, w_j\}$ has a bending point in $E^{G(I)}$, we create an agent $b_j^i$ for $V_b$, which we call a bending point agent, with the same $xy-coordinates$ as the bending point and $z$-coordinate $0$ for embedding $E^b$. Let $B = \{b_j^i | \{u_i, w_j\}\text{ in }G(I) \wedge \{u_i, w_j\}\text{ has bending point in } E^{G(I)}\}$ denote the set of all bending point agents.

Finally, for each set-vertex $w_j \in W$ we create three agents $w_j^i$ for $V_b$, where $i\in C_j$, which form an equilateral triangle with edge length $1$. Specifically, for a set-vertex $w_j \in W$, we can assume w.l.o.g. that it has three connecting edges going upwards, leftward, and rightward connecting to element-vertices $u_i, u_p$, and $u_q$, respectively. If this is not the case, we can rotate the coordinate system and apply the same reasoning. We create three set-agents $w_j^i, w_j^p, w_j^q \in V_b$ to replace $w_j$. The set-agents $w_j^i, w_j^p, w_j^q$ are embedded such that they are on the segment of the upward, leftward, and rightward edge, respectively, and are of equidistance $1$ to each other. As previous agents, each of these set-agents have $z$-coordinate $0$. Let $W' = \{w_j^i, w_j^p, w_j^q | C_j = \{i,p,q\} \in C$ be the set of set-agents. See \cref{set-vertexgadget} for an illustration.

\begin{figure}
	\input{sections/figures/gridgraph2}
	\caption{Gadget for set-vertex $w_j$, where $C_j = \{i,p,q\}$ and $w_j^i, w_j^p, w_j^q$ form equilateral triangle with edge length $1$.}
	\label{set-vertexgadget}
\end{figure}
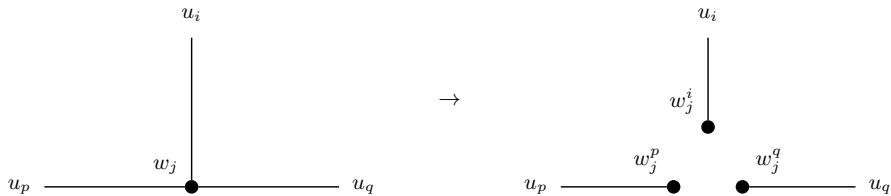

Let us modify $G(I)$ to create the graph $G'(I) = (V', E')$. Let $V'$ be the vertices consisting of all element-agents, set-agents, and bending point agents, i.e., $V' = U \cup B \cup W'$. For each edge $\{u_i, w_j\}$ in $G(I)$, if $\{u_i, w_j\}$ has a bending point in $E^{G(I)}$, then we create the edges $\{u_i, b_j^i\}, \{b_j^i, w_j^i\}$ and place them in $E'$. Otherwise we create the edge $\{u_i, w_j^i\}$ and place it in $E'$. Let us consider the embedding $E^{G'(I)}: V' \rightarrow \mathbb{Z}^2$, that maps the the vertices in $V'$ to the same $x, y$-coordinates that embedding $E^b$ assigns. That is, for each $v \in V'$ we have $E^{G'(I)}(v) = (E^b_x(v), E^b_y(v))$. See \cref{gridgraphfrompcx3c2} for an example.

\begin{figure}
	\input{sections/figures/gridgraph3}
	\caption{Example embedding $E^{G'(I)}$, where $I = (X, C)$, $X = \{1, 2, 3, 4, 5, 6\}$, and $C = \{\{1, 2, 3\}, \{1, 2, 4\}, \{1, 3, 5\}, \{2, 4, 6\}, \{3, 5, 6\}, \{4, 5, 6\}\}$. Note that the distances between $w_j^i, w_j^p, w_j^q$ are not to scale.}
	\label{gridgraphfrompcx3c2}
\end{figure}
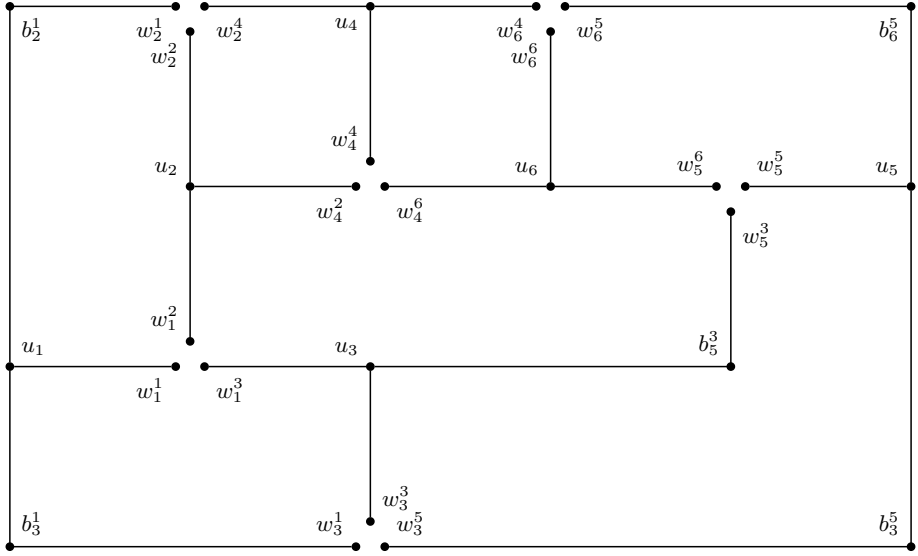

Note that for each edge in $E'$ its segment in embedding $E^{G'(I)}$ is either horizontal or vertical. We replace the segment in $E^{G'(I)}$ of each edge in $G'(I)$ with a chain of copies of three agents. This chain ensures that either all three agents $w_j^i$ for $i \in C_j$ are matched in the same triple (indicating that the corresponding set is in the solution) or none of them is matched in the same triple (indicating that the corresponding set is not in the solution).

For an edge $e \in E'$, let $s_e$ denote its corresponding segment in $E^{G'(I)}$. Let $S_{G'(I)} = \{s_e | e \in E'\}$ denote the set of all segments of the edges in $E^{G'(I)}$ and for a segment $s \in S_{G'(I)}$ let $l(s)$ denote the length of segment $s$. 

For each segment $s \in S_{G'(I)}$ we create $\hat{n} =  \left\lceil\frac{l(s)}{d}\right\rceil$ copies of the triple $A_s[z] = \{\alpha_s[z], \beta_s[z], \gamma_s[z]\}$, where $z\in \left[\hat{n}\right]$, for $V_b$. We embed the agents $A_s[z]$, where $z \in [\hat{n}]$, around segment $s$ to connect the two vertices at the endpoints of $s$ in embedding $E^{G'(I)}$. For convenience, when using $\hat{n}$, we refer to the constant corresponding to segment $s$ which should be clear from the context. Let $f,t \in V'$ denote the vertices at the endpoints of segment $s$. We merge vertex $t$ and $\gamma_s[\hat{n}]$ into one vertex, while referring to it by either name. For convenience, we shall use $\gamma_s[0]$ to refer to $f$. The agents will be embedded such that the following hold for each $z \in [\hat{n}]$: 
\begin{itemize}
	\item The distance between agents $\alpha_s[z]$ and $\beta_s[z]$ is $\epsilon$.
	\item The distance between agents $\alpha_s[z]$ (resp. $\beta_s[z]$) and $\gamma_s[z]$ is 1.
	\item The distance between agents $\alpha_s[z]$ (resp. $\beta_s[z]$) and $\gamma_s[z-1]$ is 1.
\end{itemize}
This embedding ensures that for a strictly popular outcome that either all $A_s[z]$, where $z\in [\hat{n}-1]$, or all ${\gamma_s[z-1], \alpha_s[z], \beta_s[z]}$ must be matched together.

As the length of the chain $A_s[z]$, where $z \in [\hat{n}]$, is slightly longer than $l(s)$, the chain $A_s[z]$ is `bent' into the $xy$-plane in the negative $z$-direction in embedding $E^b$. The agents in $A_s[0]$ and $A_s[\hat{n}]$ have the closest $z$-coordinate to 0. The agents in $A_s[\lfloor \hat{n}/2 \rfloor]$ and $A_s[\lceil \hat{n}/2 \rceil]$ have the furthest $z$-coordinate from 0. For each $z \in [\hat{n}]$, the vertices $\alpha_s[z]$ and $\beta_s[z]$ have the same $z$-coordinate. See \cref{gridgraphfrompcx3c3,,gridgraphfrompcx3c4,,gridgraphfrompcx3cchains} for an illustration.

\begin{figure}
	\input{sections/figures/gridgraph4}
	\caption{Projection of example chain $A_s[z]$, where $z \in [\hat{n}]$, of segment $s$, where $t, f$ are the endpoints of $s$, on the $xy$-plane.}
	\label{gridgraphfrompcx3c3}
\end{figure}
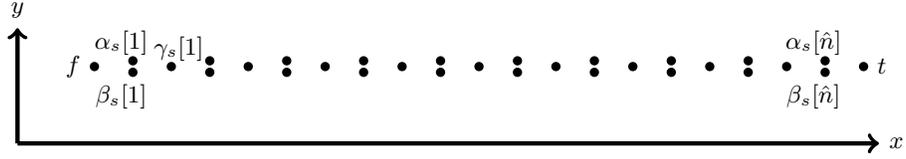
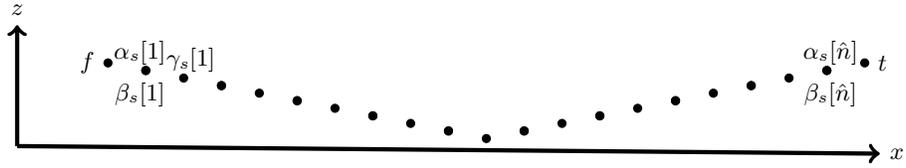
\begin{figure}
	\input{sections/figures/gridgraph5}
	\caption{Projection of example chain $A_s[z]$, where $z \in [\hat{n}]$, of segment $s$, where $t, f$ are the endpoints of $s$, on the $xz$-plane.}
	\label{gridgraphfrompcx3c4}
\end{figure}
\begin{figure}
\input{sections/figures/gridgraphchains}
\caption{Segments of embedding $E^{G'(I)}$ of \cref{gridgraphfrompcx3c2} replaced with chains. The number of agents in the chains are reduced for clarity.} 
\label{gridgraphfrompcx3cchains}
\end{figure}
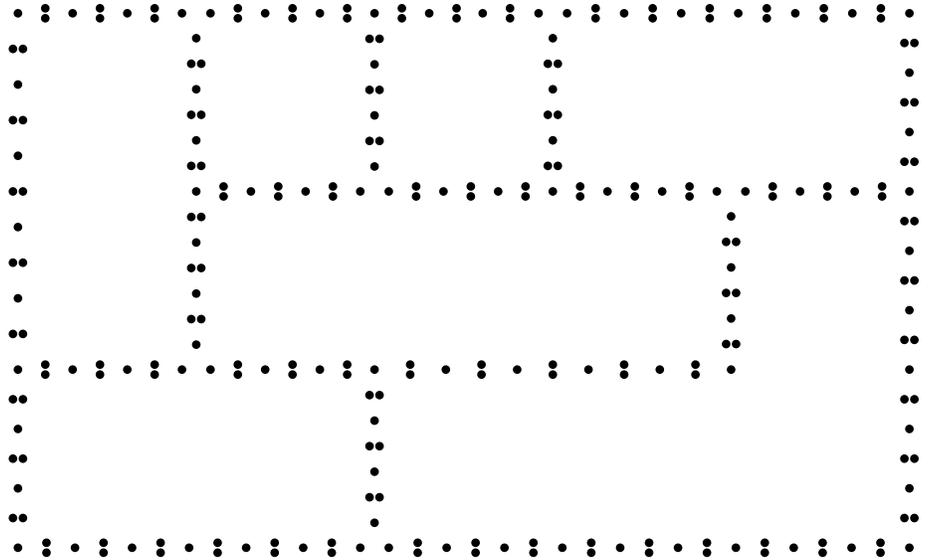
In conclusion, for the bottom layer we have that $V_b = U \cup B \cup W' \cup \\ \bigcup\limits_{s \in S_{G'(I)}, z \in [\hat{n}]}A_s[z]$. These agents in $\bigcup\limits_{s \in S_{G'(I)}, z \in [\hat{n}]}A_s[z]$ are embedded in $E^b$ as defined above and for each agent $a \in U \cup B \cup W'$, we have that $E^b(a) = (E^{G'(I)}_x, E^{G'(I)}_y, 0)$.

%% file: sections/figures/gridgraph0.tex
\newcommand*{\xMin}{0}%
\newcommand*{\xMax}{6}%
\newcommand*{\yMin}{0}%
\newcommand*{\yMax}{6}%

\resizebox{\textwidth}{!}{%
	\begin{tikzpicture}[-,auto, semithick]
		\tikzset{
			element/.style={
				fill=black,draw=black,text=black,shape=circle,inner sep=1pt,minimum size=5pt
			},
		}
	
		\foreach \i in {\xMin,...,\xMax} {
			\draw [very thin,lightgray] (\i,\yMin) -- (\i,\yMax)  node[black] [below] at (\i,\yMin) {$\i$};
		}
		\foreach \i in {\yMin,...,\yMax} {
			\draw [very thin,lightgray] (\xMin,\i) -- (\xMax,\i) node[black] [left] at (\xMin,\i) {$\i$};
		}
	
		\node[element] (1) at (2,6) {};
		\node[element] (2) at (5,6) {};
		
		\node[element] (3) at (1,5) {};
		\node[element] (4) at (2,5) {};
		\node[element] (5) at (4,5) {};
		
		\node[element] (6) at (1,4) {};
		
		\node[element] (7) at (2,3) {};
		\node[element] (8) at (6,3) {};
		
		\node[element] (9) at (2,2) {};
		\node[element] (10) at (5,2) {};
		\node[element] (11) at (6,2) {};
		
		\node[element] (12) at (5,1) {};
		
		\path (1) edge[very thick] (2);
		\path (1) edge[very thick] (1,6);
		\path (3) edge[very thick] (1,6);
		\path (1) edge[very thick] (4);
		
		\path (2) edge[very thick] (6,6);
		\path (8) edge[very thick] (6,6);
		\path (2) edge[very thick] (5,5);
		\path (5) edge[very thick] (5,5);
		
		\path (3) edge[very thick] (4);
		\path (3) edge[very thick] (6);
		
		\path (4) edge[very thick] (5);
		
		\path (5) edge[very thick] (4,4);
		\path (6) edge[very thick] (4,4);
		
		\path (6) edge[very thick] (1,3);
		\path (7) edge[very thick] (1,3);
		
		\path (7) edge[very thick] (8);
		\path (7) edge[very thick] (9);
		
		\path (8) edge[very thick] (11);
		
		\path (9) edge[very thick] (10);
		\path (9) edge[very thick] (2,1);
		\path (12) edge[very thick] (2,1);
		
		\path (10) edge[very thick] (11);
		\path (10) edge[very thick] (12);
		
		\path (11) edge[very thick] (6,1);
		\path (12) edge[very thick] (6,1);
		

		\node at (7, 3) {$\rightarrow$};

		\foreach \i in {\xMin,...,\xMax} {
			\draw [very thin,lightgray] (\i+8,\yMin) -- (\i+8,\yMax)  node[black] [below] at (\i+8,\yMin) {$\i$};
		}
		\foreach \i in {\yMin,...,\yMax} {
			\draw [very thin,lightgray] (\xMin+8,\i) -- (\xMax+8,\i) node[black] [left] at (\xMin+8,\i) {$\i$};
		}

	\node[element] (1) at (2+8,6) {};
	\node[element] (2) at (5+8-1,6) {};
	
	\node[element] (3) at (1+8,5) {};
	\node[element] (4) at (2+8,5) {};
	\node[element] (5) at (4+8-1,5) {};
	
	\node[element] (6) at (1+8,4) {};
	
	\node[element] (7) at (2+8,3) {};
	\node[element] (8) at (6+8-1,3) {};
	
	\node[element] (9) at (2+8,2) {};
	\node[element] (10) at (5+8-1,2) {};
	\node[element] (11) at (6+8-1,2) {};
	
	\node[element] (12) at (5+8-1,1) {};
	
	\path (1) edge[very thick] (2);
	\path (1) edge[very thick] (1+8,6);
	\path (3) edge[very thick] (1+8,6);
	\path (1) edge[very thick] (4);
	
	\path (2) edge[very thick] (6+8-1,6);
	\path (8) edge[very thick] (6+8-1,6);
	\path (2) edge[very thick] (5+8-1,5);
	\path (5) edge[very thick] (5+8-1,5);
	
	\path (3) edge[very thick] (4);
	\path (3) edge[very thick] (6);
	
	\path (4) edge[very thick] (5);
	
	\path (5) edge[very thick] (4+8-1,4);
	\path (6) edge[very thick] (4+8-1,4);
	
	\path (6) edge[very thick] (1+8,3);
	\path (7) edge[very thick] (1+8,3);
	
	\path (7) edge[very thick] (8);
	\path (7) edge[very thick] (9);
	
	\path (8) edge[very thick] (11);
	
	\path (9) edge[very thick] (10);
	\path (9) edge[very thick] (2+8,1);
	\path (12) edge[very thick] (2+8,1);
	
	\path (10) edge[very thick] (11);
	\path (10) edge[very thick] (12);
	
	\path (11) edge[very thick] (6+8-1,1);
	\path (12) edge[very thick] (6+8-1,1);
	\end{tikzpicture}
}

%% file: sections/figures/gridgraph1.tex
\resizebox{\textwidth}{!}{%
\begin{tikzpicture}[-,auto, semithick]
	\tikzset{
		element/.style={
			fill=black,draw=black,text=black,shape=circle,inner sep=1pt,minimum size=1pt
		},
		set/.style={
			fill=black,draw=black,text=black,shape=circle,inner sep=1pt,minimum size=1pt
		}
	}
	
	\node[element, label=above right:{$u_1$}] 		   	(1) at (0,-2.5) {};
	\node[set, label=below:{$\{1,2,3\}$}] 	(2) at (2.5,-2.5) {};
	\node[set, label=above right:{$w_1$}] 	 at (2.5,-2.5) {};
	\node[set, label=above left:{$u_3$}] 		   	(3) at (5,-2.5) {};
	\node[set, label=above left:{$\{1,3,5\}$}] 	(4) at (5,-5) {};
	\node[set, label=above right:{$w_3$}] 	 at (5,-5) {};
	\node[set, label=above left:{$u_2$}] 		   	(5) at (2.5,0) {};
	\node[set, label=below right:{$\{1,2,4\}$}]		(6) at (2.5,2.5) {};
	\node[set, label=below left:{$w_2$}]		 at (2.5,2.5) {};
	\node[set, label=below:{$\{2,4,6\}$}]	(7) at (5,0) {};
	\node[set, label=above left:{$w_4$}]	 at (5,0) {};
	\node[set, label=below left:{$u_4$}]				(8) at (5,2.5) {};
	\node[set, label=above left:{$u_6$}]			(9) at (7.5,0) {};
	\node[set, label=below right:{$\{4,5,6\}$}]		(10) at (7.5,2.5) {};
	\node[set, label=below left:{$w_6$}]		 at (7.5,2.5) {};
	\node[set, label=below left:{$\{3,5,6\}$}]	(11) at (10,0) {};
	\node[set, label=above:{$w_5$}]	 at (10,0) {};
	\node[set, label=above left:{$u_5$}]			(12) at (12.5,0) {};
	
	\path (1) edge (2);
	\path (2) edge (3);
	\path (4) edge (3);
	\path (4) edge (0,-5);
	\path (1) edge (0,-5);
	\path (5) edge (2);
	\path (5) edge (6);
	\path (1) edge (0,2.5);
	\path (6) edge (0,2.5);
	\path (5) edge (7);
	\path (7) edge (8);
	\path (6) edge (8);
	\path (7) edge (9);
	\path (10) edge (9);
	\path (10) edge (8);
	\path (11) edge (9);
	\path (11) edge (10,-2.5);
	\path (3) edge (10,-2.5);
	\path (11) edge (12);
	\path (12) edge (12.5,-5);
	\path (4) edge (12.5,-5);
	\path (12) edge (12.5,2.5);
	\path (10) edge (12.5,2.5);
	
	\draw [decorate,
	decoration = {brace}] (2.5,2.85) --  (5,2.85);
	\node[] at (3.75,3.25) {length $ = 10$};
\end{tikzpicture}
}

%% file: sections/figures/gridgraph2.tex
\resizebox{\textwidth}{!}{%
\begin{tikzpicture}[-,auto, semithick]
	\tikzset{
		element/.style={
			fill=black,draw=black,text=black,shape=circle,inner sep=1pt,minimum size=5pt
		},
		set/.style={
			fill=white,draw=white,text=black,shape=circle
		},
	}
	
	\node[element, label=above left:{$w_j$}] 	(1) at (0,0) {};
	\node[set] 	(2) at (0,2.5) {$u_i$};
	\node[set] 	(3) at (-2.5,0) {$u_p$};
	\node[set] 	(4) at (2.5,0) {$u_q$};
	
	\path (1) edge (2);
	\path (1) edge (3);
	\path (1) edge (4);
	
	\node[set]  at (3.75,1.25) {$\rightarrow$};
	
	\node[element, label=above left:{$w_j^i$}] 	(1-1) at (7.5,0.87) {};
	\node[element, label=above left:{$w_j^p$}] 	(1-2) at (7.5-0.5,0) {};
	\node[element, label=above right:{$w_j^q$}] 	(1-3) at (7.5+0.5,0) {};
	\node[set] 									(2) at (7.5,2.5) {$u_i$};
	\node[set] 									(3) at (5,0) {$u_p$};
	\node[set] 									(4) at (10,0) {$u_q$};
	
	\path (1-1) edge (2);
	\path (1-2) edge (3);
	\path (1-3) edge (4);
\end{tikzpicture}
}

%% file: sections/figures/gridgraph3.tex
\resizebox{\textwidth}{!}{%
	\begin{tikzpicture}[-,auto, semithick]
		\tikzset{
			element/.style={
				fill=black,draw=black,text=black,shape=circle,inner sep=1pt,minimum size=1pt
			},
			set/.style={
				fill=black,draw=black,text=black,shape=circle,inner sep=1pt,minimum size=1pt
			}
		}
		
		\node[element, label=above right:{$u_1$}] 	(1) at (0,-2.5) {};
		
		\node[set, label=below left:{$w_1^1$}] 		(2-1) at (2.5-0.2,-2.5) {};
		\node[set, label=below right:{$w_1^3$}] 	(2-2) at (2.5+0.2,-2.5) {};
		\node[set, label=above left:{$w_1^2$}] 		(2-3) at (2.5,-2.5+0.35){};
		
		\node[set, label=above left:{$u_3$}] 		(3) at (5,-2.5) {};
		
		\node[set, label=above left:{$w_3^1$}] 	 (4-1) at (5-0.2,-5) {};
		\node[set, label=above right:{$w_3^5$}] 	 (4-2) at (5+0.2,-5) {};
		\node[set, label=above right:{$w_3^3$}] 	 (4-3) at (5,-5+0.35) {};
		
		\node[set, label=above left:{$u_2$}] 		(5) at (2.5,0) {};
		
		\node[set, label=below left:{$w_2^1$}]		(6-1) at (2.5-0.2,2.5) {};
		\node[set, label=below right:{$w_2^4$}]		(6-2) at (2.5+0.2,2.5) {};
		\node[set, label=below left:{$w_2^2$}]		(6-3) at (2.5,2.5-0.35) {};
		
		\node[set, label=below left:{$w_4^2$}]		(7-1) at (5-0.2,0) {};
		\node[set, label=below right:{$w_4^6$}]		(7-2) at (5+0.2,0) {};
		\node[set, label=above left:{$w_4^4$}]		(7-3) at (5,0+0.35) {};
		
		\node[set, label=below left:{$u_4$}]		(8) at (5,2.5) {};
		
		\node[set, label=above left:{$u_6$}]		(9) at (7.5,0) {};
		
		\node[set, label=below left:{$w_6^4$}]		(10-1) at (7.5-0.2,2.5) {};
		\node[set, label=below right:{$w_6^5$}]		(10-2) at (7.5+0.2,2.5) {};
		\node[set, label=below left:{$w_6^6$}]		(10-3) at (7.5,2.5-0.35){};
		
		\node[set, label=above left:{$w_5^6$}]		(11-1) at (10-0.2,0) {};
		\node[set, label=above right:{$w_5^5$}]		(11-2) at (10+0.2,0) {};
		\node[set, label=below right:{$w_5^3$}]		(11-3) at (10,-0.35) {};
		
		\node[set, label=above left:{$u_5$}]		(12) at (12.5,0) {};

		\node[set, label=below right:{$b_2^1$}]		(13) at (0,2.5) {};
		\node[set, label=above right:{$b_3^1$}]		(14) at (0,-5) {};
		\node[set, label=above left:{$b_5^3$}]		(15) at (10,-2.5) {};
		\node[set, label=above left:{$b_3^5$}]		(16) at (12.5,-5) {};
		\node[set, label=below left:{$b_6^5$}]		(17) at (12.5,2.5) {};
		
		\path (1) edge (2-1);
		\path (2-2) edge (3);
		\path (4-3) edge (3);
		\path (4-1) edge (0,-5);
		\path (1) edge (0,-5);
		\path (5) edge (2-3);
		\path (5) edge (6-3);
		\path (1) edge (0,2.5);
		\path (6-1) edge (0,2.5);
		\path (5) edge (7-1);
		\path (7-3) edge (8);
		\path (6-2) edge (8);
		\path (7-2) edge (9);
		\path (10-3) edge (9);
		\path (10-1) edge (8);
		\path (11-1) edge (9);
		\path (11-3) edge (10,-2.5);
		\path (3) edge (10,-2.5);
		\path (11-2) edge (12);
		\path (12) edge (12.5,-5);
		\path (4-2) edge (12.5,-5);
		\path (12) edge (12.5,2.5);
		\path (10-2) edge (12.5,2.5);
	\end{tikzpicture}
}

%% file: sections/figures/gridgraph4.tex
\resizebox{\textwidth}{!}{%
	\begin{tikzpicture}[-,auto, semithick]
		\tikzset{
			element/.style={
				fill=black,draw=black,text=black,shape=circle,inner sep=1pt,minimum size=1pt
			},
			set/.style={
				fill=white,draw=white,text=black,shape=circle
			},
		}
		
		\node[element, label= left:{$f$}] 	(1) at (0,0) {};
		\node[element, label= right:{$t$}] 	(2) at (10,0) {};

		\node[element, label={[shift={(-0.15,-0.1)}]$\alpha_s[1]$}] (3-1) at (0.5,0.075){};
		\node[element, label={[shift={(-0.15,-0.65)}]$\beta_s[1]$}] (3-1) at (0.5,-0.075) {};
		\node[element, label={[shift={(0.1,-0.1)}]$\gamma_s[1]$}] 	(3-1) at (1,0) {};
		
		\node[element] 	(3-1) at (1.5,0.075) {};
		\node[element] 	(3-1) at (1.5,-0.075) {};
		\node[element] 	(3-1) at (2,0) {};
		
		\node[element] 	(3-1) at (2.5,0.075) {};
		\node[element] 	(3-1) at (2.5,-0.075) {};
		\node[element] 	(3-1) at (3,0) {};
		
		\node[element] 	(3-1) at (3.5,0.075) {};
		\node[element] 	(3-1) at (3.5,-0.075) {};
		\node[element] 	(3-1) at (4,0) {};
		
		\node[element] 	(3-1) at (4.5,0.075) {};
		\node[element] 	(3-1) at (4.5,-0.075) {};
		\node[element] 	(3-1) at (5,0) {};
		
		\node[element] 	(3-1) at (5.5,0.075) {};
		\node[element] 	(3-1) at (5.5,-0.075) {};
		\node[element] 	(3-1) at (6,0) {};
		
		\node[element] 	(3-1) at (6.5,0.075) {};
		\node[element] 	(3-1) at (6.5,-0.075) {};
		\node[element] 	(3-1) at (7,0) {};
		
		\node[element] 	(3-1) at (7.5,0.075) {};
		\node[element] 	(3-1) at (7.5,-0.075) {};
		\node[element] 	(3-1) at (8,0) {};
		
		\node[element] 	(3-1) at (8.5,0.075) {};
		\node[element] 	(3-1) at (8.5,-0.075) {};
		\node[element] 	(3-1) at (9,0) {};
		
		\node[element, label={[shift={(-0.15,-0.1)}]$\alpha_s[\hat{n}]$}] 	(3-1) at (9.5,0.075) {};
		\node[element, label={[shift={(-0.15,-0.65)}]$\beta_s[\hat{n}]$}] 	(3-1) at (9.5,-0.075) {};
		
		\draw[->,ultra thick] (-1,-1)--(10.2,-1) node[right]{$x$};
		\draw[->,ultra thick] (-1,-1)--(-1,0.5) node[above]{$y$};
	\end{tikzpicture}
}

%% file: sections/figures/gridgraph5.tex
\resizebox{\textwidth}{!}{%
	\begin{tikzpicture}[-,auto, semithick]
		\tikzset{
			element/.style={
				fill=black,draw=black,text=black,shape=circle,inner sep=1pt,minimum size=1pt
			},
			set/.style={
				fill=white,draw=white,text=black,shape=circle
			},
		}
		
		\node[element, label= left:{$f$}] 	(1) at (0,0) {};
		\node[element, label= right:{$t$}] 	(2) at (10,0) {};

		\node[element, label={[shift={(-0.075,-0.1)}]$\alpha_s[1]$}] (3-1) at (0.5,-0.1){};
		\node[element, label={[shift={(-0.075,-0.65)}]$\beta_s[1]$}] (3-1) at (0.5,-0.1) {};
		\node[element, label={[shift={(0.1,-0.1)}]$\gamma_s[1]$}] 	(3-1) at (1,-0.2) {};
		
		\node[element] 	(3-1) at (1.5,-0.3) {};
		\node[element] 	(3-1) at (1.5,-0.3) {};
		\node[element] 	(3-1) at (2,-0.4) {};
		
		\node[element] 	(3-1) at (2.5,-0.5) {};
		\node[element] 	(3-1) at (2.5,-0.5) {};
		\node[element] 	(3-1) at (3,-0.6) {};
		
		\node[element] 	(3-1) at (3.5,-0.7) {};
		\node[element] 	(3-1) at (3.5,-0.7) {};
		\node[element] 	(3-1) at (4,-0.8) {};
		
		\node[element] 	(3-1) at (4.5,-0.9) {};
		\node[element] 	(3-1) at (4.5,-0.9) {};
		\node[element] 	(3-1) at (5,-1) {};
		
		\node[element] 	(3-1) at (5.5,-0.9) {};
		\node[element] 	(3-1) at (5.5,-0.9) {};
		\node[element] 	(3-1) at (6,-0.8) {};
		
		\node[element] 	(3-1) at (6.5,-0.7) {};
		\node[element] 	(3-1) at (6.5,-0.7) {};
		\node[element] 	(3-1) at (7,-0.6) {};
		
		\node[element] 	(3-1) at (7.5,-0.5) {};
		\node[element] 	(3-1) at (7.5,-0.5) {};
		\node[element] 	(3-1) at (8,-0.4) {};
		
		\node[element] 	(3-1) at (8.5,-0.3) {};
		\node[element] 	(3-1) at (8.5,-0.3) {};
		\node[element] 	(3-1) at (9,-0.2) {};
		
		\node[element, label={[shift={(0.05,-0.1)}]$\alpha_s[\hat{n}]$}] 	(3-1) at (9.5,-0.1) {};
		\node[element, label={[shift={(0.05,-0.65)}]$\beta_s[\hat{n}]$}] 	(3-1) at (9.5,-0.1) {};

		\draw[->,ultra thick] (-1.2,-1.1)--(10.2,-1.2) node[right]{$x$};
		\draw[->,ultra thick] (-1.2,-1.1)--(-1.2,0.5) node[above]{$z$};
	\end{tikzpicture}
}

%% file: sections/figures/gridgraphchains.tex
\resizebox{\textwidth}{!}{%
	\begin{tikzpicture}[-,auto, semithick]
		\tikzset{
			element/.style={
				fill=black,draw=black,text=black,shape=circle,inner sep=1pt,minimum size=1pt
			}, set/.style={
				fill=black,draw=black,text=black,shape=circle,inner sep=1pt,minimum size=1pt
			}, expand bubble/.style={
				preaction={draw,line width=0.6em},
				white,fill,draw,line width=0.5em,
			},
		}
		\node[element] 	(1) at (0,-2.5) {};
		
		\node[element] 	(1-2-1) at (2.3/6,-2.5+0.07) {};
		\node[element] 	(1-2-2) at (2.3/6,-2.5-0.07) {};
		\node[element] 	(1-2-3) at (4.6/6,-2.5) {};
		\node[element] 	(1-2-4) at (6.9/6,-2.5+0.07) {};
		\node[element] 	(1-2-5) at (6.9/6,-2.5-0.07) {};
		\node[element] 	(1-2-6) at (9.2/6,-2.5) {};
		\node[element] 	(1-2-7) at (11.5/6,-2.5+0.07) {};
		\node[element] 	(1-2-8) at (11.5/6,-2.5-0.07) {};
		
		\node[element] 	(1-14-1) at (0+0.07,-2.5-2.5/6) {};
		\node[element] 	(1-14-2) at (0-0.07,-2.5-2.5/6) {};
		\node[element] 	(1-14-3) at (0,-2.5-5/6) {};
		\node[element] 	(1-14-4) at (0+0.07,-2.5-7.5/6) {};
		\node[element] 	(1-14-5) at (0-0.07,-2.5-7.5/6) {};
		\node[element] 	(1-14-6) at (0,-2.5-10/6) {};
		\node[element] 	(1-14-7) at (0+0.07,-2.5-12.5/6) {};
		\node[element] 	(1-14-8) at (0-0.07,-2.5-12.5/6) {};
		
		\node[element] 	(1-13-1) at (0+0.07,-2.5+5/10) {};
		\node[element] 	(1-13-2) at (0-0.07,-2.5+5/10) {};
		\node[element] 	(1-13-3) at (0,-2.5+10/10) {};
		\node[element] 	(1-13-4) at (0+0.07,-2.5+15/10) {};
		\node[element] 	(1-13-5) at (0-0.07,-2.5+15/10) {};
		\node[element] 	(1-13-6) at (0,-2.5+20/10) {};
		\node[element] 	(1-13-7) at (0+0.07,-2.5+25/10) {};
		\node[element] 	(1-13-8) at (0-0.07,-2.5+25/10) {};
		\node[element] 	(1-13-9) at (0,-2.5+30/10) {};
		\node[element] 	(1-13-10) at (0+0.07,-2.5+35/10) {};
		\node[element] 	(1-13-11) at (0-0.07,-2.5+35/10) {};
		\node[element] 	(1-13-12) at (0,-2.5+40/10) {};
		\node[element] 	(1-13-13) at (0+0.07,-2.5+45/10) {};
		\node[element] 	(1-13-14) at (0-0.07,-2.5+45/10) {};
		
		\node[set] 		(2-1) at (2.5-0.2,-2.5) {};
		\node[set] 		(2-2) at (2.5+0.2,-2.5) {};
		\node[set] 		(2-3) at (2.5,-2.5+0.35){};
		
		\node[element] 		(3) at (5,-2.5) {};
		
		\node[element] 		(3-2-1) at (5-2.3/6,-2.5+0.07) {};
		\node[element] 		(3-2-2) at (5-2.3/6,-2.5-0.07) {};
		\node[element] 		(3-2-3) at (5-4.6/6,-2.5) {};
		\node[element] 		(3-2-4) at (5-6.9/6,-2.5+0.07) {};
		\node[element] 		(3-2-5) at (5-6.9/6,-2.5-0.07) {};
		\node[element] 		(3-2-6) at (5-9.2/6,-2.5) {};
		\node[element] 		(3-2-7) at (5-11.5/6,-2.5+0.07) {};
		\node[element] 		(3-2-8) at (5-11.5/6,-2.5-0.07) {};
		
		\node[element] 		(3-4-1) at (5+0.07,-2.5-2.15/6) {};
		\node[element] 		(3-4-2) at (5-0.07,-2.5-2.15/6) {};
		\node[element] 		(3-4-3) at (5,-2.5-4.3/6) {};
		\node[element] 		(3-4-4) at (5+0.07,-2.5-6.45/6) {};
		\node[element] 		(3-4-5) at (5-0.07,-2.5-6.45/6) {};
		\node[element] 		(3-4-6) at (5,-2.5-8.6/6) {};
		\node[element] 		(3-4-7) at (5+0.07,-2.5-10.75/6) {};
		\node[element] 		(3-4-8) at (5-0.07,-2.5-10.75/6) {};
		
		\node[element] 		(3-15-1) at (5+0.5,-2.5+0.07) {};
		\node[element] 		(3-15-2) at (5+0.5,-2.5-0.07) {};
		\node[element] 		(3-15-3) at (5+1,-2.5) {};
		\node[element] 		(3-15-4) at (5+1.5,-2.5+0.07) {};
		\node[element] 		(3-15-5) at (5+1.5,-2.5-0.07) {};
		\node[element] 		(3-15-6) at (5+2,-2.5) {};
		\node[element] 		(3-15-7) at (5+2.5,-2.5+0.07) {};
		\node[element] 		(3-15-8) at (5+2.5,-2.5-0.07) {};
		\node[element] 		(3-15-9) at (5+3,-2.5) {};
		\node[element] 		(3-15-10) at (5+3.5,-2.5+0.07) {};
		\node[element] 		(3-15-11) at (5+3.5,-2.5-0.07) {};
		\node[element] 		(3-15-12) at (5+4,-2.5) {};
		\node[element] 		(3-15-13) at (5+4.5,-2.5+0.07) {};
		\node[element] 		(3-15-14) at (5+4.5,-2.5-0.07) {};
		
		\node[set] 	 (4-1) at (5-0.2,-5) {};
		\node[set] 	 (4-2) at (5+0.2,-5) {};
		\node[set] 	 (4-3) at (5,-5+0.35) {};
		
		\node[element] 		(5) at (2.5,0) {};
		
		\node[element] 		(5-2-1) at (2.5+0.07,-2.15/6) {};
		\node[element] 		(5-2-2) at (2.5-0.07,-2.15/6) {};
		\node[element] 		(5-2-3) at (2.5,-4.3/6) {};
		\node[element] 		(5-2-4) at (2.5+0.07,-6.45/6) {};
		\node[element] 		(5-2-5) at (2.5-0.07,-6.45/6) {};
		\node[element] 		(5-2-6) at (2.5,-8.6/6) {};
		\node[element] 		(5-2-7) at (2.5+0.07,-10.75/6) {};
		\node[element] 		(5-2-8) at (2.5-0.07,-10.75/6) {};
		
		\node[element] 		(5-6-1) at (2.5+0.07,2.15/6) {};
		\node[element] 		(5-6-2) at (2.5-0.07,2.15/6) {};
		\node[element] 		(5-6-3) at (2.5,4.3/6) {};
		\node[element] 		(5-6-4) at (2.5+0.07,6.45/6) {};
		\node[element] 		(5-6-5) at (2.5-0.07,6.45/6) {};
		\node[element] 		(5-6-6) at (2.5,8.6/6) {};
		\node[element] 		(5-6-7) at (2.5+0.07,10.75/6) {};
		\node[element] 		(5-6-8) at (2.5-0.07,10.75/6) {};
		
		\node[element] 		(5-7-1) at (2.5+2.3/6,0.07) {};
		\node[element] 		(5-7-2) at (2.5+2.3/6,-0.07) {};
		\node[element] 		(5-7-3) at (2.5+4.6/6,0) {};
		\node[element] 		(5-7-4) at (2.5+6.9/6,0.07) {};
		\node[element] 		(5-7-5) at (2.5+6.9/6,-0.07) {};
		\node[element] 		(5-7-6) at (2.5+9.2/6,0) {};
		\node[element] 		(5-7-7) at (2.5+11.5/6,0.07) {};
		\node[element] 		(5-7-8) at (2.5+11.5/6,-0.07) {};
		
		\node[set]		(6-1) at (2.5-0.2,2.5) {};
		\node[set]		(6-2) at (2.5+0.2,2.5) {};
		\node[set]		(6-3) at (2.5,2.5-0.35) {};
		
		\node[set]		(7-1) at (5-0.2,0) {};
		\node[set]		(7-2) at (5+0.2,0) {};
		\node[set]		(7-3) at (5,0+0.35) {};
		
		\node[element]		(8) at (5,2.5) {};
		
		\node[element]		(8-7-1) at (5+0.07,2.5-2.15/6) {};
		\node[element]		(8-7-2) at (5-0.07,2.5-2.15/6) {};
		\node[element]		(8-7-3) at (5,2.5-4.3/6) {};
		\node[element]		(8-7-4) at (5+0.07,2.5-6.45/6) {};
		\node[element]		(8-7-5) at (5-0.07,2.5-6.45/6) {};
		\node[element]		(8-7-6) at (5,2.5-8.6/6) {};
		\node[element]		(8-7-7) at (5+0.07,2.5-10.75/6) {};
		\node[element]		(8-7-8) at (5-0.07,2.5-10.75/6) {};
		
		\node[element]		(8-6-1) at (5-2.3/6,2.5+0.07) {};
		\node[element]		(8-6-2) at (5-2.3/6,2.5-0.07) {};
		\node[element]		(8-6-3) at (5-4.6/6,2.5) {};
		\node[element]		(8-6-4) at (5-6.9/6,2.5+0.07) {};
		\node[element]		(8-6-5) at (5-6.9/6,2.5-0.07) {};
		\node[element]		(8-6-6) at (5-9.2/6,2.5) {};
		\node[element]		(8-6-7) at (5-11.5/6,2.5+0.07) {};
		\node[element]		(8-6-8) at (5-11.5/6,2.5-0.07) {};
		
		\node[element]		(8-10-1) at (5+2.3/6,2.5+0.07) {};
		\node[element]		(8-10-2) at (5+2.3/6,2.5-0.07) {};
		\node[element]		(8-10-3) at (5+4.6/6,2.5) {};
		\node[element]		(8-10-4) at (5+6.9/6,2.5+0.07) {};
		\node[element]		(8-10-5) at (5+6.9/6,2.5-0.07) {};
		\node[element]		(8-10-6) at (5+9.1/6,2.5) {};
		\node[element]		(8-10-7) at (5+11.4/6,2.5+0.07) {};
		\node[element]		(8-10-8) at (5+11.4/6,2.5-0.07) {};
		
		\node[element]		(9) at (7.5,0) {};
		
		\node[element]		(9-7-1) at (7.5-2.3/6,0.07) {};
		\node[element]		(9-7-2) at (7.5-2.3/6,-0.07) {};
		\node[element]		(9-7-3) at (7.5-4.6/6,0) {};
		\node[element]		(9-7-4) at (7.5-6.9/6,0.07) {};
		\node[element]		(9-7-5) at (7.5-6.9/6,-0.07) {};
		\node[element]		(9-7-6) at (7.5-9.2/6,0) {};
		\node[element]		(9-7-7) at (7.5-11.5/6,0.07) {};
		\node[element]		(9-7-8) at (7.5-11.5/6,-0.07) {};
		
		\node[element]		(9-10-1) at (7.5+0.07,2.15/6) {};
		\node[element]		(9-10-2) at (7.5-0.07,2.15/6) {};
		\node[element]		(9-10-3) at (7.5,4.3/6) {};
		\node[element]		(9-10-4) at (7.5+0.07,6.45/6) {};
		\node[element]		(9-10-5) at (7.5-0.07,6.45/6) {};
		\node[element]		(9-10-6) at (7.5,8.6/6) {};
		\node[element]		(9-10-7) at (7.5+0.07,10.75/6) {};
		\node[element]		(9-10-8) at (7.5-0.07,10.75/6) {};
		
		\node[element]		(9-11-1) at (7.5+2.3/6,0.07) {};
		\node[element]		(9-11-2) at (7.5+2.3/6,-0.07) {};
		\node[element]		(9-11-3) at (7.5+4.6/6,0) {};
		\node[element]		(9-11-4) at (7.5+6.9/6,0.07) {};
		\node[element]		(9-11-5) at (7.5+6.9/6,-0.07) {};
		\node[element]		(9-11-6) at (7.5+9.2/6,0) {};
		\node[element]		(9-11-7) at (7.5+11.5/6,0.07) {};
		\node[element]		(9-11-8) at (7.5+11.5/6,-0.07) {};
		
		\node[set]		(10-1) at (7.5-0.2,2.5) {};
		\node[set]		(10-2) at (7.5+0.2,2.5) {};
		\node[set]		(10-3) at (7.5,2.5-0.35){};
		
		\node[set]		(11-1) at (10-0.2,0) {};
		\node[set]		(11-2) at (10+0.2,0) {};
		\node[set]		(11-3) at (10,-0.35) {};
		
		\node[element]		(12) at (12.5,0) {};
		
		\node[element]		(12-11-1) at (12.5-2.3/6,0.07) {};
		\node[element]		(12-11-2) at (12.5-2.3/6,-0.07) {};
		\node[element]		(12-11-3) at (12.5-4.6/6,0) {};
		\node[element]		(12-11-4) at (12.5-6.9/6,0.07) {};
		\node[element]		(12-11-5) at (12.5-6.9/6,-0.07) {};
		\node[element]		(12-11-6) at (12.5-9.2/6,0) {};
		\node[element]		(12-11-7) at (12.5-11.5/6,0.07) {};
		\node[element]		(12-11-8) at (12.5-11.5/6,-0.07) {};
		
		\node[element]		(12-16-1) at (12.5+0.07,-5/12) {};
		\node[element]		(12-16-2) at (12.5-0.07,-5/12) {};
		\node[element]		(12-16-3) at (12.5,-5*2/12) {};
		\node[element]		(12-16-4) at (12.5+0.07,-5*3/12) {};
		\node[element]		(12-16-5) at (12.5-0.07,-5*3/12) {};
		\node[element]		(12-16-6) at (12.5,-5*4/12) {};
		\node[element]		(12-16-7) at (12.5+0.07,-5*5/12) {};
		\node[element]		(12-16-8) at (12.5-0.07,-5*5/12) {};
		\node[element]		(12-16-10) at (12.5,-5*6/12) {};
		\node[element]		(12-16-11) at (12.5+0.07,-5*7/12) {};
		\node[element]		(12-16-12) at (12.5-0.07,-5*7/12) {};
		\node[element]		(12-16-13) at (12.5,-5*8/12) {};
		\node[element]		(12-16-14) at (12.5+0.07,-5*9/12) {};
		\node[element]		(12-16-15) at (12.5-0.07,-5*9/12) {};
		\node[element]		(12-16-16) at (12.5,-5*10/12) {};
		\node[element]		(12-16-17) at (12.5+0.07,-5*11/12) {};
		\node[element]		(12-16-18) at (12.5-0.07,-5*11/12) {};
		
		\node[element]		(12-17-1) at (12.5+0.07,2.5/6) {};
		\node[element]		(12-17-2) at (12.5-0.07,2.5/6) {};
		\node[element]		(12-17-3) at (12.5,2.5*2/6) {};
		\node[element]		(12-17-4) at (12.5+0.07,2.5*3/6) {};
		\node[element]		(12-17-5) at (12.5-0.07,2.5*3/6) {};
		\node[element]		(12-17-6) at (12.5,2.5*4/6) {};
		\node[element]		(12-17-7) at (12.5+0.07,2.5*5/6) {};
		\node[element]		(12-17-8) at (12.5-0.07,2.5*5/6) {};
		
		\node[element]		(13) at (0,2.5) {};
		
		\node[element]		(13-6-1) at (2.3/6,2.5+0.07) {};
		\node[element]		(13-6-2) at (2.3/6,2.5-0.07) {};
		\node[element]		(13-6-3) at (2.3*2/6,2.5) {};
		\node[element]		(13-6-4) at (2.3*3/6,2.5+0.07) {};
		\node[element]		(13-6-5) at (2.3*3/6,2.5-0.07) {};
		\node[element]		(13-6-6) at (2.3*4/6,2.5) {};
		\node[element]		(13-6-7) at (2.3*5/6,2.5+0.07) {};
		\node[element]		(13-6-8) at (2.3*5/6,2.5-0.07) {};
		
		\node[element]		(14) at (0,-5) {};
		
		\node[element]		(14-4-1) at (4.8/12,-5+0.07) {};
		\node[element]		(14-4-1) at (4.8/12,-5-0.07) {};
		\node[element]		(14-4-1) at (4.8*2/12,-5) {};
		\node[element]		(14-4-1) at (4.8*3/12,-5+0.07) {};
		\node[element]		(14-4-1) at (4.8*3/12,-5-0.07) {};
		\node[element]		(14-4-1) at (4.8*4/12,-5) {};
		\node[element]		(14-4-1) at (4.8*5/12,-5+0.07) {};
		\node[element]		(14-4-1) at (4.8*5/12,-5-0.07) {};
		\node[element]		(14-4-1) at (4.8*6/12,-5) {};
		\node[element]		(14-4-1) at (4.8*7/12,-5+0.07) {};
		\node[element]		(14-4-1) at (4.8*7/12,-5-0.07) {};
		\node[element]		(14-4-1) at (4.8*8/12,-5) {};
		\node[element]		(14-4-1) at (4.8*9/12,-5+0.07) {};
		\node[element]		(14-4-1) at (4.8*9/12,-5-0.07) {};
		\node[element]		(14-4-1) at (4.8*10/12,-5) {};
		\node[element]		(14-4-1) at (4.8*11/12,-5+0.07) {};
		\node[element]		(14-4-1) at (4.8*11/12,-5-0.07) {};
		
		\node[element]		(15) at (10,-2.5) {};
		
		\node[element]		(15-11-1) at (10.07,-2.5+2.15/6) {};
		\node[element]		(15-11-2) at (10-0.07,-2.5+2.15/6) {};
		\node[element]		(15-11-3) at (10,-2.5+2.15*2/6) {};
		\node[element]		(15-11-4) at (10.07,-2.5+2.15*3/6) {};
		\node[element]		(15-11-5) at (10-0.07,-2.5+2.15*3/6) {};
		\node[element]		(15-11-6) at (10,-2.5+2.15*4/6) {};
		\node[element]		(15-11-7) at (10.07,-2.5+2.15*5/6) {};
		\node[element]		(15-11-8) at (10-0.07,-2.5+2.15*5/6) {};
		
		\node[element]		(16) at (12.5,-5) {};
		
		\node[element]		(16-4-1) at (12.5-7.3/18,-5+0.07) {};
		\node[element]		(16-4-2) at (12.5-7.3/18,-5-0.07) {};
		\node[element]		(16-4-3) at (12.5-7.3*2/18,-5) {};
		\node[element]		(16-4-4) at (12.5-7.3*3/18,-5+0.07) {};
		\node[element]		(16-4-5) at (12.5-7.3*3/18,-5-0.07) {};
		\node[element]		(16-4-6) at (12.5-7.3*4/18,-5) {};
		\node[element]		(16-4-7) at (12.5-7.3*5/18,-5+0.07) {};
		\node[element]		(16-4-8) at (12.5-7.3*5/18,-5-0.07) {};
		\node[element]		(16-4-9) at (12.5-7.3*6/18,-5) {};
		\node[element]		(16-4-10) at (12.5-7.3*7/18,-5+0.07) {};
		\node[element]		(16-4-11) at (12.5-7.3*7/18,-5-0.07) {};
		\node[element]		(16-4-12) at (12.5-7.3*8/18,-5) {};
		\node[element]		(16-4-13) at (12.5-7.3*9/18,-5+0.07) {};
		\node[element]		(16-4-14) at (12.5-7.3*9/18,-5-0.07) {};
		\node[element]		(16-4-15) at (12.5-7.3*10/18,-5) {};
		\node[element]		(16-4-16) at (12.5-7.3*11/18,-5+0.07) {};
		\node[element]		(16-4-17) at (12.5-7.3*11/18,-5-0.07) {};
		\node[element]		(16-4-18) at (12.5-7.3*12/18,-5) {};
		\node[element]		(16-4-19) at (12.5-7.3*13/18,-5+0.07) {};
		\node[element]		(16-4-20) at (12.5-7.3*13/18,-5-0.07) {};
		\node[element]		(16-4-21) at (12.5-7.3*14/18,-5) {};
		\node[element]		(16-4-22) at (12.5-7.3*15/18,-5+0.07) {};
		\node[element]		(16-4-23) at (12.5-7.3*15/18,-5-0.07) {};
		\node[element]		(16-4-24) at (12.5-7.3*16/18,-5) {};
		\node[element]		(16-4-25) at (12.5-7.3*17/18,-5+0.07) {};
		\node[element]		(16-4-26) at (12.5-7.3*17/18,-5-0.07) {};
		
		\node[element]		(17) at (12.5,2.5) {};
		
		\node[element]		(17-10-1) at (12.5-4.8/12,2.5+0.07) {};
		\node[element]		(17-10-2) at (12.5-4.8/12,2.5-0.07) {};
		\node[element]		(17-10-3) at (12.5-4.8*2/12,2.5) {};
		\node[element]		(17-10-4) at (12.5-4.8*3/12,2.5+0.07) {};
		\node[element]		(17-10-5) at (12.5-4.8*3/12,2.5-0.07) {};
		\node[element]		(17-10-6) at (12.5-4.8*4/12,2.5) {};
		\node[element]		(17-10-7) at (12.5-4.8*5/12,2.5+0.07) {};
		\node[element]		(17-10-8) at (12.5-4.8*5/12,2.5-0.07) {};
		\node[element]		(17-10-9) at (12.5-4.8*6/12,2.5) {};
		\node[element]		(17-10-10) at (12.5-4.8*7/12,2.5+0.07) {};
		\node[element]		(17-10-11) at (12.5-4.8*7/12,2.5-0.07) {};
		\node[element]		(17-10-12) at (12.5-4.8*8/12,2.5) {};
		\node[element]		(17-10-13) at (12.5-4.8*9/12,2.5+0.07) {};
		\node[element]		(17-10-14) at (12.5-4.8*9/12,2.5-0.07) {};
		\node[element]		(17-10-15) at (12.5-4.8*10/12,2.5) {};
		\node[element]		(17-10-16) at (12.5-4.8*11/12,2.5+0.07) {};
		\node[element]		(17-10-17) at (12.5-4.8*11/12,2.5-0.07) {};
		

	\end{tikzpicture}
}

%% file: sections/construction-toplayer.tex
\subsubsection{Top layer.}
In top layer, the embedding of the agents is shaped somewhat similarly to a snowflake. It consists of a triple of agents $d_{1,0}^1, d_{2,0}^1, d_{3,0}^1$, which we call center agents, placed in the center in embedding $E^t$ that form an equilateral triangle with edge length 1. Out each of these center agents, a binary tree-like structure grows. In this structure, the parents and their two children also form a equilateral triangle. The edge length of the equilateral triangle at depth 0 is large, but decreases exponentially as the depth increases. The total depth depends on the size of $X$. Finally, the edges are replaced with chains similar to the chains in the bottom layer. The constructed top layer will be placed on the $xy$-plane with $z$-coordinate $10|X|+10$ directly above the bottom layer to create distance from the bottom layer. In the construction of the top layer, the embedding of the edges will not have any bending points.

We shall construct an initial snowflake graph $G_{snow} = (V_{snow}, E_{snow})$ with embedding $E^{G_{snow}} : V_{snow} \rightarrow \mathbb{R}^2$. Let $k \in \mathbb{R}$ be such that $|X| = 3\cdot 2^k$. The binary tree structure attached to an center vertices $d_{1,0}^1, d_{2,0}^1, d_{3,0}^1 \in V_{snow}$ will have depth $\lfloor k \rfloor$. We shall label the vertices in the binary tree structure attached to center vertex $d_{i,0}^1$ at depth $j$ by $d_{i,j}^l \in V_{snow}$, where $l \in [j]$. The edge length of the equilateral triangle that a vertex at depth $j$ forms with its children is $10(|X|+|C|) + 4^{\lfloor k \rfloor-j+2}$. See \cref{snowflake1} for an example initial snowflake graph and embedding.

\begin{figure}
	\includegraphics[width=0.9\textwidth]{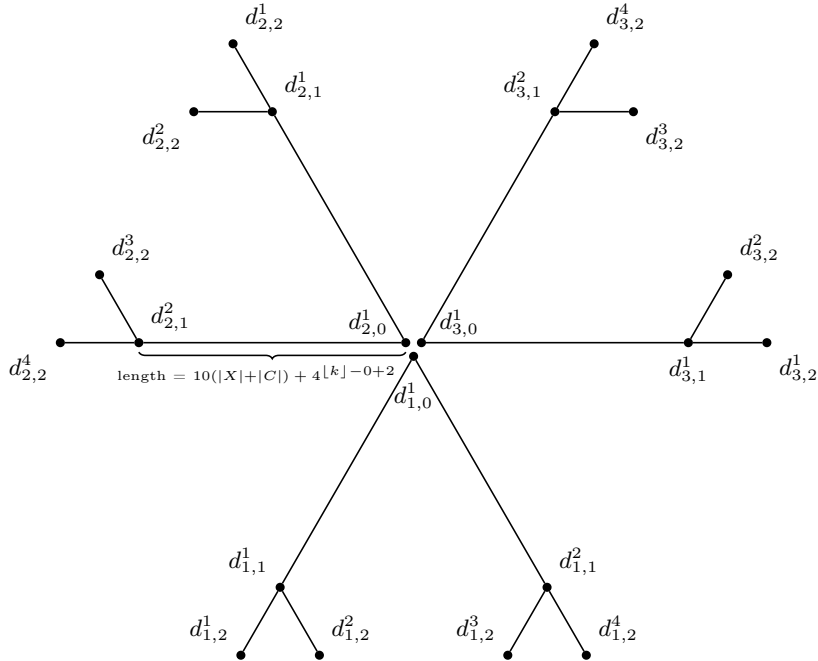}
	\caption{Initial snowflake graph $G_{snow}$ and its embedding $E^{G_{snow}}$ where $|X|=12$. Note that $\lfloor k\rfloor =2$.}
	\label{snowflake1}
\end{figure}

We create a balanced binary tree gadget that has depth $2$. a vertex at depth $j$ forms an equilateral triangle with edge length $10(|X|+|C|)+4^{2-j}$. The vertices in the balanced binary tree gadget will follow the same naming convention as the aforementioned vertices in the binary tree structure in the initial snowflake graph. See \cref{balancedbinarytreegadget} for a figure of the balanced binary tree gadget.

\begin{figure}
	\centering
	\includegraphics[width=0.5\textwidth]{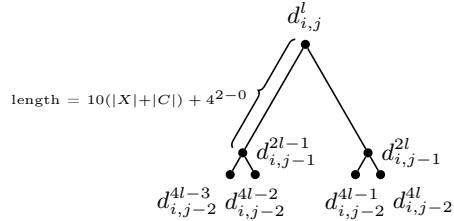}
	\caption{balanced binary tree gadget.}
	\label{balancedbinarytreegadget}
\end{figure}


If $k \notin \mathbb{N}$, we have that the number of leaves in $G_{snow}$ is smaller than $|X|$. We replace $\frac{|X|- 3\cdot 2^{\lfloor k \rfloor}}{3}$ leaves with the balanced binary tree gadget to create the unbalanced snowflake graph $G'_{snow} = (V'_{snow}, E'_{snow})$. Let $d_{i,j}^l$ be a leaf that is replaced by the balanced binary tree gadget. We label the root vertex of that balanced binary tree gadget with $d_{i,j}^l$.

For the embedding $E^{G'_{snow}}: V'_{snow} \rightarrow \mathbb{R}^2$, the balanced binary tree gadgets are embedded in the same manner as the binary tree structures in $E^{G_{snow}}$. This replacement ensures that that the number of leaves in $G'_{snow}$ is equivalent to $|X|$. Note that it is always possible to replace (a part of) the leaves of $G_{snow}$ with the balanced binary tree gadget to achieve the desired number of leaves in $G'_{snow}$. See \cref{snowflake2} for an example graph and embedding. 

\begin{figure}
	\includegraphics[width=\textwidth]{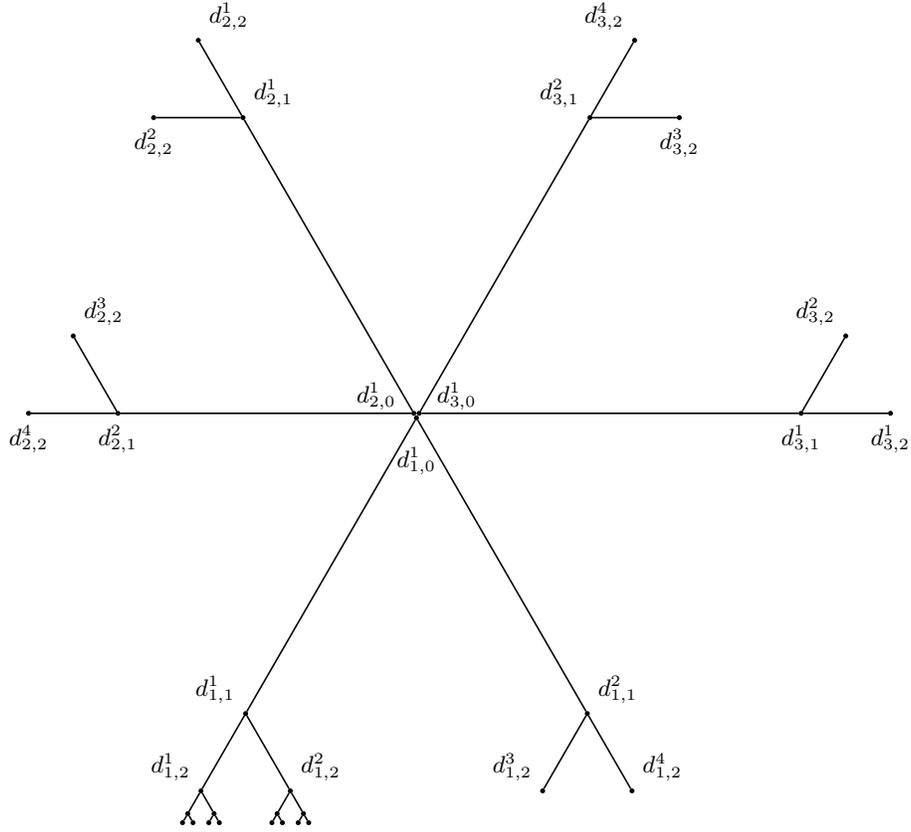}
	\caption{unbalanced snowflake graph $G'_{snow}$ and its embedding $E^{G'_{snow}}$ where $|X|=21$. Note that $\lfloor k\rfloor=2$.}
	\label{snowflake2}
\end{figure}

We modify the snowflake graph by replacing each internal vertex with 3 vertices that form a equilateral triangle with edge length 1. One these vertices will be placed on the exact same location as the original internal vertex. The other 2 vertices will be placed on the edges between the internal vertex and its children. The vertices placed on the edges are connected to the child at the end of its corresponding edge. See \cref{internalvertexgadget,,finaltoplayer} for an illustration of the gadget.

\begin{figure}
	\centering
	\includegraphics[width=0.5\textwidth]{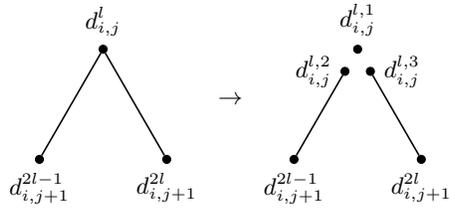}
	\caption{internal vertex gadget.}
	\label{internalvertexgadget}
\end{figure}

\begin{figure}
	\centering
	\includegraphics[width=0.9\textwidth]{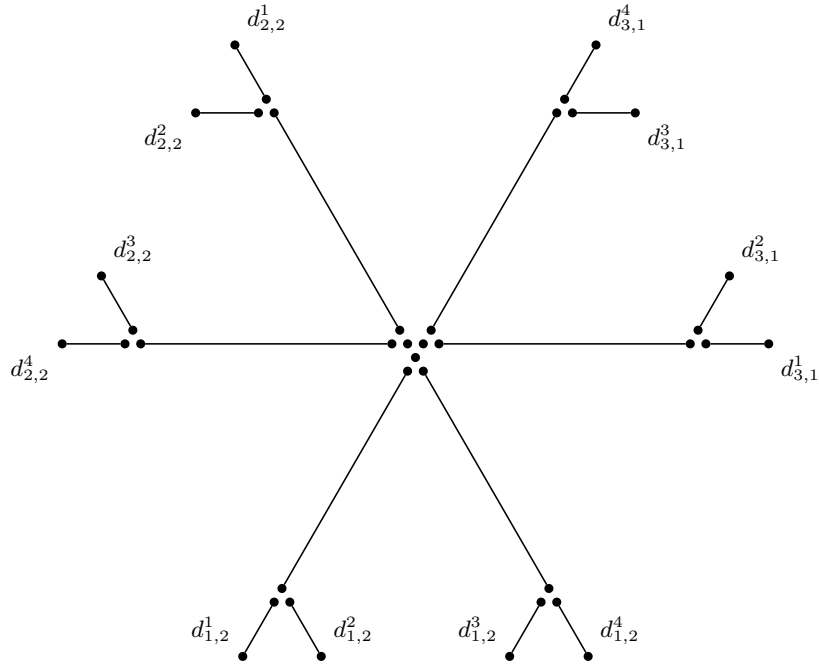}
	\caption{Graph $G_{triple}$, which is internal vertex gadget applied to $G_{snow}$ of \cref{balancedbinarytreegadget}.}
	\label{finaltoplayer}
\end{figure}

Let $G_{triple} = (V_{triple}, E_{triple})$ denote the the graph after applying the internal vertex gadget to graph $G'_{snow}$ and embedding $E^{G_{triple}}: V_{triple} \rightarrow \mathbb{R}^2$ be as described above. 

Let $d_{i,j}^l$ be a leaf of one of the trees in $G'_{snow}$. For convenience, we shall also use the label $d_{i,j}^{l,1}$ refer to the same leaf.

Let $l(e)$, where $e \in E_{triple}$, denote the length of $e$ in embedding $E^{G_{triple}}$. The final step for the top layer is to replace each edge $e$ in $G_{triple}$ with $\hat{n} = \left\lceil\frac{l(e)}{d}\right\rceil$ copies of the triple $A_e[z] = \{\alpha_e[z], \beta_e[z], \gamma_e[z]\}$, where $z \in [\hat{n}]$. These triples are embedded such that they have the same distance properties as they have in the bottom layer. These properties are:
\begin{itemize}
	\item The distance between agents $\alpha_s[z]$ and $\beta_s[z]$ is $\epsilon$.
	\item The distance between agents $\alpha_s[z]$ (resp. $\beta_s[z]$) and $\gamma_s[z]$ is 1.
	\item The distance between agents $\alpha_s[z]$ (resp. $\beta_s[z]$) and $\gamma_s[z-1]$ is 1.
\end{itemize}

The embedding of the chain $A_s[z]$ is almost the exact same as in the bottom layer. Instead of `bending' the chain $A_s[z]$ in the negative direction, in the top layer the chain is `bent' in the positive $z$-direction in embedding $E^t$. See \cref{gridgraphfrompcx3c5,,gridgraphfrompcx3c6,,gridgraphfrompcx3c7} for an example.

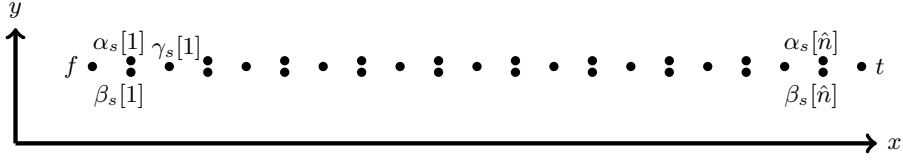
\begin{figure}
	\input{sections/figures/gridgraph6}
	\caption{Projection of example chain $A_s[z]$, where $z \in [\hat{n}]$, of edge $e$, where $t, f$ are the endpoints of $e$, on the $xy$-plane.}
	\label{gridgraphfrompcx3c5}
\end{figure}
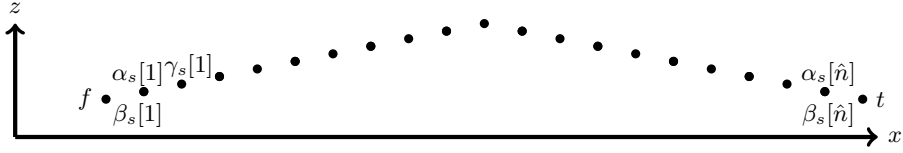
\begin{figure}
	\input{sections/figures/gridgraph7}
	\caption{Projection of example chain $A_s[z]$, where $z \in [\hat{n}]$, of edge $e$, where $t, f$ are the endpoints of $e$, on the $xz$-plane.}
	\label{gridgraphfrompcx3c6}
\end{figure}
\begin{figure}
\input{sections/figures/gridgraph8}
\caption{The edges of \cref{finaltoplayer} replaced with the chains. The number of agents and the distances between the agents are exaggerated for clarity.}
\label{gridgraphfrompcx3c7}
\end{figure}
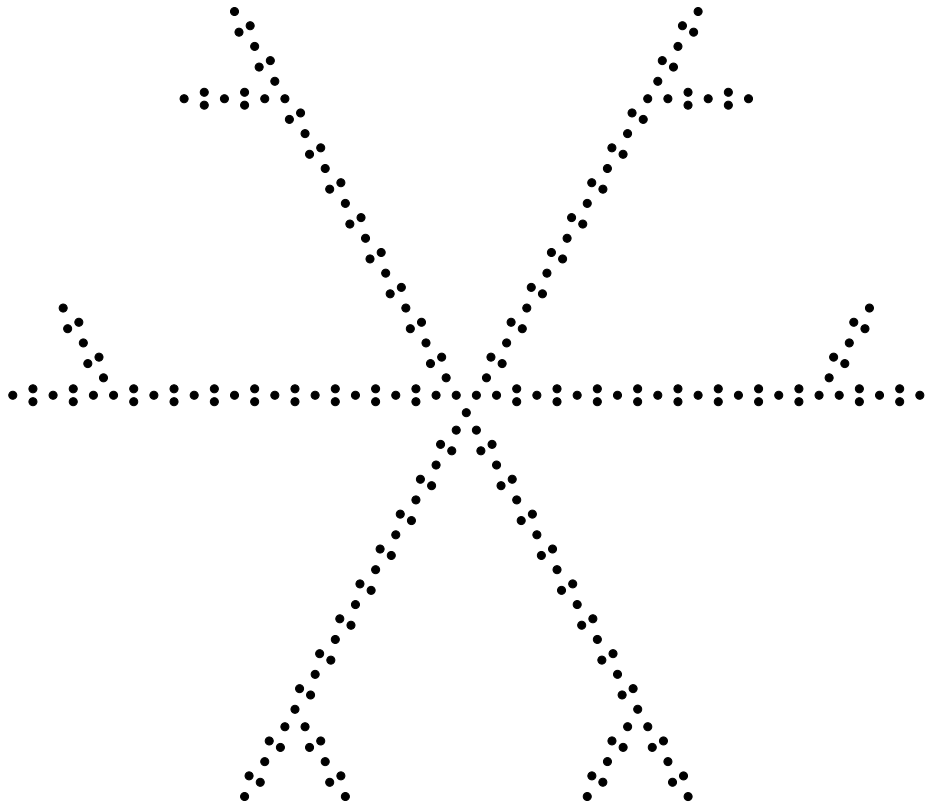

In conclusion, for the top layer we have that $V_t = V_{triple} \cup \\ \bigcup\limits_{e \in E_{triple}, z \in [\hat{n}]}A_s[z]$. As mentioned in the beginning of this section, the top layer is placed on the $xy$-plane with $z$-coordinate $10|X|+10$. Thus, for each agent $a \in V_{triple}$, we have that $E^t(a) = (E^{G_{triple}}_x, E^{G_{triple}}_y, 10|X|+10)$. The agents in $\bigcup\limits_{e \in E_{triple}, z \in [\hat{n}]}A_s[z]$ are embedded in $E^t$ as defined above.

%% file: sections/figures/gridgraph6.tex
\resizebox{\textwidth}{!}{%
	\begin{tikzpicture}[-,auto, semithick]
		\tikzset{
			element/.style={
				fill=black,draw=black,text=black,shape=circle,inner sep=1pt,minimum size=1pt
			},
			set/.style={
				fill=white,draw=white,text=black,shape=circle
			},
		}
		
		\node[element, label= left:{$f$}] 	(1) at (0,0) {};
		\node[element, label= right:{$t$}] 	(2) at (10,0) {};

		\node[element, label={[shift={(-0.15,-0.1)}]$\alpha_s[1]$}] (3-1) at (0.5,0.075){};
		\node[element, label={[shift={(-0.15,-0.65)}]$\beta_s[1]$}] (3-1) at (0.5,-0.075) {};
		\node[element, label={[shift={(0.1,-0.1)}]$\gamma_s[1]$}] 	(3-1) at (1,0) {};
		
		\node[element] 	(3-1) at (1.5,0.075) {};
		\node[element] 	(3-1) at (1.5,-0.075) {};
		\node[element] 	(3-1) at (2,0) {};
		
		\node[element] 	(3-1) at (2.5,0.075) {};
		\node[element] 	(3-1) at (2.5,-0.075) {};
		\node[element] 	(3-1) at (3,0) {};
		
		\node[element] 	(3-1) at (3.5,0.075) {};
		\node[element] 	(3-1) at (3.5,-0.075) {};
		\node[element] 	(3-1) at (4,0) {};
		
		\node[element] 	(3-1) at (4.5,0.075) {};
		\node[element] 	(3-1) at (4.5,-0.075) {};
		\node[element] 	(3-1) at (5,0) {};
		
		\node[element] 	(3-1) at (5.5,0.075) {};
		\node[element] 	(3-1) at (5.5,-0.075) {};
		\node[element] 	(3-1) at (6,0) {};
		
		\node[element] 	(3-1) at (6.5,0.075) {};
		\node[element] 	(3-1) at (6.5,-0.075) {};
		\node[element] 	(3-1) at (7,0) {};
		
		\node[element] 	(3-1) at (7.5,0.075) {};
		\node[element] 	(3-1) at (7.5,-0.075) {};
		\node[element] 	(3-1) at (8,0) {};
		
		\node[element] 	(3-1) at (8.5,0.075) {};
		\node[element] 	(3-1) at (8.5,-0.075) {};
		\node[element] 	(3-1) at (9,0) {};
		
		\node[element, label={[shift={(-0.15,-0.1)}]$\alpha_s[\hat{n}]$}] 	(3-1) at (9.5,0.075) {};
		\node[element, label={[shift={(-0.15,-0.65)}]$\beta_s[\hat{n}]$}] 	(3-1) at (9.5,-0.075) {};
		
		\draw[->,ultra thick] (-1,-1)--(10.2,-1) node[right]{$x$};
		\draw[->,ultra thick] (-1,-1)--(-1,0.5) node[above]{$y$};
	\end{tikzpicture}
}

%% file: sections/figures/gridgraph7.tex
\resizebox{\textwidth}{!}{%
	\begin{tikzpicture}[-,auto, semithick]
		\tikzset{
			element/.style={
				fill=black,draw=black,text=black,shape=circle,inner sep=1pt,minimum size=1pt
			},
			set/.style={
				fill=white,draw=white,text=black,shape=circle
			},
		}
		
		\node[element, label= left:{$f$}] 	(1) at (0,0) {};
		\node[element, label= right:{$t$}] 	(2) at (10,0) {};

		\node[element, label={[shift={(-0.075,-0.1)}]$\alpha_s[1]$}] (3-1) at (0.5,0.1){};
		\node[element, label={[shift={(-0.075,-0.65)}]$\beta_s[1]$}] (3-1) at (0.5,0.1) {};
		\node[element, label={[shift={(0.1,-0.1)}]$\gamma_s[1]$}] 	(3-1) at (1,0.2) {};
		
		\node[element] 	(3-1) at (1.5,0.3) {};
		\node[element] 	(3-1) at (1.5,0.3) {};
		\node[element] 	(3-1) at (2,0.4) {};
		
		\node[element] 	(3-1) at (2.5,0.5) {};
		\node[element] 	(3-1) at (2.5,0.5) {};
		\node[element] 	(3-1) at (3,0.6) {};
		
		\node[element] 	(3-1) at (3.5,0.7) {};
		\node[element] 	(3-1) at (3.5,0.7) {};
		\node[element] 	(3-1) at (4,0.8) {};
		
		\node[element] 	(3-1) at (4.5,0.9) {};
		\node[element] 	(3-1) at (4.5,0.9) {};
		\node[element] 	(3-1) at (5,1) {};
		
		\node[element] 	(3-1) at (5.5,0.9) {};
		\node[element] 	(3-1) at (5.5,0.9) {};
		\node[element] 	(3-1) at (6,0.8) {};
		
		\node[element] 	(3-1) at (6.5,0.7) {};
		\node[element] 	(3-1) at (6.5,0.7) {};
		\node[element] 	(3-1) at (7,0.6) {};
		
		\node[element] 	(3-1) at (7.5,0.5) {};
		\node[element] 	(3-1) at (7.5,0.5) {};
		\node[element] 	(3-1) at (8,0.4) {};
		
		\node[element] 	(3-1) at (8.5,0.3) {};
		\node[element] 	(3-1) at (8.5,0.3) {};
		\node[element] 	(3-1) at (9,0.2) {};
		
		\node[element, label={[shift={(0.05,-0.1)}]$\alpha_s[\hat{n}]$}] 	(3-1) at (9.5,0.1) {};
		\node[element, label={[shift={(0.05,-0.65)}]$\beta_s[\hat{n}]$}] 	(3-1) at (9.5,0.1) {};

		\draw[->,ultra thick] (-1.2,-0.5)--(10.2,-0.5) node[right]{$x$};
		\draw[->,ultra thick] (-1.2,-0.5)--(-1.2,1) node[above]{$z$};
	\end{tikzpicture}
}

%% file: sections/figures/gridgraph8.tex
\resizebox{\textwidth}{!}{%
	\begin{tikzpicture}[-,auto, semithick]
		\tikzset{
			element/.style={
				fill=black,draw=black,text=black,shape=circle,inner sep=4pt,minimum size=4pt
			},
			set/.style={
				fill=white,draw=white,text=black,shape=circle
			}, 
			expand bubble/.style={
				preaction={draw,line width=0.6em},
				white,fill,draw,line width=0.5em,
			},
		}
		\node[element] at (0.5,0.) {};
		\node[element] at (-0.5,0.) {};
		
		\node[element] at (0.,-0.8660254037844386) {};
		
		\node[element] at (-8.5,-15.5924318643) {};
		
		\node[element] at (8.503442101261754,-15.590444566209953) {};
		
		\node[element] at (17.5,0.) {};
		
		\node[element] at (9.,14.722431864335459) {};
		
		\node[element] at (-17.5,0.) {};
		
		\node[element] at (-9.,14.722431864335459) {};
		
		\node[element] at (-11.,-19.92255888322) {};
		
		\node[element] at (-6.000000000001899,-19.9225588832211) {};
		
		\node[element] at (6.,-19.92255888322) {};
		
		\node[element] at (11.003442101259846,-19.92454618131115) {};
		
		\node[element] at (-11.5,19.05255888322) {};
		
		\node[element] at (-13.999999999967393,14.722431864316636) {};
		
		\node[element] at (-20.,4.33012701892) {};
		
		\node[element] at (-22.499999999998103,-1.0962342145148796E-12) {};
		
		\node[element] at (22.5,0.) {};
		
		\node[element] at (20.,4.330127018922194) {};
		
		\node[element] at (14.,14.7224318643) {};
		
		\node[element] at (11.50000000003071,19.05255888323992) {};
		
		\node[element] at (-1.,0.8660254037844387) {};
		
		\node[element] at (-1.5,0.) {};
		
		\node[element] at (1.,0.8660254037844386) {};
		
		\node[element] at (1.5,0.) {};
		
		\node[element] at (3.5,0.) {};
		
		\node[element] at (5.5,0.) {};
		
		\node[element] at (7.5,0.) {};
		
		\node[element] at (9.5,0.) {};
		
		\node[element] at (11.5,0.) {};
		
		\node[element] at (13.5,0.) {};
		
		\node[element] at (15.5,0.) {};
		
		\node[element] at (2.5,0.3201562118716411) {};
		
		\node[element] at (2.5,-0.3201562118716411) {};
		
		\node[element] at (4.5,0.3201562118716411) {};
		
		\node[element] at (4.5,-0.3201562118716411) {};
		
		\node[element] at (6.5,0.3201562118716411) {};
		
		\node[element] at (6.5,-0.3201562118716411) {};
		
		\node[element] at (8.5,0.3201562118716494) {};
		
		\node[element] at (8.5,-0.32015621187164944) {};
		
		\node[element] at (10.5,0.32015621187164384) {};
		
		\node[element] at (10.5,-0.3201562118716439) {};
		
		\node[element] at (12.5,0.32015621187163) {};
		
		\node[element] at (12.5,-0.32015621187163) {};
		
		\node[element] at (14.5,0.32015621187163) {};
		
		\node[element] at (14.5,-0.32015621187163) {};
		
		\node[element] at (16.5,0.3201562118716688) {};
		
		\node[element] at (16.5,-0.32015621187166887) {};
		
		\node[element] at (2.,2.5980762113533156) {};
		
		\node[element] at (3.,4.330127018922194) {};
		
		\node[element] at (4.,6.062177826491069) {};
		
		\node[element] at (5.,7.794228634059945) {};
		
		\node[element] at (6.,9.526279441628821) {};
		
		\node[element] at (7.,11.2583302491977) {};
		
		\node[element] at (8.,12.990381056766577) {};
		
		\node[element] at (1.2227365873397644,1.8921289135046986) {};
		
		\node[element] at (1.7772634126602358,1.5719727016330558) {};
		
		\node[element] at (2.777263412660237,3.3040235092019326) {};
		
		\node[element] at (2.222736587339764,3.6241797210735767) {};
		
		\node[element] at (3.2227365873397615,5.356230528642455) {};
		
		\node[element] at (3.7772634126602402,5.036074316770809) {};
		
		\node[element] at (4.777263412660231,6.76812512433969) {};
		
		\node[element] at (4.222736587339771,7.088281336211326) {};
		
		\node[element] at (5.777263412660234,8.500175931908561) {};
		
		\node[element] at (5.222736587339765,8.820332143780202) {};
		
		\node[element] at (6.222736587339767,10.552382951349085) {};
		
		\node[element] at (6.777263412660235,10.232226739477444) {};
		
		\node[element] at (7.777263412660203,11.964277547046333) {};
		
		\node[element] at (7.222736587339793,12.28443375891794) {};
		
		\node[element] at (8.222736587339822,14.016484566486817) {};
		
		\node[element] at (8.777263412660194,13.696328354615233) {};
		
		\node[element] at (9.500000000006144,15.588457268116352) {};
		
		\node[element] at (10.,14.722431864328366) {};
		
		\node[element] at (10.500000000018408,17.320508075678102) {};
		
		\node[element] at (10.277263412671406,16.29440456595941) {};
		
		\node[element] at (9.722736587353138,16.614560777835027) {};
		
		\node[element] at (11.277263412683617,18.026455373521273) {};
		
		\node[element] at (10.722736587365524,18.346611585396786) {};
		
		\node[element] at (12.,14.722431864314183) {};
		
		\node[element] at (13.000000000002256,15.042588076178719) {};
		
		\node[element] at (12.999999999997716,14.402275652435467) {};
		
		\node[element] at (11.000000000002267,15.042588076192999) {};
		
		\node[element] at (10.999999999997726,14.402275652449552) {};
		
		\node[element] at (18.,0.8660254037844387) {};
		
		\node[element] at (19.,2.5980762113533156) {};
		
		\node[element] at (18.22273658733974,1.8921289135047132) {};
		
		\node[element] at (18.777263412660265,1.5719727016330403) {};
		
		\node[element] at (19.777263412660265,3.3040235092019152) {};
		
		\node[element] at (19.22273658733974,3.624179721073588) {};
		
		\node[element] at (18.5,0.) {};
		
		\node[element] at (20.5,0.) {};
		
		\node[element] at (21.5,0.32015621187167437) {};
		
		\node[element] at (21.5,-0.3201562118716744) {};
		
		\node[element] at (19.5,-0.3201562118716744) {};
		
		\node[element] at (19.5,0.32015621187167437) {};
		
		\node[element] at (-3.5,0.) {};
		
		\node[element] at (-5.5,0.) {};
		
		\node[element] at (-7.5,0.) {};
		
		\node[element] at (-9.5,0.) {};
		
		\node[element] at (-11.5,0.) {};
		
		\node[element] at (-13.5,0.) {};
		
		\node[element] at (-15.5,0.) {};
		
		\node[element] at (-16.5,0.3201562118716466) {};
		
		\node[element] at (-16.5,-0.32015621187164667) {};
		
		\node[element] at (-14.5,0.32015621187163) {};
		
		\node[element] at (-14.5,-0.32015621187163) {};
		
		\node[element] at (-12.5,0.32015621187163) {};
		
		\node[element] at (-12.5,-0.32015621187163) {};
		
		\node[element] at (-10.5,0.3201562118716328) {};
		
		\node[element] at (-10.5,-0.3201562118716328) {};
		
		\node[element] at (-8.5,0.32015621187164384) {};
		
		\node[element] at (-8.5,-0.3201562118716439) {};
		
		\node[element] at (-6.5,0.3201562118716411) {};
		
		\node[element] at (-6.5,-0.3201562118716411) {};
		
		\node[element] at (-4.5,0.3201562118716411) {};
		
		\node[element] at (-4.5,-0.3201562118716411) {};
		
		\node[element] at (-2.5,0.32015621187164245) {};
		
		\node[element] at (-2.5,-0.3201562118716425) {};
		
		\node[element] at (-2.,2.5980762113533165) {};
		
		\node[element] at (-3.,4.330127018922194) {};
		
		\node[element] at (-4.,6.062177826491069) {};
		
		\node[element] at (-5.,7.794228634059952) {};
		
		\node[element] at (-6.,9.526279441628823) {};
		
		\node[element] at (-7.,11.258330249197689) {};
		
		\node[element] at (-8.,12.990381056766578) {};
		
		\node[element] at (-8.22273658733973,14.016484566486852) {};
		
		\node[element] at (-8.777263412660258,13.69632835461518) {};
		
		\node[element] at (-7.777263412660218,11.964277547046319) {};
		
		\node[element] at (-7.22273658733977,12.284433758917947) {};
		
		\node[element] at (-6.222736587339765,10.552382951349072) {};
		
		\node[element] at (-6.7772634126602185,10.23222673947744) {};
		
		\node[element] at (-5.777263412660244,8.500175931908561) {};
		
		\node[element] at (-5.222736587339754,8.820332143780213) {};
		
		\node[element] at (-4.2227365873397655,7.088281336211331) {};
		
		\node[element] at (-4.777263412660233,6.768125124339692) {};
		
		\node[element] at (-3.777263412660237,5.036074316770808) {};
		
		\node[element] at (-3.2227365873397584,5.356230528642454) {};
		
		\node[element] at (-2.2227365873397558,3.62417972107358) {};
		
		\node[element] at (-2.7772634126602433,3.3040235092019286) {};
		
		\node[element] at (-1.777263412660235,1.571972701633057) {};
		
		\node[element] at (-1.222736587339765,1.8921289135046986) {};
		
		\node[element] at (-18.000000000000192,0.8660254037843291) {};
		
		\node[element] at (-18.5,0.) {};
		
		\node[element] at (-19.00000000000056,2.598076211352972) {};
		
		\node[element] at (-20.5,-6.577405287091727E-13) {};
		
		\node[element] at (-21.499999999999105,0.32015621187367704) {};
		
		\node[element] at (-21.499999999998963,-0.32015621187543103) {};
		
		\node[element] at (-19.499999999999886,-0.32015621187200327) {};
		
		\node[element] at (-19.5000000000001,0.3201562118713456) {};
		
		\node[element] at (-18.777263412660524,1.5719727016327898) {};
		
		\node[element] at (-18.22273658734022,1.8921289135044983) {};
		
		\node[element] at (-19.77726341266323,3.3040235091990144) {};
		
		\node[element] at (-19.22273658733733,3.624179721073954) {};
		
		\node[element] at (-10.,14.722431864331696) {};
		
		\node[element] at (-12.,14.722431864324168) {};
		
		\node[element] at (-9.500000000003242,15.588457268117999) {};
		
		\node[element] at (-10.50000000000976,17.32050807568309) {};
		
		\node[element] at (-12.999999999984906,15.042588076242883) {};
		
		\node[element] at (-12.999999999982496,14.40227565239792) {};
		
		\node[element] at (-11.000000000001197,15.042588076199516) {};
		
		\node[element] at (-10.999999999998787,14.402275652456348) {};
		
		\node[element] at (-10.27726341266613,16.294404565963696) {};
		
		\node[element] at (-9.72273658734688,16.61456077783742) {};
		
		\node[element] at (-11.277263412708495,18.02645537348929) {};
		
		\node[element] at (-10.722736587301267,18.346611585413807) {};
		
		\node[element] at (0.500101214362474,-1.7319923635428225) {};
		
		\node[element] at (1.5003036430874208,-3.4639262830595876) {};
		
		\node[element] at (2.50050607181237,-5.195860202576356) {};
		
		\node[element] at (3.5007085005373213,-6.92779412209313) {};
		
		\node[element] at (4.500910929262278,-8.659728041609911) {};
		
		\node[element] at (5.501113357987235,-10.391661961126694) {};
		
		\node[element] at (6.501315786712181,-12.123595880643457) {};
		
		\node[element] at (7.501518215437113,-13.855529800160198) {};
		
		\node[element] at (7.729934586918504,-14.880383917212601) {};
		
		\node[element] at (8.275025729780358,-14.565590449157538) {};
		
		\node[element] at (7.2786617025168505,-12.829452330059159) {};
		
		\node[element] at (6.724172299632444,-13.149673350744504) {};
		
		\node[element] at (6.2784592737919445,-11.097518410542383) {};
		
		\node[element] at (5.72396987090747,-11.417739431227767) {};
		
		\node[element] at (5.278256845066982,-9.365584491025619) {};
		
		\node[element] at (4.723767442182533,-9.685805511710988) {};
		
		\node[element] at (4.278054416342035,-7.633650571508832) {};
		
		\node[element] at (3.7235650134575655,-7.953871592194213) {};
		
		\node[element] at (3.27785198761708,-5.90171665199205) {};
		
		\node[element] at (2.7233625847326084,-6.2219376726774325) {};
		
		\node[element] at (2.2776495588921404,-4.169782732475276) {};
		
		\node[element] at (1.7231601560076508,-4.490003753160669) {};
		
		\node[element] at (1.2774471301671986,-2.437848812958505) {};
		
		\node[element] at (0.7229577272826965,-2.758069833643906) {};
		
		\node[element] at (-0.4998987788075001,-1.7321092397649804) {};
		
		\node[element] at (-1.4996963364225004,-3.464276911726065) {};
		
		\node[element] at (-2.499493894037501,-5.196444583687151) {};
		
		\node[element] at (-3.4992914516524998,-6.928612255648232) {};
		
		\node[element] at (-4.499089009267505,-8.660779927609324) {};
		
		\node[element] at (-5.498886566882493,-10.392947599570387) {};
		
		\node[element] at (-6.498684124497479,-12.125115271531447) {};
		
		\node[element] at (-7.498481682112478,-13.85728294349253) {};
		
		\node[element] at (-8.271823196912203,-14.567524382361928) {};
		
		\node[element] at (-7.726658485200276,-14.882190425430602) {};
		
		\node[element] at (-6.721300783214163,-13.15124480685428) {};
		
		\node[element] at (-7.275865023395801,-12.831153408169708) {};
		
		\node[element] at (-6.2760674657807956,-11.09898573620863) {};
		
		\node[element] at (-5.72150322559917,-11.419077134893193) {};
		
		\node[element] at (-5.276269908165813,-9.366818064247578) {};
		
		\node[element] at (-4.721705667984187,-9.68690946293214) {};
		
		\node[element] at (-4.27647235055079,-7.63465039228651) {};
		
		\node[element] at (-3.7219081103692124,-7.954741790971045) {};
		
		\node[element] at (-3.276674792935787,-5.902482720325423) {};
		
		\node[element] at (-2.722110552754211,-6.222574119009958) {};
		
		\node[element] at (-2.2768772353207907,-4.1703150483643405) {};
		
		\node[element] at (-1.7223129951392115,-4.490406447048876) {};
		
		\node[element] at (-1.2770796777057918,-2.438147376403253) {};
		
		\node[element] at (-0.7225154375242089,-2.7582387750877917) {};
		
		\node[element] at (9.003098088521511,-16.456668494761615) {};
		
		\node[element] at (10.002410063041049,-18.189116351864975) {};
		
		\node[element] at (8.003098167399914,-16.45627130860572) {};
		
		\node[element] at (7.002410299676228,-18.187924793397414) {};
		
		\node[element] at (7.225554423555656,-17.16190983250353) {};
		
		\node[element] at (7.779954043520483,-17.482286269499618) {};
		
		\node[element] at (6.228704092760568,-18.897768935008887) {};
		
		\node[element] at (6.773706206915663,-19.21271474160852) {};
		
		\node[element] at (9.225427104183659,-17.482860391433327) {};
		
		\node[element] at (9.780081047378895,-17.16292445519325) {};
		
		\node[element] at (10.2303000189465,-19.214087652925347) {};
		
		\node[element] at (10.775552145354403,-18.899574880250796) {};
		
		\node[element] at (-8.0000000000002,-16.458457268084526) {};
		
		\node[element] at (-9.00000000000021,-16.458457268084317) {};
		
		\node[element] at (-7.000000000000606,-18.190508075653582) {};
		
		\node[element] at (-10.000000000000606,-18.19050807565305) {};
		
		\node[element] at (-9.777263412660474,-17.164404565932884) {};
		
		\node[element] at (-9.222736587340343,-17.484560777804482) {};
		
		\node[element] at (-10.777263412663114,-18.896455373499144) {};
		
		\node[element] at (-10.222736587337499,-19.21661158537392) {};
		
		\node[element] at (-7.777263412660744,-17.484560777804862) {};
		
		\node[element] at (-7.222736587340055,-17.164404565933257) {};
		
		\node[element] at (-6.777263412663967,-19.216611585374512) {};
		
		\node[element] at (-6.22273658733854,-18.89645537350017) {};
		
	\end{tikzpicture}
}

%% file: sections/construction-ascendinglayer.tex
\subsubsection{Ascending layer.}
The purpose of the ascending layer is to connect the bottom layer with the top layer. In the construction of the ascending layer, the embedding of the edges will not have any bending points.

Let us arbitrarily label the leaves in the top layer as $l_1, \dots , l_m$. We shall connect vertex agent $u_i$ from the bottom layer to the leaf $l_i$ of the top layer, where $i \in [m]$.

We shall create a graph $G_{asc} = (V_{asc}, E_{asc})$ and embedding $E^{G_{asc}} : V_{asc} \rightarrow \mathbb{R}^3$. For each $u_i \in U$, let us create 2 auxiliary vertices $u'_i, l'_i$. The vertex set $V_{asc}$ consists of all element vertices from the bottom layer, the leaves from the top layer and the aforementioned auxiliary vertices. That is, $V_{asc} = U \cup \{l_i, u'_i, l'_i | i \in [m]\} $. In embedding $E^{G_{asc}}$, the element vertices and leaves are embedded in the same way as in embedding $E^b$ and $E^t$, respectively. That is, for $i \in [m]$, $E^{G_{asc}}(u_i) = E^b(u_i)$ and $E^{G_{asc}}(l_i) = E^t(l_i)$.

The auxiliary vertex $u'_i$ is placed directly above vertex $u_i$ at height $10i$ in embedding $E^{G_{asc}}$. That is, $E^{G_{asc}}(u'_i) = (E^b_x(u_i), E^b_y(u_i), 10i)$. The auxiliary vertex $l'_i$ is placed directly below vertex $l_i$ also at height $10i$ in embedding $E^{G_{asc}}$. That is, $E^{G_{asc}}(l'_i) = (E^t_x(l_i), E^t_y(l_i), 10i)$.

The edges in $E_{asc}$ consists of $\{u_i, u'_i\}, \{u'_i, l'_i\}$, and $\{l'_i,l_i\}$, for each $i \in [m]$. See \cref{ascending1} for an illustration.

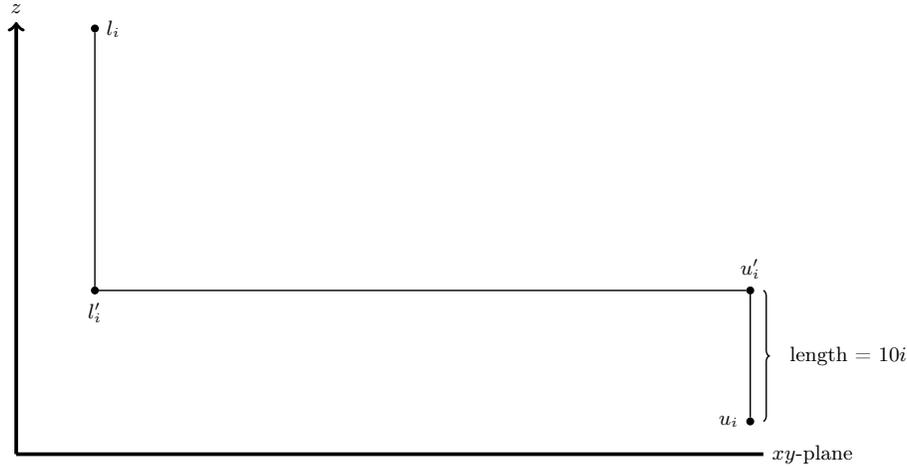
\begin{figure}
	\input{sections/figures/ascending1}
	\caption{Initial ascending graph with element vertex $u_i$ and leaf $l_i$.}
	\label{ascending1}
\end{figure}

Let $l(e)$, where $e \in E_{asc}$, denote the length of $e$ in embedding $E^{G_{asc}}$. The final step for the ascending layer is to replace edge $e$ in $G_{asc}$ with $\hat{n} = \left\lceil \frac{l(e)}{d} \right\rceil$ copies of the triple $A_e[z] = \{\alpha_e[z], \beta_e[z], \gamma_e[z]\}$, where $z \in [\hat{n}]$. These triples are embedded such that they have the same distance properties as in the bottom and top layer.

As in the bottom layer and the top layer the chain $A_e$ is longer than edge $e$. We shall define 3 distinct bending directions for the edges $\{u_i, u'i\}, \{u'_i, l'_i\}$, and $\{l'_i,l_i\}$ in embedding $E^a$. The chain $A_{\{u_i, u'_i\}}[z]$ will be bent in the $u'_i - l'_i$ direction. The chain $A_{\{u'_i, l'_i\}}[z]$ will be bent in the negative $z$-direction. The chain $A_{\{l'_i, l_i\}}[z]$ will be bent in the $l'_i-u'_i$ direction. See \cref{ascending2} for an illustration.

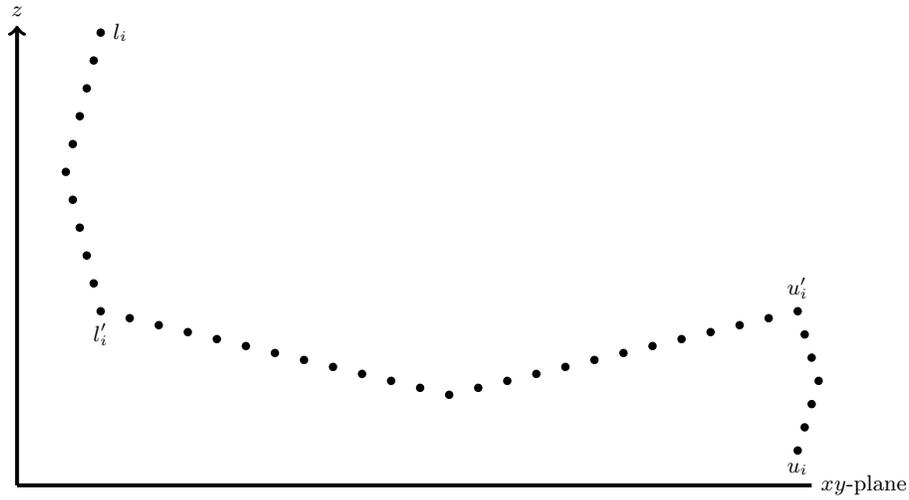
\begin{figure}
	\input{sections/figures/ascending2v2}
	\caption{Replacing the edges in \cref{ascending1} with the chains $A_{e}[z]$, where $e \in \{\{u_i, u'_i\}, \{u'_i, l'_i\}, \{l'_i,l_i\}\}$.}
	\label{ascending2}
\end{figure}

In conclusion, for the ascending layer we have that $V_a = V_{asc} \cup \bigcup\limits_{e \in E_{asc}, z \in [\hat{n}]} A_e[z]$. These agents are embedded as described above.

%% file: sections/figures/ascending1.tex
\resizebox{\textwidth}{!}{%
	\begin{tikzpicture}[-,auto, semithick]
		\tikzset{
			element/.style={
				fill=black,draw=black,text=black,shape=circle,inner sep=1pt,minimum size=1pt
			},
			set/.style={
				fill=white,draw=white,text=black,shape=circle
			},
		}
		
		\node[element, label= right:{$l_i$}] 	(1) at (0,6) {};
		\node[element, label= left:{$u_i$}] 	(2) at (10,0) {};
		
		\node[element, label= below:{$l'_i$}] 	(3) at (0,2) {};
		\node[element, label= above:{$u'_i$}] 	(4) at (10,2) {};

		\path (1) edge (3);
		\path (4) edge (3);
		\path (4) edge (2);
		
		\draw[-,ultra thick] (-1.2,-0.5)--(10.2,-0.5) node[right]{$xy$-plane};
		\draw[->,ultra thick] (-1.2,-0.5)--(-1.2,6.1) node[above]{$z$};
		
		\draw [decorate,
		decoration = {brace,mirror}] (10.2,0) --  (10.2,2) ;
		\node[]  at (11.5,1) {length = $10i$};
	\end{tikzpicture}
}

%% file: sections/figures/ascending2v2.tex
\resizebox{\textwidth}{!}{%
	\begin{tikzpicture}[-,auto, semithick]
		\tikzset{
			element/.style={
				fill=black,draw=black,text=black,shape=circle,inner sep=1pt,minimum size=1pt
			},
			set/.style={
				fill=white,draw=white,text=black,shape=circle
			},
		}
		
		\node[element, label= right:{$l_i$}] 	(1) at (0,6) {};
		\node[element, label= below:{$u_i$}] 	(2) at (10,0) {};
		
		\node[element, label= below:{$l'_i$}] 	(3) at (0,2) {};
		\node[element, label= above:{$u'_i$}] 	(4) at (10,2) {};
		
		\node[element] at (10.1,1/3) {};
		\node[element] at (10.2,2/3) {};
		\node[element] at (10.3,3/3) {};
		\node[element] at (10.2,4/3) {};
		\node[element] at (10.1,5/3) {};
		
		\node[element] at (5/12, 2-0.1) {};
		\node[element] at (10/12,2-0.2) {};
		\node[element] at (15/12,2-0.3) {};
		\node[element] at (20/12,2-0.4) {};
		\node[element] at (25/12,2-0.5) {};
		\node[element] at (30/12,2-0.6) {};
		\node[element] at (35/12,2-0.7) {};
		\node[element] at (40/12,2-0.8) {};
		\node[element] at (45/12,2-0.9) {};
		\node[element] at (50/12,2-1) {};
		\node[element] at (55/12,2-1.1) {};
		\node[element] at (60/12,2-1.2) {};
		\node[element] at (65/12,2-1.1) {};
		\node[element] at (70/12,2-1) {};
		\node[element] at (75/12,2-0.9) {};
		\node[element] at (80/12,2-0.8) {};
		\node[element] at (85/12,2-0.7) {};
		\node[element] at (90/12,2-0.6) {};
		\node[element] at (95/12,2-0.5) {};
		\node[element] at (100/12,2-0.4) {};
		\node[element] at (105/12,2-0.3) {};
		\node[element] at (110/12,2-0.2) {};
		\node[element] at (115/12,2-0.1) {};
		
		\node[element] at (0-0.1,2.4) {};
		\node[element] at (0-0.2,2.8) {};
		\node[element] at (0-0.3,3.2) {};
		\node[element] at (0-0.4,3.6) {};
		\node[element] at (0-0.5,4) {};
		\node[element] at (0-0.4,4.4) {};
		\node[element] at (0-0.3,4.8) {};
		\node[element] at (0-0.2,5.2) {};
		\node[element] at (0-0.1,5.6) {};
		
		\draw[-,ultra thick] (-1.2,-0.5)--(10.2,-0.5) node[right]{$xy$-plane};
		\draw[->,ultra thick] (-1.2,-0.5)--(-1.2,6.1) node[above]{$z$};
	\end{tikzpicture}
}

%% file: sections/construction-size.tex
We shall demonstrate that the size of the constructed 3D-EuclidSR game $G$ is polynomial in the size of the X3C instance $I$. This is shown by determining the order of the number of agents in each layer of $G$ separately.

\subsubsection{Bottom layer.}
We compute the bound on the number of agents in the bottom layer by determining the maximum length of a segment and the total number of segments in $G'(I)$ with embedding $E^{G'(I)}$.

The length of a segment in embedding $E^{G'(I)}$ is at most $10\cdot (|X| + |C|)$, which is the maximum width and height of the graph in the embedding. Each segment $s$ is replaced with an agent triple chain $A_s[z]$ of which the amount of is linear in the length of the segment. Note that the number of segments is at most $2\cdot |E'| = 2 \cdot 3 \cdot |X|$. Additionally, there are element agents, set agents and bending point agents. The number of bending point agents is at most the number of element agents. Note that $|X| = |C|$. Thus, the order of the number of agents is 
\begin{align*}
	&\mathcal{O}(2\cdot 3\cdot|X| \cdot (10\cdot (|X| + |C|)) + |X| + |C| + |X|) \\
	& = \mathcal{O}(60|X|^2 + 60|X||C| + 2|X|  + |C|) \\
	&= \mathcal{O}(|X|^2 + |X||C|) \\
	&= \mathcal{O}(|X|^2).
\end{align*}

\subsubsection{Top layer.}
The bound on the number of agents in top layer be computed by analyzing the binary tree structures and balanced binary tree gadgets of $G'_{snow}$. 

The depth of the tree is $\lfloor k \rfloor + 2$. The edge length between depth $j$ and $j+1$ is $10(|X|+|C|)+4^{\lfloor k \rfloor-j+2}$. The number of edges between depth $j$ and $j+1$ is $2^{j+1}$. There are $2^j$ internal vertices at depth $j$. We have that $|X| \geq 3\cdot 2^{\lfloor k \rfloor} \iff \lg{\frac{|X|}{3}} \geq \lfloor k \rfloor$.

Each edge $e \in E_{triple}$ is replaced with an agent triple chain $A_e$. The amount of agents in chain $A_e$ is linear in the length of the edge $e$. Additionally, each internal vertex in the tree is replaced by a triple. Thus, the total amount of agents is in the order of
\begin{align*}
	&\mathcal{O}(3\cdot (\sum\limits_{j=0}^{\lfloor k \rfloor + 1}2^{j+1} \cdot (10(|X|+|C|)+4^{\lfloor k \rfloor -j + 2})) + 3\cdot \sum\limits_{j = 0}^{\lfloor k \rfloor + 1}3\cdot2^{j}) \\ 
	&=\mathcal{O}(3\cdot (\sum\limits_{j=0}^{\lfloor k \rfloor + 1}2^{j+1}\cdot 20 |X|+2^{2\lfloor k \rfloor -j + 5}) + 9\cdot (2^{\lfloor k \rfloor + 2}- 1)) \\
	&=\mathcal{O}(60|X|\cdot (\sum\limits_{j=0}^{\lfloor k \rfloor + 1}2^{j+1})+ 3\cdot (\sum\limits_{j=0}^{\lfloor k \rfloor + 1}2^{2\lfloor k \rfloor -j + 5}) + 9\cdot (2^{\lfloor k \rfloor + 2}- 1))\\
	&=\mathcal{O}(120|X|\cdot (\sum\limits_{j=0}^{\lfloor k \rfloor + 1}2^{j})+ 3\cdot 2^{2\lfloor k \rfloor + 5}\cdot (\sum\limits_{j=0}^{\lfloor k \rfloor + 1}\frac{1}{2^{j}}) + 9\cdot (2^{\lfloor k \rfloor + 2}- 1))\\
	&=\mathcal{O}(120|X|\cdot (2^{\lfloor k \rfloor + 2}- 1) + 3\cdot 2^{2\lfloor k \rfloor + 5}\cdot (2 - \frac{1}{2^{\lfloor k \rfloor + 1}}) + 9\cdot (2^{\lfloor k \rfloor + 2}- 1)) \\
	&=\mathcal{O}(480|X|\cdot 2^{\lfloor k \rfloor}- 120|X| + 3\cdot 2^5 \cdot (2^{\lfloor k \rfloor})^2 \cdot (2 - \frac{1}{2\cdot 2^{\lfloor k \rfloor}}) + 36\cdot 2^{\lfloor k \rfloor}- 9) \\
	&=\mathcal{O}(480|X|\cdot \frac{|X|}{3}- 120|X| + 3\cdot 2^5 \cdot (\frac{|X|}{3})^2 \cdot (2 - \frac{1}{2\cdot \frac{|X|}{3}}) + 36\cdot \frac{|X|}{3}- 9) \\
	&= \mathcal{O}(|X|^2).
\end{align*}

\subsubsection{Ascending layer.}
We bound the number of agents in the ascending layer by counting the number and lengths of edges in the initial ascending graph.

There are $2$ edges in $E_{asc}$ parallel to the $z$-axis of length $10i$, where $i \in [|X|]$. Additionally, there are $|X|$ edges in $E_{asc}$ parallel to the $xy$-plane, of length at most the width of the top layer. The width of the top layer is $\sum\limits_{j=0}^{\lfloor k \rfloor + 1}10(|X|+|C|)+4^{\lfloor k \rfloor - j + 2}$.

Each edge $e \in E_{asc}$ is replaced with an agent triple chain $A_e$. The amount of agents in chain $A_e$ is linear in the length of edge $e$. Thus, the total amount of agents is in the order of
\begin{align*}
	& \mathcal{O}((\sum\limits_{i = 1}^{|X|}2\cdot10i) + |X|\cdot(\sum\limits_{j = 0}^{\lfloor k \rfloor + 1}10(|X|+|C|)+4^{\lfloor k \rfloor - j + 2})) \\
	& = \mathcal{O}(20\cdot(\sum\limits_{i = 1}^{|X|}i) + |X|\cdot((\sum\limits_{j = 0}^{\lfloor k \rfloor + 1}20|X|)+(\sum\limits_{j = 0}^{\lfloor k \rfloor + 1}4^{\lfloor k \rfloor - j + 2}))) \\
	& = \mathcal{O}(\frac{20|X|(|X|+1)}{2} + (\lfloor k \rfloor + 2) \cdot 20|X|^2+|X|\cdot 4^{\lfloor k \rfloor + 2}(\sum\limits_{j = 0}^{\lfloor k \rfloor + 1}\frac{1}{4^{j}})) \\
	& = \mathcal{O}(\frac{20|X|^2+20|X|}{2} + (\lfloor k \rfloor + 2) \cdot 20|X|^2+|X|\cdot (4\cdot2^{\lfloor k \rfloor})^2\cdot(\frac{1-\frac{1}{4}^{\lfloor k \rfloor +2}}{1- \frac{1}{4}})) \\
	& = \mathcal{O}(\frac{20|X|^2+20|X|}{2} + (\lfloor k \rfloor + 2) \cdot 20|X|^2+|X|\cdot (4\cdot\frac{|X|}{3})^2\cdot(\frac{4(1-(2^4\cdot (2^{\lfloor k \rfloor})^{2})^{-1})}{3})) \\
	& = \mathcal{O}(\frac{20|X|^2+20|X|}{2} + (\lfloor k \rfloor + 2) \cdot 20|X|^2+\frac{16|X|^3}{9}\cdot(\frac{4(1-(2^4\cdot (\frac{|X|}{3})^{2})^{-1})}{3})) \\
	& = \mathcal{O}(\frac{20|X|^2+20|X|}{2} + (\lfloor k \rfloor + 2) \cdot 20|X|^2+\frac{64|X|^3}{27}\cdot(1-\frac{9}{16|X|^2})) \\
	& = \mathcal{O}(\frac{20|X|^2+20|X|}{2} + (\lg(\frac{|X|}{3}) + 2) \cdot 20|X|^2+\frac{64|X|^3}{27}-\frac{576|X|}{432}) \\
	& = \mathcal{O}(|X|^3).
\end{align*}

%% file: sections/ppoutcome.tex
We shall define the permanent popular outcome $\pi_{pp}$ in this section. In this outcome every agent is assigned a room with its closest neighbors. That is, for any outcome $\pi$ and an arbitrary agent $a$, we have that $\delta(a, \pi_{pp}(a)) \leq \delta(a, \pi(a))$ and $\delta(a, \pi_S(a)) \leq \delta(a, \pi(a))$. Additionally, a permanent popular outcome always exists independent of $I$ having a solution.

We shall describe how the agents are assigned to the rooms separately for each layer to construct the permanent popular outcome $\pi_{pp}$.

\subsubsection{Bottom layer.}
We shall first construct the rooms of the bottom layer. Let $u_i$ be an element-agent and $t$ be a set-agent or bending point agent such that $\{u_i, t\} \in E'$. For the chain $A_{\{u_i, t\}}$, w.l.o.g. assume that $t = \gamma_{\{u_i, t\}}[\hat{n}]$. We shall create the rooms $\{\alpha_{\{u_i, t\}}[z], \beta_{\{u_i, t\}}[z], \gamma_{\{u_i, t\}}[z]\} \in \pi_{pp}$ for $z \in [\hat{n}]$. 

Let $w^i_j$ be a set-agent and $f$ be a element-agent or bending point agent such that $\{f,w^i_j\} \in E'$. For the chain $A_{\{f, w^i_j\}}$, w.l.o.g. assume that $f = \gamma_{\{f, w^i_j\}}[0]$ and $w^i_j = \gamma_{\{f,w^i_j\}}[\hat{n}]$. We shall create the rooms $\{\alpha_{\{u_i, t\}}[z], \beta_{\{u_i, t\}}[z], \gamma_{\{u_i, t\}}[z]\} \in \pi_{pp}$ for $z \in [\hat{n}]$. 

See \cref{ppgridgraphrooms} for an illustration.

\begin{figure}
	\input{sections/figures/ppgridgraph}
	\caption{Illustration of the rooms in the bottom layer from \cref{gridgraphfrompcx3cchains} in the permanent popular outcome.}
	\label{ppgridgraphrooms}
\end{figure}
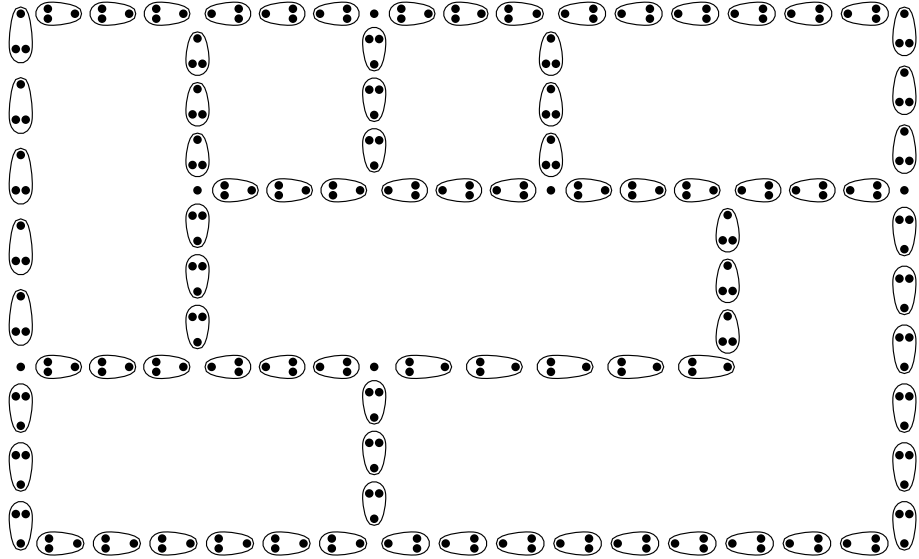

\subsubsection{Ascending layer.}
For the ascending layer, let $u_i$ be an element-agent with corresponding leaf $l_i$ and auxiliary agents $u_i', l_i'$. For chain $A_{\{u_i,u'_i\}}$, w.l.o.g. assume that $u_i = \gamma_{\{u_i,u'_i\}}[0]$ and $u'_i = \gamma_{\{u_i,u'_i\}}[\hat{n}]$. We shall create the rooms $\{\gamma_{\{u_i, u'_i\}}[z-1], \alpha_{\{u_i, u'_i\}}[z], \beta_{\{u_i, u'_i\}}[z]\} \in \pi_{pp}$ for $z \in [\hat{n}]$. 

For chain $A_{\{u'_i,l'_i\}}$, w.l.o.g. assume that $u'_i = \gamma_{\{u'_i,l'_i\}}[0]$ and $l'_i = \gamma_{\{u'_i,l'_i\}}[\hat{n}]$. We shall create the rooms $\{\gamma_{\{u'_i, l'_i\}}[z-1], \alpha_{\{u'_i, l'_i\}}[z], \beta_{\{u'_i, l'_i\}}[z]\} \in \pi_{pp}$ for $z \in [\hat{n}]$. 

For chain $A_{\{l'_i,l_i\}}$, w.l.o.g. assume that $l'_i = \gamma_{\{l'_i,l_i\}}[0]$ and $l_i = \gamma_{\{l'_i,l_i\}}[\hat{n}]$. We shall create the rooms $\{\gamma_{\{l'_i, l_i\}}[z-1], \alpha_{\{l'_i, l_i\}}[z], \beta_{\{l'_i, l_i\}}[z]\} \in \pi_{pp}$ for $z \in [\hat{n}]$. 

See \cref{ppascendingrooms} for an illustration.

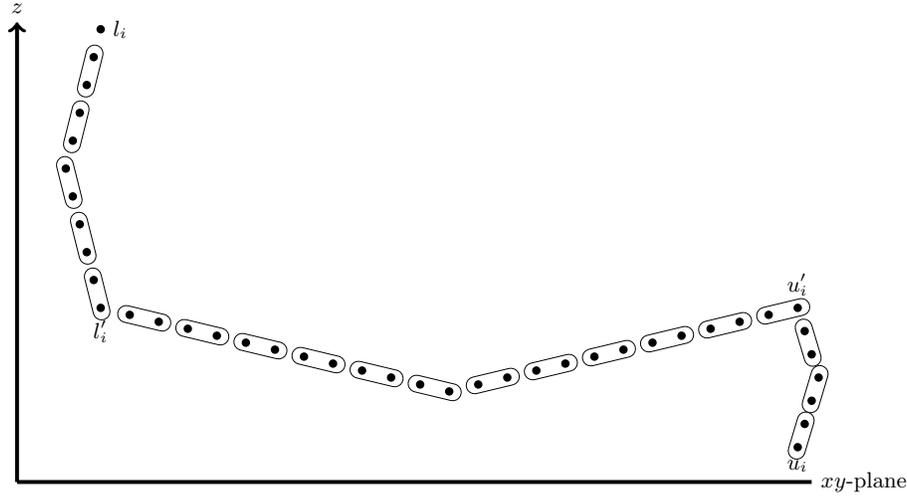
\begin{figure}
\input{sections/figures/ppascendingv2}
\caption{Illustration of the rooms in the ascending layer in the permanent popular outcome.}
\label{ppascendingrooms}
\end{figure}

\subsubsection{Top layer.}
Finally, for the top layer, the rooms depend on the parity of $\lfloor k \rfloor$. The construction of the rooms starts at the leaves and ``crawls up'' towards the center agents. Let us first consider the case where $\lfloor k \rfloor$ is odd. 

Let $j \in \{j' \in [k] | j' \mod 2 = 1\}$ be chosen arbitrarily. Let $e$ be an edge with endpoints $d_{i,j}^{l',1}$ and $d_{i,j-1}^{l,p}$, where $l' \in \{2l-1, 2l\}$ and $p \in \{2,3\}$. W.l.o.g. assume that $d_{i,j}^{l',1} = \gamma_e[0]$. We shall create the rooms $\{\gamma_e[z-1], \alpha_e[z], \beta_e[z]\} \in \pi_{pp}$ for $z \in [\hat{n}]$ and room $\{d_{i,j-1}^{l',1}, d_{i,j-1}^{l',2}, d_{i,j-1}^{l',3}\}$.

Let $j \in \{j' \in [k-1] | j' \mod 2 = 1\}$ be chosen arbitrarily. Let $e$ be an edge with endpoints $d_{i,j}^{l,p}$ and $d_{i,j+1}^{l',1}$, where $l' \in \{2l-1, 2l\}$ and $p \in \{2,3\}$. W.l.o.g. assume that $d_{i,j}^{l,p} = \gamma_e[0]$. We shall create the rooms $\{\gamma_e[z-1], \alpha_e[z], \beta_e[z]\} \in \pi_{pp}$ for $z \in [\hat{n}]$.

Let us now consider the case where $\lfloor k \rfloor$ is even. First we create room $\{d_{1,0}^{1,1}, d_{2,0}^{1,1}, d_{3,0}^{1,1}\}$.

Let $j \in \{j' \in [k] | j' \mod 2 = 0\}$ be chosen arbitrarily. Let $e$ be an edge with endpoints $d_{i,j}^{l',1}$ and $d_{i,j-1}^{l,p}$, where $l' \in \{2l-1, 2l\}$ and $p \in \{2,3\}$. W.l.o.g. assume that $d_{i,j}^{l',1} = \gamma_e[0]$. We shall create the rooms $\{\gamma_e[z-1], \alpha_e[z], \beta_e[z]\} \in \pi_{pp}$ for $z \in [\hat{n}]$ and room $\{d_{i,j-1}^{l',1}, d_{i,j-1}^{l',2}, d_{i,j-1}^{l',3}\}$.

Let $j \in \{j' \in [0,k-1] | j' \mod 2 = 0\}$ be chosen arbitrarily. Let $e$ be an edge with endpoints $d_{i,j}^{l,p}$ and $d_{i,j+1}^{l',1}$, where $l' \in \{2l-1, 2l\}$ and $p \in \{2,3\}$. W.l.o.g. assume that $d_{i,j}^{l,p} = \gamma_e[0]$. We shall create the rooms $\{\gamma_e[z-1], \alpha_e[z], \beta_e[z]\} \in \pi_{pp}$ for $z \in [\hat{n}]$. See \cref{toplayerppoutcome} for an illustration.
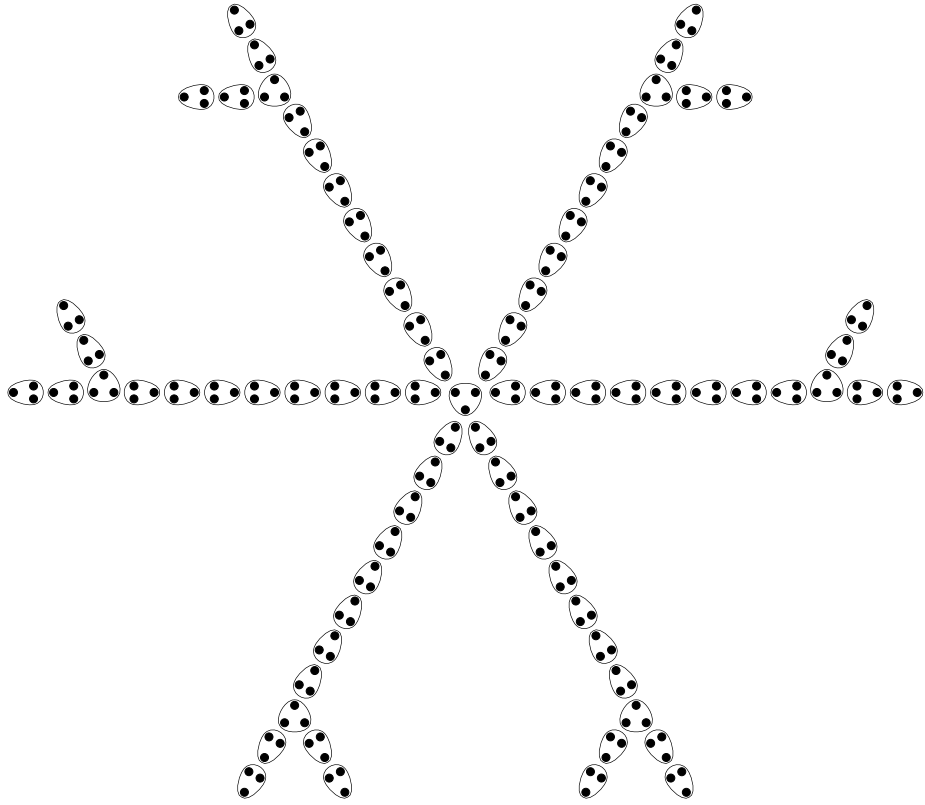
\begin{figure}
	\centering
	\input{sections/figures/snowflakeppoutcome}
	\caption{Illustration of the rooms in the top layer in the permanent popular outcome.}
	\label{toplayerppoutcome}
\end{figure}

%% file: sections/figures/ppgridgraph.tex
\resizebox{\textwidth}{!}{%
	\begin{tikzpicture}[-,auto, semithick]
		\tikzset{
			element/.style={
				fill=black,draw=black,text=black,shape=circle,inner sep=1pt,minimum size=1pt
			}, set/.style={
				fill=black,draw=black,text=black,shape=circle,inner sep=1pt,minimum size=1pt
			}, expand bubble/.style={
				preaction={draw,line width=0.6em},
				white,fill,draw,line width=0.5em,
			},
		}
		\node[element] 	(1) at (0,-2.5) {};
		
		\node[element] 	(1-2-1) at (2.3/6,-2.5+0.07) {};
		\node[element] 	(1-2-2) at (2.3/6,-2.5-0.07) {};
		\node[element] 	(1-2-3) at (4.6/6,-2.5) {};
		\node[element] 	(1-2-4) at (6.9/6,-2.5+0.07) {};
		\node[element] 	(1-2-5) at (6.9/6,-2.5-0.07) {};
		\node[element] 	(1-2-6) at (9.2/6,-2.5) {};
		\node[element] 	(1-2-7) at (11.5/6,-2.5+0.07) {};
		\node[element] 	(1-2-8) at (11.5/6,-2.5-0.07) {};
		
		\node[element] 	(1-14-1) at (0+0.07,-2.5-2.5/6) {};
		\node[element] 	(1-14-2) at (0-0.07,-2.5-2.5/6) {};
		\node[element] 	(1-14-3) at (0,-2.5-5/6) {};
		\node[element] 	(1-14-4) at (0+0.07,-2.5-7.5/6) {};
		\node[element] 	(1-14-5) at (0-0.07,-2.5-7.5/6) {};
		\node[element] 	(1-14-6) at (0,-2.5-10/6) {};
		\node[element] 	(1-14-7) at (0+0.07,-2.5-12.5/6) {};
		\node[element] 	(1-14-8) at (0-0.07,-2.5-12.5/6) {};
		
		\node[element] 	(1-13-1) at (0+0.07,-2.5+5/10) {};
		\node[element] 	(1-13-2) at (0-0.07,-2.5+5/10) {};
		\node[element] 	(1-13-3) at (0,-2.5+10/10) {};
		\node[element] 	(1-13-4) at (0+0.07,-2.5+15/10) {};
		\node[element] 	(1-13-5) at (0-0.07,-2.5+15/10) {};
		\node[element] 	(1-13-6) at (0,-2.5+20/10) {};
		\node[element] 	(1-13-7) at (0+0.07,-2.5+25/10) {};
		\node[element] 	(1-13-8) at (0-0.07,-2.5+25/10) {};
		\node[element] 	(1-13-9) at (0,-2.5+30/10) {};
		\node[element] 	(1-13-10) at (0+0.07,-2.5+35/10) {};
		\node[element] 	(1-13-11) at (0-0.07,-2.5+35/10) {};
		\node[element] 	(1-13-12) at (0,-2.5+40/10) {};
		\node[element] 	(1-13-13) at (0+0.07,-2.5+45/10) {};
		\node[element] 	(1-13-14) at (0-0.07,-2.5+45/10) {};
		
		\node[set] 		(2-1) at (2.5-0.2,-2.5) {};
		\node[set] 		(2-2) at (2.5+0.2,-2.5) {};
		\node[set] 		(2-3) at (2.5,-2.5+0.35){};
		
		\node[element] 		(3) at (5,-2.5) {};
		
		\node[element] 		(3-2-1) at (5-2.3/6,-2.5+0.07) {};
		\node[element] 		(3-2-2) at (5-2.3/6,-2.5-0.07) {};
		\node[element] 		(3-2-3) at (5-4.6/6,-2.5) {};
		\node[element] 		(3-2-4) at (5-6.9/6,-2.5+0.07) {};
		\node[element] 		(3-2-5) at (5-6.9/6,-2.5-0.07) {};
		\node[element] 		(3-2-6) at (5-9.2/6,-2.5) {};
		\node[element] 		(3-2-7) at (5-11.5/6,-2.5+0.07) {};
		\node[element] 		(3-2-8) at (5-11.5/6,-2.5-0.07) {};
		
		\node[element] 		(3-4-1) at (5+0.07,-2.5-2.15/6) {};
		\node[element] 		(3-4-2) at (5-0.07,-2.5-2.15/6) {};
		\node[element] 		(3-4-3) at (5,-2.5-4.3/6) {};
		\node[element] 		(3-4-4) at (5+0.07,-2.5-6.45/6) {};
		\node[element] 		(3-4-5) at (5-0.07,-2.5-6.45/6) {};
		\node[element] 		(3-4-6) at (5,-2.5-8.6/6) {};
		\node[element] 		(3-4-7) at (5+0.07,-2.5-10.75/6) {};
		\node[element] 		(3-4-8) at (5-0.07,-2.5-10.75/6) {};
		
		\node[element] 		(3-15-1) at (5+0.5,-2.5+0.07) {};
		\node[element] 		(3-15-2) at (5+0.5,-2.5-0.07) {};
		\node[element] 		(3-15-3) at (5+1,-2.5) {};
		\node[element] 		(3-15-4) at (5+1.5,-2.5+0.07) {};
		\node[element] 		(3-15-5) at (5+1.5,-2.5-0.07) {};
		\node[element] 		(3-15-6) at (5+2,-2.5) {};
		\node[element] 		(3-15-7) at (5+2.5,-2.5+0.07) {};
		\node[element] 		(3-15-8) at (5+2.5,-2.5-0.07) {};
		\node[element] 		(3-15-9) at (5+3,-2.5) {};
		\node[element] 		(3-15-10) at (5+3.5,-2.5+0.07) {};
		\node[element] 		(3-15-11) at (5+3.5,-2.5-0.07) {};
		\node[element] 		(3-15-12) at (5+4,-2.5) {};
		\node[element] 		(3-15-13) at (5+4.5,-2.5+0.07) {};
		\node[element] 		(3-15-14) at (5+4.5,-2.5-0.07) {};
		
		\node[set] 	 (4-1) at (5-0.2,-5) {};
		\node[set] 	 (4-2) at (5+0.2,-5) {};
		\node[set] 	 (4-3) at (5,-5+0.35) {};
		
		\node[element] 		(5) at (2.5,0) {};
		
		\node[element] 		(5-2-1) at (2.5+0.07,-2.15/6) {};
		\node[element] 		(5-2-2) at (2.5-0.07,-2.15/6) {};
		\node[element] 		(5-2-3) at (2.5,-4.3/6) {};
		\node[element] 		(5-2-4) at (2.5+0.07,-6.45/6) {};
		\node[element] 		(5-2-5) at (2.5-0.07,-6.45/6) {};
		\node[element] 		(5-2-6) at (2.5,-8.6/6) {};
		\node[element] 		(5-2-7) at (2.5+0.07,-10.75/6) {};
		\node[element] 		(5-2-8) at (2.5-0.07,-10.75/6) {};
		
		\node[element] 		(5-6-1) at (2.5+0.07,2.15/6) {};
		\node[element] 		(5-6-2) at (2.5-0.07,2.15/6) {};
		\node[element] 		(5-6-3) at (2.5,4.3/6) {};
		\node[element] 		(5-6-4) at (2.5+0.07,6.45/6) {};
		\node[element] 		(5-6-5) at (2.5-0.07,6.45/6) {};
		\node[element] 		(5-6-6) at (2.5,8.6/6) {};
		\node[element] 		(5-6-7) at (2.5+0.07,10.75/6) {};
		\node[element] 		(5-6-8) at (2.5-0.07,10.75/6) {};
		
		\node[element] 		(5-7-1) at (2.5+2.3/6,0.07) {};
		\node[element] 		(5-7-2) at (2.5+2.3/6,-0.07) {};
		\node[element] 		(5-7-3) at (2.5+4.6/6,0) {};
		\node[element] 		(5-7-4) at (2.5+6.9/6,0.07) {};
		\node[element] 		(5-7-5) at (2.5+6.9/6,-0.07) {};
		\node[element] 		(5-7-6) at (2.5+9.2/6,0) {};
		\node[element] 		(5-7-7) at (2.5+11.5/6,0.07) {};
		\node[element] 		(5-7-8) at (2.5+11.5/6,-0.07) {};
		
		\node[set]		(6-1) at (2.5-0.2,2.5) {};
		\node[set]		(6-2) at (2.5+0.2,2.5) {};
		\node[set]		(6-3) at (2.5,2.5-0.35) {};
		
		\node[set]		(7-1) at (5-0.2,0) {};
		\node[set]		(7-2) at (5+0.2,0) {};
		\node[set]		(7-3) at (5,0+0.35) {};
		
		\node[element]		(8) at (5,2.5) {};
		
		\node[element]		(8-7-1) at (5+0.07,2.5-2.15/6) {};
		\node[element]		(8-7-2) at (5-0.07,2.5-2.15/6) {};
		\node[element]		(8-7-3) at (5,2.5-4.3/6) {};
		\node[element]		(8-7-4) at (5+0.07,2.5-6.45/6) {};
		\node[element]		(8-7-5) at (5-0.07,2.5-6.45/6) {};
		\node[element]		(8-7-6) at (5,2.5-8.6/6) {};
		\node[element]		(8-7-7) at (5+0.07,2.5-10.75/6) {};
		\node[element]		(8-7-8) at (5-0.07,2.5-10.75/6) {};
		
		\node[element]		(8-6-1) at (5-2.3/6,2.5+0.07) {};
		\node[element]		(8-6-2) at (5-2.3/6,2.5-0.07) {};
		\node[element]		(8-6-3) at (5-4.6/6,2.5) {};
		\node[element]		(8-6-4) at (5-6.9/6,2.5+0.07) {};
		\node[element]		(8-6-5) at (5-6.9/6,2.5-0.07) {};
		\node[element]		(8-6-6) at (5-9.2/6,2.5) {};
		\node[element]		(8-6-7) at (5-11.5/6,2.5+0.07) {};
		\node[element]		(8-6-8) at (5-11.5/6,2.5-0.07) {};
		
		\node[element]		(8-10-1) at (5+2.3/6,2.5+0.07) {};
		\node[element]		(8-10-2) at (5+2.3/6,2.5-0.07) {};
		\node[element]		(8-10-3) at (5+4.6/6,2.5) {};
		\node[element]		(8-10-4) at (5+6.9/6,2.5+0.07) {};
		\node[element]		(8-10-5) at (5+6.9/6,2.5-0.07) {};
		\node[element]		(8-10-6) at (5+9.1/6,2.5) {};
		\node[element]		(8-10-7) at (5+11.4/6,2.5+0.07) {};
		\node[element]		(8-10-8) at (5+11.4/6,2.5-0.07) {};
		
		\node[element]		(9) at (7.5,0) {};
		
		\node[element]		(9-7-1) at (7.5-2.3/6,0.07) {};
		\node[element]		(9-7-2) at (7.5-2.3/6,-0.07) {};
		\node[element]		(9-7-3) at (7.5-4.6/6,0) {};
		\node[element]		(9-7-4) at (7.5-6.9/6,0.07) {};
		\node[element]		(9-7-5) at (7.5-6.9/6,-0.07) {};
		\node[element]		(9-7-6) at (7.5-9.2/6,0) {};
		\node[element]		(9-7-7) at (7.5-11.5/6,0.07) {};
		\node[element]		(9-7-8) at (7.5-11.5/6,-0.07) {};
		
		\node[element]		(9-10-1) at (7.5+0.07,2.15/6) {};
		\node[element]		(9-10-2) at (7.5-0.07,2.15/6) {};
		\node[element]		(9-10-3) at (7.5,4.3/6) {};
		\node[element]		(9-10-4) at (7.5+0.07,6.45/6) {};
		\node[element]		(9-10-5) at (7.5-0.07,6.45/6) {};
		\node[element]		(9-10-6) at (7.5,8.6/6) {};
		\node[element]		(9-10-7) at (7.5+0.07,10.75/6) {};
		\node[element]		(9-10-8) at (7.5-0.07,10.75/6) {};
		
		\node[element]		(9-11-1) at (7.5+2.3/6,0.07) {};
		\node[element]		(9-11-2) at (7.5+2.3/6,-0.07) {};
		\node[element]		(9-11-3) at (7.5+4.6/6,0) {};
		\node[element]		(9-11-4) at (7.5+6.9/6,0.07) {};
		\node[element]		(9-11-5) at (7.5+6.9/6,-0.07) {};
		\node[element]		(9-11-6) at (7.5+9.2/6,0) {};
		\node[element]		(9-11-7) at (7.5+11.5/6,0.07) {};
		\node[element]		(9-11-8) at (7.5+11.5/6,-0.07) {};
		
		\node[set]		(10-1) at (7.5-0.2,2.5) {};
		\node[set]		(10-2) at (7.5+0.2,2.5) {};
		\node[set]		(10-3) at (7.5,2.5-0.35){};
		
		\node[set]		(11-1) at (10-0.2,0) {};
		\node[set]		(11-2) at (10+0.2,0) {};
		\node[set]		(11-3) at (10,-0.35) {};
		
		\node[element]		(12) at (12.5,0) {};
		
		\node[element]		(12-11-1) at (12.5-2.3/6,0.07) {};
		\node[element]		(12-11-2) at (12.5-2.3/6,-0.07) {};
		\node[element]		(12-11-3) at (12.5-4.6/6,0) {};
		\node[element]		(12-11-4) at (12.5-6.9/6,0.07) {};
		\node[element]		(12-11-5) at (12.5-6.9/6,-0.07) {};
		\node[element]		(12-11-6) at (12.5-9.2/6,0) {};
		\node[element]		(12-11-7) at (12.5-11.5/6,0.07) {};
		\node[element]		(12-11-8) at (12.5-11.5/6,-0.07) {};
		
		\node[element]		(12-16-1) at (12.5+0.07,-5/12) {};
		\node[element]		(12-16-2) at (12.5-0.07,-5/12) {};
		\node[element]		(12-16-3) at (12.5,-5*2/12) {};
		\node[element]		(12-16-4) at (12.5+0.07,-5*3/12) {};
		\node[element]		(12-16-5) at (12.5-0.07,-5*3/12) {};
		\node[element]		(12-16-6) at (12.5,-5*4/12) {};
		\node[element]		(12-16-7) at (12.5+0.07,-5*5/12) {};
		\node[element]		(12-16-8) at (12.5-0.07,-5*5/12) {};
		\node[element]		(12-16-9) at (12.5,-5*6/12) {};
		\node[element]		(12-16-10) at (12.5+0.07,-5*7/12) {};
		\node[element]		(12-16-11) at (12.5-0.07,-5*7/12) {};
		\node[element]		(12-16-12) at (12.5,-5*8/12) {};
		\node[element]		(12-16-13) at (12.5+0.07,-5*9/12) {};
		\node[element]		(12-16-14) at (12.5-0.07,-5*9/12) {};
		\node[element]		(12-16-15) at (12.5,-5*10/12) {};
		\node[element]		(12-16-16) at (12.5+0.07,-5*11/12) {};
		\node[element]		(12-16-17) at (12.5-0.07,-5*11/12) {};
		
		\node[element]		(12-17-1) at (12.5+0.07,2.5/6) {};
		\node[element]		(12-17-2) at (12.5-0.07,2.5/6) {};
		\node[element]		(12-17-3) at (12.5,2.5*2/6) {};
		\node[element]		(12-17-4) at (12.5+0.07,2.5*3/6) {};
		\node[element]		(12-17-5) at (12.5-0.07,2.5*3/6) {};
		\node[element]		(12-17-6) at (12.5,2.5*4/6) {};
		\node[element]		(12-17-7) at (12.5+0.07,2.5*5/6) {};
		\node[element]		(12-17-8) at (12.5-0.07,2.5*5/6) {};
		
		\node[element]		(13) at (0,2.5) {};
		
		\node[element]		(13-6-1) at (2.3/6,2.5+0.07) {};
		\node[element]		(13-6-2) at (2.3/6,2.5-0.07) {};
		\node[element]		(13-6-3) at (2.3*2/6,2.5) {};
		\node[element]		(13-6-4) at (2.3*3/6,2.5+0.07) {};
		\node[element]		(13-6-5) at (2.3*3/6,2.5-0.07) {};
		\node[element]		(13-6-6) at (2.3*4/6,2.5) {};
		\node[element]		(13-6-7) at (2.3*5/6,2.5+0.07) {};
		\node[element]		(13-6-8) at (2.3*5/6,2.5-0.07) {};
		
		\node[element]		(14) at (0,-5) {};
		
		\node[element]		(14-4-1) at (4.8/12,-5+0.07) {};
		\node[element]		(14-4-2) at (4.8/12,-5-0.07) {};
		\node[element]		(14-4-3) at (4.8*2/12,-5) {};
		\node[element]		(14-4-4) at (4.8*3/12,-5+0.07) {};
		\node[element]		(14-4-5) at (4.8*3/12,-5-0.07) {};
		\node[element]		(14-4-6) at (4.8*4/12,-5) {};
		\node[element]		(14-4-7) at (4.8*5/12,-5+0.07) {};
		\node[element]		(14-4-8) at (4.8*5/12,-5-0.07) {};
		\node[element]		(14-4-9) at (4.8*6/12,-5) {};
		\node[element]		(14-4-10) at (4.8*7/12,-5+0.07) {};
		\node[element]		(14-4-11) at (4.8*7/12,-5-0.07) {};
		\node[element]		(14-4-12) at (4.8*8/12,-5) {};
		\node[element]		(14-4-13) at (4.8*9/12,-5+0.07) {};
		\node[element]		(14-4-14) at (4.8*9/12,-5-0.07) {};
		\node[element]		(14-4-15) at (4.8*10/12,-5) {};
		\node[element]		(14-4-16) at (4.8*11/12,-5+0.07) {};
		\node[element]		(14-4-17) at (4.8*11/12,-5-0.07) {};
		
		\node[element]		(15) at (10,-2.5) {};
		
		\node[element]		(15-11-1) at (10.07,-2.5+2.15/6) {};
		\node[element]		(15-11-2) at (10-0.07,-2.5+2.15/6) {};
		\node[element]		(15-11-3) at (10,-2.5+2.15*2/6) {};
		\node[element]		(15-11-4) at (10.07,-2.5+2.15*3/6) {};
		\node[element]		(15-11-5) at (10-0.07,-2.5+2.15*3/6) {};
		\node[element]		(15-11-6) at (10,-2.5+2.15*4/6) {};
		\node[element]		(15-11-7) at (10.07,-2.5+2.15*5/6) {};
		\node[element]		(15-11-8) at (10-0.07,-2.5+2.15*5/6) {};
		
		\node[element]		(16) at (12.5,-5) {};
		
		\node[element]		(16-4-1) at (12.5-7.3/18,-5+0.07) {};
		\node[element]		(16-4-2) at (12.5-7.3/18,-5-0.07) {};
		\node[element]		(16-4-3) at (12.5-7.3*2/18,-5) {};
		\node[element]		(16-4-4) at (12.5-7.3*3/18,-5+0.07) {};
		\node[element]		(16-4-5) at (12.5-7.3*3/18,-5-0.07) {};
		\node[element]		(16-4-6) at (12.5-7.3*4/18,-5) {};
		\node[element]		(16-4-7) at (12.5-7.3*5/18,-5+0.07) {};
		\node[element]		(16-4-8) at (12.5-7.3*5/18,-5-0.07) {};
		\node[element]		(16-4-9) at (12.5-7.3*6/18,-5) {};
		\node[element]		(16-4-10) at (12.5-7.3*7/18,-5+0.07) {};
		\node[element]		(16-4-11) at (12.5-7.3*7/18,-5-0.07) {};
		\node[element]		(16-4-12) at (12.5-7.3*8/18,-5) {};
		\node[element]		(16-4-13) at (12.5-7.3*9/18,-5+0.07) {};
		\node[element]		(16-4-14) at (12.5-7.3*9/18,-5-0.07) {};
		\node[element]		(16-4-15) at (12.5-7.3*10/18,-5) {};
		\node[element]		(16-4-16) at (12.5-7.3*11/18,-5+0.07) {};
		\node[element]		(16-4-17) at (12.5-7.3*11/18,-5-0.07) {};
		\node[element]		(16-4-18) at (12.5-7.3*12/18,-5) {};
		\node[element]		(16-4-19) at (12.5-7.3*13/18,-5+0.07) {};
		\node[element]		(16-4-20) at (12.5-7.3*13/18,-5-0.07) {};
		\node[element]		(16-4-21) at (12.5-7.3*14/18,-5) {};
		\node[element]		(16-4-22) at (12.5-7.3*15/18,-5+0.07) {};
		\node[element]		(16-4-23) at (12.5-7.3*15/18,-5-0.07) {};
		\node[element]		(16-4-24) at (12.5-7.3*16/18,-5) {};
		\node[element]		(16-4-25) at (12.5-7.3*17/18,-5+0.07) {};
		\node[element]		(16-4-26) at (12.5-7.3*17/18,-5-0.07) {};
		
		\node[element]		(17) at (12.5,2.5) {};
		
		\node[element]		(17-10-1) at (12.5-4.8/12,2.5+0.07) {};
		\node[element]		(17-10-2) at (12.5-4.8/12,2.5-0.07) {};
		\node[element]		(17-10-3) at (12.5-4.8*2/12,2.5) {};
		\node[element]		(17-10-4) at (12.5-4.8*3/12,2.5+0.07) {};
		\node[element]		(17-10-5) at (12.5-4.8*3/12,2.5-0.07) {};
		\node[element]		(17-10-6) at (12.5-4.8*4/12,2.5) {};
		\node[element]		(17-10-7) at (12.5-4.8*5/12,2.5+0.07) {};
		\node[element]		(17-10-8) at (12.5-4.8*5/12,2.5-0.07) {};
		\node[element]		(17-10-9) at (12.5-4.8*6/12,2.5) {};
		\node[element]		(17-10-10) at (12.5-4.8*7/12,2.5+0.07) {};
		\node[element]		(17-10-11) at (12.5-4.8*7/12,2.5-0.07) {};
		\node[element]		(17-10-12) at (12.5-4.8*8/12,2.5) {};
		\node[element]		(17-10-13) at (12.5-4.8*9/12,2.5+0.07) {};
		\node[element]		(17-10-14) at (12.5-4.8*9/12,2.5-0.07) {};
		\node[element]		(17-10-15) at (12.5-4.8*10/12,2.5) {};
		\node[element]		(17-10-16) at (12.5-4.8*11/12,2.5+0.07) {};
		\node[element]		(17-10-17) at (12.5-4.8*11/12,2.5-0.07) {};
		

		\begin{pgfonlayer}{background}
			\path[expand bubble]plot [smooth cycle,tension=1] coordinates {(12-17-1) (12-17-2) (12-17-3)};
			\path[expand bubble]plot [smooth cycle,tension=1] coordinates {(12-17-4) (12-17-5) (12-17-6)};
			\path[expand bubble]plot [smooth cycle,tension=1] coordinates {(12-17-7) (12-17-8) (17)};
			\path[expand bubble]plot [smooth cycle,tension=1] coordinates {(17-10-1) (17-10-2) (17-10-3)};
			\path[expand bubble]plot [smooth cycle,tension=1] coordinates {(17-10-4) (17-10-5) (17-10-6)};
			\path[expand bubble]plot [smooth cycle,tension=1] coordinates {(17-10-7) (17-10-8) (17-10-9)};
			\path[expand bubble]plot [smooth cycle,tension=1] coordinates {(17-10-10) (17-10-11) (17-10-12)};
			\path[expand bubble]plot [smooth cycle,tension=1] coordinates {(17-10-13) (17-10-14) (17-10-15)};
			\path[expand bubble]plot [smooth cycle,tension=1] coordinates {(17-10-16) (17-10-17) (10-2)};
			
			\path[expand bubble]plot [smooth cycle,tension=1] coordinates {(12-16-1) (12-16-2) (12-16-3)};
			\path[expand bubble]plot [smooth cycle,tension=1] coordinates {(12-16-4) (12-16-5) (12-16-6)};
			\path[expand bubble]plot [smooth cycle,tension=1] coordinates {(12-16-7) (12-16-8) (12-16-9)};
			\path[expand bubble]plot [smooth cycle,tension=1] coordinates {(12-16-10) (12-16-11) (12-16-12)};
			\path[expand bubble]plot [smooth cycle,tension=1] coordinates {(12-16-13) (12-16-14) (12-16-15)};
			\path[expand bubble]plot [smooth cycle,tension=1] coordinates {(12-16-16) (12-16-17) (16)};
			\path[expand bubble]plot [smooth cycle,tension=1] coordinates {(16-4-1) (16-4-2) (16-4-3)};
			\path[expand bubble]plot [smooth cycle,tension=1] coordinates {(16-4-4) (16-4-5) (16-4-6)};
			\path[expand bubble]plot [smooth cycle,tension=1] coordinates {(16-4-7) (16-4-8) (16-4-9)};
			\path[expand bubble]plot [smooth cycle,tension=1] coordinates {(16-4-10) (16-4-11) (16-4-12)};
			\path[expand bubble]plot [smooth cycle,tension=1] coordinates {(16-4-13) (16-4-14) (16-4-15)};
			\path[expand bubble]plot [smooth cycle,tension=1] coordinates {(16-4-16) (16-4-17) (16-4-18)};
			\path[expand bubble]plot [smooth cycle,tension=1] coordinates {(16-4-19) (16-4-20) (16-4-21)};
			\path[expand bubble]plot [smooth cycle,tension=1] coordinates {(16-4-22) (16-4-23) (16-4-24)};
			\path[expand bubble]plot [smooth cycle,tension=1] coordinates {(16-4-25) (16-4-26) (4-2)};
			
			\path[expand bubble]plot [smooth cycle,tension=1] coordinates {(12-11-1) (12-11-2) (12-11-3)};
			\path[expand bubble]plot [smooth cycle,tension=1] coordinates {(12-11-4) (12-11-5) (12-11-6)};
			\path[expand bubble]plot [smooth cycle,tension=1] coordinates {(12-11-7) (12-11-8) (11-2)};
			
			\path[expand bubble]plot [smooth cycle,tension=1] coordinates {(9-11-1) (9-11-2) (9-11-3)};
			\path[expand bubble]plot [smooth cycle,tension=1] coordinates {(9-11-4) (9-11-5) (9-11-6)};
			\path[expand bubble]plot [smooth cycle,tension=1] coordinates {(9-11-7) (9-11-8) (11-1)};
			
			\path[expand bubble]plot [smooth cycle,tension=1] coordinates {(9-10-1) (9-10-2) (9-10-3)};
			\path[expand bubble]plot [smooth cycle,tension=1] coordinates {(9-10-4) (9-10-5) (9-10-6)};
			\path[expand bubble]plot [smooth cycle,tension=1] coordinates {(9-10-7) (9-10-8) (10-3)};
			
			\path[expand bubble]plot [smooth cycle,tension=1] coordinates {(9-7-1) (9-7-2) (9-7-3)};
			\path[expand bubble]plot [smooth cycle,tension=1] coordinates {(9-7-4) (9-7-5) (9-7-6)};
			\path[expand bubble]plot [smooth cycle,tension=1] coordinates {(9-7-7) (9-7-8) (7-2)};
			
			\path[expand bubble]plot [smooth cycle,tension=1] coordinates {(8-10-1) (8-10-2) (8-10-3)};
			\path[expand bubble]plot [smooth cycle,tension=1] coordinates {(8-10-4) (8-10-5) (8-10-6)};
			\path[expand bubble]plot [smooth cycle,tension=1] coordinates {(8-10-7) (8-10-8) (10-1)};
			
			\path[expand bubble]plot [smooth cycle,tension=1] coordinates {(8-6-1) (8-6-2) (8-6-3)};
			\path[expand bubble]plot [smooth cycle,tension=1] coordinates {(8-6-4) (8-6-5) (8-6-6)};
			\path[expand bubble]plot [smooth cycle,tension=1] coordinates {(8-6-7) (8-6-8) (6-2)};
			
			\path[expand bubble]plot [smooth cycle,tension=1] coordinates {(8-7-1) (8-7-2) (8-7-3)};
			\path[expand bubble]plot [smooth cycle,tension=1] coordinates {(8-7-4) (8-7-5) (8-7-6)};
			\path[expand bubble]plot [smooth cycle,tension=1] coordinates {(8-7-7) (8-7-8) (7-3)};
			
			\path[expand bubble]plot [smooth cycle,tension=1] coordinates {(5-7-1) (5-7-2) (5-7-3)};
			\path[expand bubble]plot [smooth cycle,tension=1] coordinates {(5-7-4) (5-7-5) (5-7-6)};
			\path[expand bubble]plot [smooth cycle,tension=1] coordinates {(5-7-7) (5-7-8) (7-1)};
			
			\path[expand bubble]plot [smooth cycle,tension=1] coordinates {(5-6-1) (5-6-2) (5-6-3)};
			\path[expand bubble]plot [smooth cycle,tension=1] coordinates {(5-6-4) (5-6-5) (5-6-6)};
			\path[expand bubble]plot [smooth cycle,tension=1] coordinates {(5-6-7) (5-6-8) (6-3)};
			
			\path[expand bubble]plot [smooth cycle,tension=1] coordinates {(5-2-1) (5-2-2) (5-2-3)};
			\path[expand bubble]plot [smooth cycle,tension=1] coordinates {(5-2-4) (5-2-5) (5-2-6)};
			\path[expand bubble]plot [smooth cycle,tension=1] coordinates {(5-2-7) (5-2-8) (2-3)};
			
			\path[expand bubble]plot [smooth cycle,tension=1] coordinates {(3-15-1) (3-15-2) (3-15-3)};
			\path[expand bubble]plot [smooth cycle,tension=1] coordinates {(3-15-4) (3-15-5) (3-15-6)};
			\path[expand bubble]plot [smooth cycle,tension=1] coordinates {(3-15-7) (3-15-8) (3-15-9)};
			\path[expand bubble]plot [smooth cycle,tension=1] coordinates {(3-15-10) (3-15-11) (3-15-12)};
			\path[expand bubble]plot [smooth cycle,tension=1] coordinates {(3-15-13) (3-15-14) (15)};
			\path[expand bubble]plot [smooth cycle,tension=1] coordinates {(15-11-1) (15-11-2) (15-11-3)};
			\path[expand bubble]plot [smooth cycle,tension=1] coordinates {(15-11-4) (15-11-5) (15-11-6)};
			\path[expand bubble]plot [smooth cycle,tension=1] coordinates {(15-11-7) (15-11-8) (11-3)};
			
			\path[expand bubble]plot [smooth cycle,tension=1] coordinates {(3-4-1) (3-4-2) (3-4-3)};
			\path[expand bubble]plot [smooth cycle,tension=1] coordinates {(3-4-4) (3-4-5) (3-4-6)};
			\path[expand bubble]plot [smooth cycle,tension=1] coordinates {(3-4-7) (3-4-8) (4-3)};
			
			\path[expand bubble]plot [smooth cycle,tension=1] coordinates {(3-2-1) (3-2-2) (3-2-3)};
			\path[expand bubble]plot [smooth cycle,tension=1] coordinates {(3-2-4) (3-2-5) (3-2-6)};
			\path[expand bubble]plot [smooth cycle,tension=1] coordinates {(3-2-7) (3-2-8) (2-2)};
			
			\path[expand bubble]plot [smooth cycle,tension=1] coordinates {(1-2-3) (1-2-1) (1-2-2)};
			\path[expand bubble]plot [smooth cycle,tension=1.1] coordinates {(1-2-4) (1-2-5) (1-2-6)};
			\path[expand bubble]plot [smooth cycle,tension=1.1] coordinates {(1-2-7) (1-2-8) (2-1)};
			
			\path[expand bubble]plot [smooth cycle,tension=1] coordinates {(1-14-1) (1-14-2) (1-14-3)};
			\path[expand bubble]plot [smooth cycle,tension=1] coordinates {(1-14-4) (1-14-5) (1-14-6)};
			\path[expand bubble]plot [smooth cycle,tension=1] coordinates {(1-14-7) (1-14-8) (14)};
			\path[expand bubble]plot [smooth cycle,tension=1] coordinates {(14-4-1) (14-4-2) (14-4-3)};
			\path[expand bubble]plot [smooth cycle,tension=1] coordinates {(14-4-4) (14-4-5) (14-4-6)};
			\path[expand bubble]plot [smooth cycle,tension=1] coordinates {(14-4-7) (14-4-8) (14-4-9)};
			\path[expand bubble]plot [smooth cycle,tension=1] coordinates {(14-4-10) (14-4-11) (14-4-12)};
			\path[expand bubble]plot [smooth cycle,tension=1] coordinates {(14-4-13) (14-4-14) (14-4-15)};
			\path[expand bubble]plot [smooth cycle,tension=1] coordinates {(14-4-16) (14-4-17) (4-1)};
			
			\path[expand bubble]plot [smooth cycle,tension=1] coordinates {(1-13-1) (1-13-2) (1-13-3)};
			\path[expand bubble]plot [smooth cycle,tension=1] coordinates {(1-13-4) (1-13-5) (1-13-6)};
			\path[expand bubble]plot [smooth cycle,tension=1] coordinates {(1-13-7) (1-13-8) (1-13-9)};
			\path[expand bubble]plot [smooth cycle,tension=1] coordinates {(1-13-10) (1-13-11) (1-13-12)};
			\path[expand bubble]plot [smooth cycle,tension=1] coordinates {(1-13-13) (1-13-14) (13)};
			\path[expand bubble]plot [smooth cycle,tension=1] coordinates {(13-6-1) (13-6-2) (13-6-3)};
			\path[expand bubble]plot [smooth cycle,tension=1] coordinates {(13-6-4) (13-6-5) (13-6-6)};
			\path[expand bubble]plot [smooth cycle,tension=1] coordinates {(13-6-7) (13-6-8) (6-1)};
		\end{pgfonlayer}
	\end{tikzpicture}
}

%% file: sections/figures/ppascendingv2.tex
\resizebox{\textwidth}{!}{%
	\begin{tikzpicture}[-,auto, semithick]
		\tikzset{
			element/.style={
				fill=black,draw=black,text=black,shape=circle,inner sep=1pt,minimum size=1pt
			},
			set/.style={
				fill=white,draw=white,text=black,shape=circle
			},halo/.style={line join=round,
			double,line cap=round,double distance=#1,shorten >=-#1/2,shorten <=-#1/2},
			halo/.default=0.7em
		}
		
		\node[element, label= right:{$l_i$}] 	(4) at (0,6) {};
		\node[element, label= below:{$u_i$}] 	(1) at (10,0) {};
		
		\node[element, label= below:{$l'_i$}] 	(3) at (0,2) {};
		\node[element, label= above:{$u'_i$}] 	(2) at (10,2) {};
		
		\node[element] (1-1) at (10.1,1/3) {};
		\node[element] (1-2) at (10.2,2/3) {};
		\node[element] (1-3) at (10.3,3/3) {};
		\node[element] (1-4) at (10.2,4/3) {};
		\node[element] (1-5) at (10.1,5/3) {};
		
		\node[element] (2-23) at (5/12, 2-0.1) {};
		\node[element] (2-22) at (10/12,2-0.2) {};
		\node[element] (2-21) at (15/12,2-0.3) {};
		\node[element] (2-20) at (20/12,2-0.4) {};
		\node[element] (2-19) at (25/12,2-0.5) {};
		\node[element] (2-18) at (30/12,2-0.6) {};
		\node[element] (2-17) at (35/12,2-0.7) {};
		\node[element] (2-16) at (40/12,2-0.8) {};
		\node[element] (2-15) at (45/12,2-0.9) {};
		\node[element] (2-14) at (50/12,2-1) {};
		\node[element] (2-13) at (55/12,2-1.1) {};
		\node[element] (2-12) at (60/12,2-1.2) {};
		\node[element] (2-11) at (65/12,2-1.1) {};
		\node[element] (2-10) at (70/12,2-1) {};
		\node[element] (2-9) at (75/12,2-0.9) {};
		\node[element] (2-8) at (80/12,2-0.8) {};
		\node[element] (2-7) at (85/12,2-0.7) {};
		\node[element] (2-6) at (90/12,2-0.6) {};
		\node[element] (2-5) at (95/12,2-0.5) {};
		\node[element] (2-4) at (100/12,2-0.4) {};
		\node[element] (2-3) at (105/12,2-0.3) {};
		\node[element] (2-2) at (110/12,2-0.2) {};
		\node[element] (2-1) at (115/12,2-0.1) {};
		
		\node[element] (3-1) at (0-0.1,2.4) {};
		\node[element] (3-2) at (0-0.2,2.8) {};
		\node[element] (3-3) at (0-0.3,3.2) {};
		\node[element] (3-4) at (0-0.4,3.6) {};
		\node[element] (3-5) at (0-0.5,4) {};
		\node[element] (3-6) at (0-0.4,4.4) {};
		\node[element] (3-7) at (0-0.3,4.8) {};
		\node[element] (3-8) at (0-0.2,5.2) {};
		\node[element] (3-9) at (0-0.1,5.6) {};
		
		\draw[-,ultra thick] (-1.2,-0.5)--(10.2,-0.5) node[right]{$xy$-plane};
		\draw[->,ultra thick] (-1.2,-0.5)--(-1.2,6.1) node[above]{$z$};
		
		\begin{scope}[on background layer]     
			\draw[halo] (1) -- (1-1);
			\draw[halo] (1-2) -- (1-3);
			\draw[halo] (1-4) -- (1-5);
			
			\draw[halo] (2) -- (2-1);
			\draw[halo] (2-2) -- (2-3);
			\draw[halo] (2-4) -- (2-5);
			\draw[halo] (2-6) -- (2-7);
			\draw[halo] (2-8) -- (2-9);
			\draw[halo] (2-10) -- (2-11);
			\draw[halo] (2-12) -- (2-13);
			\draw[halo] (2-14) -- (2-15);
			\draw[halo] (2-16) -- (2-17);
			\draw[halo] (2-18) -- (2-19);
			\draw[halo] (2-20) -- (2-21);
			\draw[halo] (2-22) -- (2-23);
			
			\draw[halo] (3) -- (3-1);
			\draw[halo] (3-2) -- (3-3);
			\draw[halo] (3-4) -- (3-5);
			\draw[halo] (3-6) -- (3-7);
			\draw[halo] (3-8) -- (3-9);
		\end{scope}
	\end{tikzpicture}
}

%% file: sections/figures/snowflakeppoutcome.tex
\resizebox{\textwidth}{!}{%
	\begin{tikzpicture}[-,auto, semithick]
		\tikzset{
			element/.style={
				fill=black,draw=black,text=black,shape=circle,inner sep=4pt,minimum size=4pt
			},
			set/.style={
				fill=white,draw=white,text=black,shape=circle
			}, 
			expand bubble/.style={
				preaction={draw,line width=1.8em},
				white,fill,draw,line width=1.6em,
			},
		}
		
		\node[element] (c1) at (0.,-0.8660254037844386) {};
		\node[element] (c2) at (-0.5,0.) {};
		\node[element] (c3) at (0.5,0.) {};
		
		\node[element] (d2200-d2101-c1) at (-1.,0.8660254037844387) {};
		\node[element] (d2200-d2101-a1) at (-1.777263412660235,1.571972701633057) {};
		\node[element] (d2200-d2101-b1) at (-1.222736587339765,1.8921289135046986) {};
		
		\node[element] (d2200-d2101-c2) at (-2.,2.5980762113533165) {};
		\node[element] (d2200-d2101-a2) at (-2.2227365873397558,3.62417972107358) {};
		\node[element] (d2200-d2101-b2) at (-2.7772634126602433,3.3040235092019286) {};
		
		\node[element] (d2200-d2101-c3) at (-3.,4.330127018922194) {};
		\node[element] (d2200-d2101-a3) at (-3.777263412660237,5.036074316770808) {};
		\node[element] (d2200-d2101-b3) at (-3.2227365873397584,5.356230528642454) {};
		
		\node[element] (d2200-d2101-c4) at (-4.,6.062177826491069) {};
		\node[element] (d2200-d2101-a4) at (-4.2227365873397655,7.088281336211331) {};
		\node[element] (d2200-d2101-b4) at (-4.777263412660233,6.768125124339692) {};
		
		\node[element] (d2200-d2101-c5) at (-5.,7.794228634059952) {};
		\node[element] (d2200-d2101-a5) at (-5.777263412660244,8.500175931908561) {};
		\node[element] (d2200-d2101-b5) at (-5.222736587339754,8.820332143780213) {};
		
		\node[element] (d2200-d2101-c6) at (-6.,9.526279441628823) {};
		\node[element] (d2200-d2101-a6) at (-6.222736587339765,10.552382951349072) {};
		\node[element] (d2200-d2101-b6) at (-6.7772634126602185,10.23222673947744) {};
		
		\node[element] (d2200-d2101-c7) at (-7.,11.258330249197689) {};
		\node[element] (d2200-d2101-a7) at (-7.777263412660218,11.964277547046319) {};
		\node[element] (d2200-d2101-b7) at (-7.22273658733977,12.284433758917947) {};
		
		\node[element] (d2200-d2101-c8) at (-8.,12.990381056766578) {};
		\node[element] (d2200-d2101-a8) at (-8.22273658733973,14.016484566486852) {};
		\node[element] (d2200-d2101-b8) at (-8.777263412660258,13.69632835461518) {};
		
		\node[element] (d2101) at (-9.,14.722431864335459) {};
		\node[element] (d2201) at (-9.500000000003242,15.588457268117999) {};
		\node[element] (d2301) at (-10.,14.722431864331696) {};
		
		\node[element] (d2201-d2102-a1) at (-10.27726341266613,16.294404565963696) {};
		\node[element] (d2201-d2102-b1) at (-9.72273658734688,16.61456077783742) {};
		\node[element] (d2201-d2102-c1) at (-10.50000000000976,17.32050807568309) {};
		
		\node[element] (d2201-d2102-a2) at (-11.277263412708495,18.02645537348929) {};
		\node[element] (d2201-d2102-b2) at (-10.722736587301267,18.346611585413807) {};
		\node[element] (d2201-d2102-c2) at (-11.5,19.05255888322) {};
		
		\node[element] (d2301-d2112-a1) at (-11.000000000001197,15.042588076199516) {};
		\node[element] (d2301-d2112-b1) at (-10.999999999998787,14.402275652456348) {};
		\node[element] (d2301-d2112-c1) at (-12.,14.722431864324168) {};
		
		\node[element] (d2301-d2112-a2) at (-12.999999999984906,15.042588076242883) {};
		\node[element] (d2301-d2112-b2) at (-12.999999999982496,14.40227565239792) {};
		\node[element] (d2301-d2112-c2) at (-13.999999999967393,14.722431864316636) {};
		
		\node[element] (d2310-d2111-c1) at (-1.5,0.) {};
		\node[element] (d2310-d2111-a1) at (-2.5,0.32015621187164245) {};
		\node[element] (d2310-d2111-b1) at (-2.5,-0.3201562118716425) {};
		
		\node[element] (d2310-d2111-c2) at (-3.5,0.) {};
		\node[element] (d2310-d2111-a2) at (-4.5,0.3201562118716411) {};
		\node[element] (d2310-d2111-b2) at (-4.5,-0.3201562118716411) {};
		
		\node[element] (d2310-d2111-c3) at (-5.5,0.) {};
		\node[element] (d2310-d2111-a3) at (-6.5,0.3201562118716411) {};
		\node[element] (d2310-d2111-b3) at (-6.5,-0.3201562118716411) {};
		
		\node[element] (d2310-d2111-c4) at (-7.5,0.) {};
		\node[element] (d2310-d2111-a4) at (-8.5,0.32015621187164384) {};
		\node[element] (d2310-d2111-b4) at (-8.5,-0.3201562118716439) {};
		
		\node[element] (d2310-d2111-c5) at (-9.5,0.) {};
		\node[element] (d2310-d2111-a5) at (-10.5,0.3201562118716328) {};
		\node[element] (d2310-d2111-b5) at (-10.5,-0.3201562118716328) {};
		
		\node[element] (d2310-d2111-c6) at (-11.5,0.) {};
		\node[element] (d2310-d2111-a6) at (-12.5,0.32015621187163) {};
		\node[element] (d2310-d2111-b6) at (-12.5,-0.32015621187163) {};
		
		\node[element] (d2310-d2111-c7) at (-13.5,0.) {};
		\node[element] (d2310-d2111-a7) at (-14.5,0.32015621187163) {};
		\node[element] (d2310-d2111-b7) at (-14.5,-0.32015621187163) {};
		
		\node[element] (d2310-d2111-c8) at (-15.5,0.) {};
		\node[element] (d2310-d2111-a8) at (-16.5,0.3201562118716466) {};
		\node[element] (d2310-d2111-b8) at (-16.5,-0.32015621187164667) {};
		
		\node[element] (d2111) at (-17.5,0.) {};
		\node[element] (d2211) at (-18.000000000000192,0.8660254037843291) {};
		\node[element] (d2311) at (-18.5,0.) {};
		
		\node[element] (d2311-d2132-a1) at (-19.499999999999886,-0.32015621187200327) {};
		\node[element] (d2311-d2132-b1) at (-19.5000000000001,0.3201562118713456) {};
		\node[element] (d2311-d2132-c1) at (-20.5,0.) {};
		
		\node[element] (d2311-d2132-a2) at (-21.499999999999105,0.32015621187367704) {};
		\node[element] (d2311-d2132-b2) at (-21.499999999998963,-0.32015621187543103) {};
		\node[element] (d2311-d2132-c2) at (-22.499999999998103,0.) {};
		
		\node[element] (d2211-d2122-a1) at (-18.777263412660524,1.5719727016327898) {};
		\node[element] (d2211-d2122-b1) at (-18.22273658734022,1.8921289135044983) {};
		\node[element] (d2211-d2122-c1) at (-19.00000000000056,2.598076211352972) {};
		
		\node[element] (d2211-d2122-a2) at (-19.77726341266323,3.3040235091990144) {};
		\node[element] (d2211-d2122-b2) at (-19.22273658733733,3.624179721073954) {};
		\node[element] (d2211-d2122-c2) at (-20.,4.33012701892) {};
		
		\node[element] (d3200-d3101-c1) at (1.5,0.) {};
		\node[element] (d3200-d3101-a1) at (2.5,0.3201562118716411) {};
		\node[element] (d3200-d3101-b1) at (2.5,-0.3201562118716411) {};
		
		\node[element] (d3200-d3101-c2) at (3.5,0.) {};
		\node[element] (d3200-d3101-a2) at (4.5,0.3201562118716411) {};
		\node[element] (d3200-d3101-b2) at (4.5,-0.3201562118716411) {};
		
		\node[element] (d3200-d3101-c3) at (5.5,0.) {};
		\node[element] (d3200-d3101-a3) at (6.5,0.3201562118716411) {};
		\node[element] (d3200-d3101-b3) at (6.5,-0.3201562118716411) {};
		
		\node[element] (d3200-d3101-c4) at (7.5,0.) {};
		\node[element] (d3200-d3101-a4) at (8.5,0.3201562118716494) {};
		\node[element] (d3200-d3101-b4) at (8.5,-0.32015621187164944) {};
		
		\node[element] (d3200-d3101-c5) at (9.5,0.) {};
		\node[element] (d3200-d3101-a5) at (10.5,0.32015621187164384) {};
		\node[element] (d3200-d3101-b5) at (10.5,-0.3201562118716439) {};
		
		\node[element] (d3200-d3101-c6) at (11.5,0.) {};
		\node[element] (d3200-d3101-a6) at (12.5,0.32015621187163) {};
		\node[element] (d3200-d3101-b6) at (12.5,-0.32015621187163) {};
		
		\node[element] (d3200-d3101-c7) at (13.5,0.) {};
		\node[element] (d3200-d3101-a7) at (14.5,0.32015621187163) {};
		\node[element] (d3200-d3101-b7) at (14.5,-0.32015621187163) {};
		
		\node[element] (d3200-d3101-c8) at (15.5,0.) {};
		\node[element] (d3200-d3101-a8) at (16.5,0.3201562118716688) {};
		\node[element] (d3200-d3101-b8) at (16.5,-0.32015621187166887) {};
		
		\node[element] (d3101) at (17.5,0.) {};
		\node[element] (d3301) at (18.,0.8660254037844387) {};
		\node[element] (d3201) at (18.5,0.) {};
		
		\node[element] (d3201-d3102-a1) at (19.5,-0.3201562118716744) {};
		\node[element] (d3201-d3102-b1) at (19.5,0.32015621187167437) {};
		\node[element] (d3201-d3102-c1) at (20.5,0.) {};
		
		\node[element] (d3201-d3102-a2) at (21.5,0.32015621187167437) {};
		\node[element] (d3201-d3102-b2) at (21.5,-0.3201562118716744) {};
		\node[element] (d3201-d3102-c2) at (22.5,0.) {};
		
		\node[element] (d3301-d3112-a1) at (18.22273658733974,1.8921289135047132) {};
		\node[element] (d3301-d3112-b1) at (18.777263412660265,1.5719727016330403) {};
		\node[element] (d3301-d3112-c1) at (19.,2.5980762113533156) {};
		
		\node[element] (d3301-d3112-a2) at (19.777263412660265,3.3040235092019152) {};
		\node[element] (d3301-d3112-b2) at (19.22273658733974,3.624179721073588) {};
		\node[element] (d3301-d3112-c2) at (20.,4.330127018922194) {};
		
		\node[element] (d3310-d3111-c1) at (1.,0.8660254037844386) {};
		\node[element] (d3310-d3111-a1) at (1.2227365873397644,1.8921289135046986) {};
		\node[element] (d3310-d3111-b1) at (1.7772634126602358,1.5719727016330558) {};
		
		\node[element] (d3310-d3111-c2) at (2.,2.5980762113533156) {};
		\node[element] (d3310-d3111-a2) at (2.777263412660237,3.3040235092019326) {};
		\node[element] (d3310-d3111-b2) at (2.222736587339764,3.6241797210735767) {};
		
		\node[element] (d3310-d3111-c3) at (3.,4.330127018922194) {};
		\node[element] (d3310-d3111-a3) at (3.2227365873397615,5.356230528642455) {};
		\node[element] (d3310-d3111-b3) at (3.7772634126602402,5.036074316770809) {};
		
		\node[element] (d3310-d3111-c4) at (4.,6.062177826491069) {};
		\node[element] (d3310-d3111-a4) at (4.777263412660231,6.76812512433969) {};
		\node[element] (d3310-d3111-b4) at (4.222736587339771,7.088281336211326) {};
		
		\node[element] (d3310-d3111-c5) at (5.,7.794228634059945) {};
		\node[element] (d3310-d3111-a5) at (5.777263412660234,8.500175931908561) {};
		\node[element] (d3310-d3111-b5) at (5.222736587339765,8.820332143780202) {};
		
		\node[element] (d3310-d3111-c6) at (6.,9.526279441628821) {};
		\node[element] (d3310-d3111-a6) at (6.222736587339767,10.552382951349085) {};
		\node[element] (d3310-d3111-b6) at (6.777263412660235,10.232226739477444) {};
		
		\node[element] (d3310-d3111-c7) at (7.,11.2583302491977) {};
		\node[element] (d3310-d3111-a7) at (7.777263412660203,11.964277547046333) {};
		\node[element] (d3310-d3111-b7) at (7.222736587339793,12.28443375891794) {};
		
		\node[element] (d3310-d3111-c8) at (8.,12.990381056766577) {};
		\node[element] (d3310-d3111-a8) at (8.222736587339822,14.016484566486817) {};
		\node[element] (d3310-d3111-b8) at (8.777263412660194,13.696328354615233) {};

		\node[element] (d3111) at (9.,14.722431864335459) {};
		\node[element] (d3211) at (10.,14.722431864328366) {};
		\node[element] (d3311) at (9.500000000006144,15.588457268116352) {};
		
		\node[element] (d3211-d3222-a1) at (10.999999999997726,14.402275652449552) {};
		\node[element] (d3211-d3222-b1) at (11.000000000002267,15.042588076192999) {};
		\node[element] (d3211-d3222-c1) at (12.,14.722431864314183) {};
		
		\node[element] (d3211-d3222-a2) at (12.999999999997716,14.402275652435467) {};
		\node[element] (d3211-d3222-b2) at (13.000000000002256,15.042588076178719) {};
		\node[element] (d3211-d3222-c2) at (14.,14.7224318643) {};

		\node[element] (d3211-d3232-a1) at (9.722736587353138,16.614560777835027) {};
		\node[element] (d3211-d3232-b1) at (10.277263412671406,16.29440456595941) {};
		\node[element] (d3211-d3232-c1) at (10.500000000018408,17.320508075678102) {};
		
		\node[element] (d3211-d3232-a2) at (11.277263412683617,18.026455373521273) {};
		\node[element] (d3211-d3232-b2) at (10.722736587365524,18.346611585396786) {};
		\node[element] (d3211-d3232-c2) at (11.50000000003071,19.05255888323992) {};
		
		\node[element] (d1300-d1111-c1) at (0.500101214362474,-1.7319923635428225) {};
		\node[element] (d1300-d1111-a1) at (1.2774471301671986,-2.437848812958505) {};
		\node[element] (d1300-d1111-b1) at (0.7229577272826965,-2.758069833643906) {};
		
		\node[element] (d1300-d1111-c2) at (1.5003036430874208,-3.4639262830595876) {};
		\node[element] (d1300-d1111-a2) at (2.2776495588921404,-4.169782732475276) {};
		\node[element] (d1300-d1111-b2) at (1.7231601560076508,-4.490003753160669) {};
		
		\node[element] (d1300-d1111-c3) at (2.50050607181237,-5.195860202576356) {};
		\node[element] (d1300-d1111-a3) at (3.27785198761708,-5.90171665199205) {};
		\node[element] (d1300-d1111-b3) at (2.7233625847326084,-6.2219376726774325) {};
		
		\node[element] (d1300-d1111-c4) at (3.5007085005373213,-6.92779412209313) {};
		\node[element] (d1300-d1111-a4) at (4.278054416342035,-7.633650571508832) {};
		\node[element] (d1300-d1111-b4) at (3.7235650134575655,-7.953871592194213) {};
		
		\node[element] (d1300-d1111-c5) at (4.500910929262278,-8.659728041609911) {};
		\node[element] (d1300-d1111-b5) at (5.278256845066982,-9.365584491025619) {};
		\node[element] (d1300-d1111-a5) at (4.723767442182533,-9.685805511710988) {};
		
		\node[element] (d1300-d1111-c6) at (5.501113357987235,-10.391661961126694) {};
		\node[element] (d1300-d1111-a6) at (6.2784592737919445,-11.097518410542383) {};
		\node[element] (d1300-d1111-b6) at (5.72396987090747,-11.417739431227767) {};
		
		\node[element] (d1300-d1111-c7) at (6.501315786712181,-12.123595880643457) {};
		\node[element] (d1300-d1111-a7) at (6.724172299632444,-13.149673350744504) {};
		\node[element] (d1300-d1111-b7) at (7.2786617025168505,-12.829452330059159) {};
		
		\node[element] (d1300-d1111-c8) at (7.501518215437113,-13.855529800160198) {};
		\node[element] (d1300-d1111-a8) at (7.729934586918504,-14.880383917212601) {};
		\node[element] (d1300-d1111-b8) at (8.275025729780358,-14.565590449157538) {};

		\node[element] (d1200-d1100-c1) at (-0.4998987788075001,-1.7321092397649804) {};
		\node[element] (d1200-d1100-a1) at (-1.2770796777057918,-2.438147376403253) {};
		\node[element] (d1200-d1100-b1) at (-0.7225154375242089,-2.7582387750877917) {};
		
		\node[element] (d1200-d1100-c2) at (-1.4996963364225004,-3.464276911726065) {};
		\node[element] (d1200-d1100-a2) at (-2.2768772353207907,-4.1703150483643405) {};
		\node[element] (d1200-d1100-b2) at (-1.7223129951392115,-4.490406447048876) {};
		
		\node[element] (d1200-d1100-c3) at (-2.499493894037501,-5.196444583687151) {};
		\node[element] (d1200-d1100-a3) at (-3.276674792935787,-5.902482720325423) {};
		\node[element] (d1200-d1100-b3) at (-2.722110552754211,-6.222574119009958) {};
		
		\node[element] (d1200-d1100-c4) at (-3.4992914516524998,-6.928612255648232) {};
		\node[element] (d1200-d1100-a4) at (-4.27647235055079,-7.63465039228651) {};
		\node[element] (d1200-d1100-b4) at (-3.7219081103692124,-7.954741790971045) {};
		
		\node[element] (d1200-d1100-c5) at (-4.499089009267505,-8.660779927609324) {};
		\node[element] (d1200-d1100-a5) at (-5.276269908165813,-9.366818064247578) {};
		\node[element] (d1200-d1100-b5) at (-4.721705667984187,-9.68690946293214) {};
		
		\node[element] (d1200-d1100-c6) at (-5.498886566882493,-10.392947599570387) {};
		\node[element] (d1200-d1100-a6) at (-6.2760674657807956,-11.09898573620863) {};
		\node[element] (d1200-d1100-b6) at (-5.72150322559917,-11.419077134893193) {};
		
		\node[element] (d1200-d1100-c7) at (-6.498684124497479,-12.125115271531447) {};
		\node[element] (d1200-d1100-a7) at (-7.275865023395801,-12.831153408169708) {};
		\node[element] (d1200-d1100-b7) at (-6.721300783214163,-13.15124480685428) {};
		
		\node[element] (d1200-d1100-c8) at (-7.498481682112478,-13.85728294349253) {};
		\node[element] (d1200-d1100-a8) at (-8.271823196912203,-14.567524382361928) {};
		\node[element] (d1200-d1100-b8) at (-7.726658485200276,-14.882190425430602) {};
		
		\node[element] (d1111) at (8.503442101261754,-15.590444566209953) {};
		\node[element] (d1211) at (8.003098167399914,-16.45627130860572) {};
		\node[element] (d1311) at (9.003098088521511,-16.456668494761615) {};
		
		\node[element] (d1311-d1132-a1) at (9.225427104183659,-17.482860391433327) {};
		\node[element] (d1311-d1132-b1) at (9.780081047378895,-17.16292445519325) {};
		\node[element] (d1311-d1132-c1) at (10.002410063041049,-18.189116351864975) {};
		
		\node[element] (d1311-d1132-a2) at (10.775552145354403,-18.899574880250796) {};
		\node[element] (d1311-d1132-b2) at (10.2303000189465,-19.214087652925347) {};
		\node[element] (d1311-d1132-c2) at (11.003442101259846,-19.92454618131115) {};
		
		\node[element] (d1311-d1122-a1) at (7.225554423555656,-17.16190983250353) {};
		\node[element] (d1311-d1122-b1) at (7.779954043520483,-17.482286269499618) {};
		\node[element] (d1311-d1122-c1) at (7.002410299676228,-18.187924793397414) {};
		
		\node[element] (d1311-d1122-a2) at (6.228704092760568,-18.897768935008887) {};
		\node[element] (d1311-d1122-b2) at (6.773706206915663,-19.21271474160852) {};
		\node[element] (d1311-d1122-c2) at (6.,-19.92255888322) {};
		
		\node[element] (d1101) at (-8.5,-15.5924318643) {};
		\node[element] (d1201) at (-9.00000000000021,-16.458457268084317) {};
		\node[element] (d1301) at (-8.0000000000002,-16.458457268084526) {};
		
		\node[element] (d1301-d1112-a1) at (-7.777263412660744,-17.484560777804862) {};
		\node[element] (d1301-d1112-b1) at (-7.222736587340055,-17.164404565933257) {};
		\node[element] (d1301-d1112-c1) at (-7.000000000000606,-18.190508075653582) {};
		
		\node[element] (d1301-d1112-a2) at (-6.777263412663967,-19.216611585374512) {};
		\node[element] (d1301-d1112-b2) at (-6.22273658733854,-18.89645537350017) {};
		\node[element] (d1301-d1112-c2) at (-6.000000000001899,-19.9225588832211) {};
		
		\node[element] (d1301-d1102-a1) at (-9.777263412660474,-17.164404565932884) {};
		\node[element] (d1301-d1102-b1) at (-9.222736587340343,-17.484560777804482) {};
		\node[element] (d1301-d1102-c1) at (-10.000000000000606,-18.19050807565305) {};
	
		\node[element] (d1301-d1102-a2) at (-10.777263412663114,-18.896455373499144) {};
		\node[element] (d1301-d1102-b2) at (-10.222736587337499,-19.21661158537392) {};
		\node[element] (d1301-d1102-c2) at (-11.,-19.92255888322) {};

		\begin{pgfonlayer}{background}
			\path[expand bubble]plot [smooth cycle,tension=1] coordinates {(c1) (c2) (c3)};
			
			\path[expand bubble]plot [smooth cycle,tension=1] coordinates {(d1300-d1111-a1) (d1300-d1111-b1) (d1300-d1111-c1)};
			\path[expand bubble]plot [smooth cycle,tension=1] coordinates {(d1300-d1111-a2) (d1300-d1111-b2) (d1300-d1111-c2)};
			\path[expand bubble]plot [smooth cycle,tension=1] coordinates {(d1300-d1111-a3) (d1300-d1111-b3) (d1300-d1111-c3)};
			\path[expand bubble]plot [smooth cycle,tension=1] coordinates {(d1300-d1111-a4) (d1300-d1111-b4) (d1300-d1111-c4)};
			\path[expand bubble]plot [smooth cycle,tension=1] coordinates {(d1300-d1111-a5) (d1300-d1111-b5) (d1300-d1111-c5)};
			\path[expand bubble]plot [smooth cycle,tension=1] coordinates {(d1300-d1111-a6) (d1300-d1111-b6) (d1300-d1111-c6)};
			\path[expand bubble]plot [smooth cycle,tension=1] coordinates {(d1300-d1111-a7) (d1300-d1111-b7) (d1300-d1111-c7)};
			\path[expand bubble]plot [smooth cycle,tension=1] coordinates {(d1300-d1111-a8) (d1300-d1111-b8) (d1300-d1111-c8)};
			
			\path[expand bubble]plot [smooth cycle,tension=1] coordinates {(d1200-d1100-c1) (d1200-d1100-a1) (d1200-d1100-b1)};
			\path[expand bubble]plot [smooth cycle,tension=1] coordinates {(d1200-d1100-c2) (d1200-d1100-a2) (d1200-d1100-b2)};
			\path[expand bubble]plot [smooth cycle,tension=1] coordinates {(d1200-d1100-c3) (d1200-d1100-a3) (d1200-d1100-b3)};
			\path[expand bubble]plot [smooth cycle,tension=1] coordinates {(d1200-d1100-c4) (d1200-d1100-a4) (d1200-d1100-b4)};
			\path[expand bubble]plot [smooth cycle,tension=1] coordinates {(d1200-d1100-c5) (d1200-d1100-a5) (d1200-d1100-b5)};
			\path[expand bubble]plot [smooth cycle,tension=1] coordinates {(d1200-d1100-c6) (d1200-d1100-a6) (d1200-d1100-b6)};
			\path[expand bubble]plot [smooth cycle,tension=1] coordinates {(d1200-d1100-c7) (d1200-d1100-a7) (d1200-d1100-b7)};
			\path[expand bubble]plot [smooth cycle,tension=1] coordinates {(d1200-d1100-c8) (d1200-d1100-a8) (d1200-d1100-b8)};
			
			\path[expand bubble]plot [smooth cycle,tension=1] coordinates {(d1111) (d1211) (d1311)};
			
			\path[expand bubble]plot [smooth cycle,tension=1] coordinates {(d1311-d1132-c1) (d1311-d1132-a1) (d1311-d1132-b1)};
			\path[expand bubble]plot [smooth cycle,tension=1] coordinates {(d1311-d1132-c2) (d1311-d1132-a2) (d1311-d1132-b2)};
			
			\path[expand bubble]plot [smooth cycle,tension=1] coordinates {(d1311-d1122-c1) (d1311-d1122-a1) (d1311-d1122-b1)};
			\path[expand bubble]plot [smooth cycle,tension=1] coordinates {(d1311-d1122-c2) (d1311-d1122-a2) (d1311-d1122-b2)};
			
			\path[expand bubble]plot [smooth cycle,tension=1] coordinates {(d1101) (d1201) (d1301)};
			
			\path[expand bubble]plot [smooth cycle,tension=1] coordinates {(d1301-d1112-a1) (d1301-d1112-b1) (d1301-d1112-c1)};
			\path[expand bubble]plot [smooth cycle,tension=1] coordinates {(d1301-d1112-a2) (d1301-d1112-b2) (d1301-d1112-c2)};
			
			\path[expand bubble]plot [smooth cycle,tension=1] coordinates {(d1301-d1102-a1) (d1301-d1102-b1) (d1301-d1102-c1)};
			\path[expand bubble]plot [smooth cycle,tension=1] coordinates {(d1301-d1102-a2) (d1301-d1102-b2) (d1301-d1102-c2)};

			\path[expand bubble]plot [smooth cycle,tension=1] coordinates {(d3200-d3101-c1) (d3200-d3101-a1) (d3200-d3101-b1)};
			\path[expand bubble]plot [smooth cycle,tension=1] coordinates {(d3200-d3101-c2) (d3200-d3101-a2) (d3200-d3101-b2)};
			\path[expand bubble]plot [smooth cycle,tension=1] coordinates {(d3200-d3101-c3) (d3200-d3101-a3) (d3200-d3101-b3)};
			\path[expand bubble]plot [smooth cycle,tension=1] coordinates {(d3200-d3101-c4) (d3200-d3101-a4) (d3200-d3101-b4)};
			\path[expand bubble]plot [smooth cycle,tension=1] coordinates {(d3200-d3101-c5) (d3200-d3101-a5) (d3200-d3101-b5)};
			\path[expand bubble]plot [smooth cycle,tension=1] coordinates {(d3200-d3101-c6) (d3200-d3101-a6) (d3200-d3101-b6)};
			\path[expand bubble]plot [smooth cycle,tension=1] coordinates {(d3200-d3101-c7) (d3200-d3101-a7) (d3200-d3101-b7)};
			\path[expand bubble]plot [smooth cycle,tension=1] coordinates {(d3200-d3101-c8) (d3200-d3101-a8) (d3200-d3101-b8)};
			
			\path[expand bubble]plot [smooth cycle,tension=1] coordinates {(d3101) (d3201) (d3301)};
			
			\path[expand bubble]plot [smooth cycle,tension=1] coordinates {(d3201-d3102-a1) (d3201-d3102-b1) (d3201-d3102-c1)};
			\path[expand bubble]plot [smooth cycle,tension=1] coordinates {(d3201-d3102-a2) (d3201-d3102-b2) (d3201-d3102-c2)};
			
			\path[expand bubble]plot [smooth cycle,tension=1] coordinates {(d3301-d3112-a1) (d3301-d3112-b1) (d3301-d3112-c1)};
			\path[expand bubble]plot [smooth cycle,tension=1] coordinates {(d3301-d3112-a2) (d3301-d3112-b2) (d3301-d3112-c2)};
			
			\path[expand bubble]plot [smooth cycle,tension=1] coordinates {(d2310-d2111-c1) (d2310-d2111-a1) (d2310-d2111-b1)};
			\path[expand bubble]plot [smooth cycle,tension=1] coordinates {(d2310-d2111-c2) (d2310-d2111-a2) (d2310-d2111-b2)};
			\path[expand bubble]plot [smooth cycle,tension=1] coordinates {(d2310-d2111-c3) (d2310-d2111-a3) (d2310-d2111-b3)};
			\path[expand bubble]plot [smooth cycle,tension=1] coordinates {(d2310-d2111-c4) (d2310-d2111-a4) (d2310-d2111-b4)};
			\path[expand bubble]plot [smooth cycle,tension=1] coordinates {(d2310-d2111-c5) (d2310-d2111-a5) (d2310-d2111-b5)};
			\path[expand bubble]plot [smooth cycle,tension=1] coordinates {(d2310-d2111-c6) (d2310-d2111-a6) (d2310-d2111-b6)};
			\path[expand bubble]plot [smooth cycle,tension=1] coordinates {(d2310-d2111-c7) (d2310-d2111-a7) (d2310-d2111-b7)};
			\path[expand bubble]plot [smooth cycle,tension=1] coordinates {(d2310-d2111-c8) (d2310-d2111-a8) (d2310-d2111-b8)};
			
			\path[expand bubble]plot [smooth cycle,tension=1] coordinates {(d2111) (d2211) (d2311)};
			
			\path[expand bubble]plot [smooth cycle,tension=1] coordinates {(d2311-d2132-a1) (d2311-d2132-b1) (d2311-d2132-c1)};
			\path[expand bubble]plot [smooth cycle,tension=1] coordinates {(d2311-d2132-a2) (d2311-d2132-b2) (d2311-d2132-c2)};
			
			\path[expand bubble]plot [smooth cycle,tension=1] coordinates {(d2211-d2122-a1) (d2211-d2122-b1) (d2211-d2122-c1)};
			\path[expand bubble]plot [smooth cycle,tension=1] coordinates {(d2211-d2122-a2) (d2211-d2122-b2) (d2211-d2122-c2)};
			
			\path[expand bubble]plot [smooth cycle,tension=1] coordinates {(d2200-d2101-c1) (d2200-d2101-a1) (d2200-d2101-b1)};
			\path[expand bubble]plot [smooth cycle,tension=1] coordinates {(d2200-d2101-c2) (d2200-d2101-a2) (d2200-d2101-b2)};
			\path[expand bubble]plot [smooth cycle,tension=1] coordinates {(d2200-d2101-c3) (d2200-d2101-a3) (d2200-d2101-b3)};
			\path[expand bubble]plot [smooth cycle,tension=1] coordinates {(d2200-d2101-c4) (d2200-d2101-a4) (d2200-d2101-b4)};
			\path[expand bubble]plot [smooth cycle,tension=1] coordinates {(d2200-d2101-c5) (d2200-d2101-a5) (d2200-d2101-b5)};
			\path[expand bubble]plot [smooth cycle,tension=1] coordinates {(d2200-d2101-c6) (d2200-d2101-a6) (d2200-d2101-b6)};
			\path[expand bubble]plot [smooth cycle,tension=1] coordinates {(d2200-d2101-c7) (d2200-d2101-a7) (d2200-d2101-b7)};
			\path[expand bubble]plot [smooth cycle,tension=1] coordinates {(d2200-d2101-c8) (d2200-d2101-a8) (d2200-d2101-b8)};
			
			\path[expand bubble]plot [smooth cycle,tension=1] coordinates {(d2101) (d2201) (d2301)};
			
			\path[expand bubble]plot [smooth cycle,tension=1] coordinates {(d2301-d2112-a1) (d2301-d2112-b1) (d2301-d2112-c1)};
			\path[expand bubble]plot [smooth cycle,tension=1] coordinates {(d2301-d2112-a2) (d2301-d2112-b2) (d2301-d2112-c2)};
			
			\path[expand bubble]plot [smooth cycle,tension=1] coordinates {(d2201-d2102-a1) (d2201-d2102-b1) (d2201-d2102-c1)};
			\path[expand bubble]plot [smooth cycle,tension=1] coordinates {(d2201-d2102-a2) (d2201-d2102-b2) (d2201-d2102-c2)};
			
			\path[expand bubble]plot [smooth cycle,tension=1] coordinates {(d3310-d3111-a1) (d3310-d3111-b1) (d3310-d3111-c1)};
			\path[expand bubble]plot [smooth cycle,tension=1] coordinates {(d3310-d3111-a2) (d3310-d3111-b2) (d3310-d3111-c2)};
			\path[expand bubble]plot [smooth cycle,tension=1] coordinates {(d3310-d3111-a3) (d3310-d3111-b3) (d3310-d3111-c3)};
			\path[expand bubble]plot [smooth cycle,tension=1] coordinates {(d3310-d3111-a4) (d3310-d3111-b4) (d3310-d3111-c4)};
			\path[expand bubble]plot [smooth cycle,tension=1] coordinates {(d3310-d3111-a5) (d3310-d3111-b5) (d3310-d3111-c5)};
			\path[expand bubble]plot [smooth cycle,tension=1] coordinates {(d3310-d3111-a6) (d3310-d3111-b6) (d3310-d3111-c6)};
			\path[expand bubble]plot [smooth cycle,tension=1] coordinates {(d3310-d3111-a7) (d3310-d3111-b7) (d3310-d3111-c7)};
			\path[expand bubble]plot [smooth cycle,tension=1] coordinates {(d3310-d3111-a8) (d3310-d3111-b8) (d3310-d3111-c8)};
			
			\path[expand bubble]plot [smooth cycle,tension=1] coordinates {(d3111) (d3211) (d3311)};
			
			\path[expand bubble]plot [smooth cycle,tension=1] coordinates {(d3211-d3222-a1) (d3211-d3222-b1) (d3211-d3222-c1)};
			\path[expand bubble]plot [smooth cycle,tension=1] coordinates {(d3211-d3222-a2) (d3211-d3222-b2) (d3211-d3222-c2)};
			
			\path[expand bubble]plot [smooth cycle,tension=1] coordinates {(d3211-d3232-a1) (d3211-d3232-b1) (d3211-d3232-c1)};
			\path[expand bubble]plot [smooth cycle,tension=1] coordinates {(d3211-d3232-a2) (d3211-d3232-b2) (d3211-d3232-c2)};
		\end{pgfonlayer}
		         (d3211-d3232-a2)
	\end{tikzpicture}
}

%% file: sections/reducedoutcome.tex
In this section, we shall define the reduced outcome $\pi_S$, where $S \subseteq C$ is a solution of $I$. Like the permanent popular outcome, the reduced outcome will assign every agent to a room with its closest neighbors. Additionally, a reduced outcome exists if and only if $I$ has a solution.

We shall describe how the agents are assigned to the rooms separately for each layer to construct the reduced outcome $\pi_S$.

\subsubsection{Bottom layer.} For the bottom layer, we shall use solution $S$ to assign every agent in the bottom layer to a room. These rooms only contain agents from the bottom layer.

Let $C_j \in S$ and $i \in C_j$. First, for each $C_j \in S$, we shall create the room $\{w_j^i | i \in C_j\}$.

We shall first consider the chains from the element-agent $u_i$ towards the set-agent $w^i_j$ corresponding to the 3-set $C_j \in S$ that contains the element $i$ in the solution. Let $t$ be a set-agent or bending point agent such that $\{u_i, t\}\in E'$. For chain $A_{\{u_i, t\}}$, assume w.l.o.g. that $u_i = \gamma_{\{u_i, t\}}[0]$. We shall create the rooms $\{\alpha_{\{u_i, t\}}[z], \beta_{\{u_i, t\}}[z], \gamma_{\{u_i, t\}}[z-1]\}$ for $z \in [\hat{n}]$. 

Let $f$ be a element-agent or bending point agent such that $\{f, w^i_j\} \in E'$. For the chain $A_{\{f, w^i_j\}}$, assume w.l.o.g. that $f = \gamma_{\{f, w^i_j\}}[0]$. We shall create the rooms $\{\alpha_{\{f, w^i_j\}}[z], \beta_{\{f, w^i_j\}}[z], \gamma_{\{f, w^i_j\}}[z-1]\}$ for $z \in [\hat{n}]$.

Let $C_{j'} \in C$ be a 3-set such that $C_{j'} \notin S$. We shall consider the chains from the element-agent $u_i$, where $i \in C_{j'}$, towards the set-agent $w^i_{j'}$.  Let $t$ be a set-agent or bending point agent such that $\{u_i, t\}\in E'$. For chain $A_{\{u_i, t\}}$, assume w.l.o.g. that $t = \gamma_{\{u_i, t\}}[\hat{n}]$. We shall create the rooms $\{\alpha_{\{u_i, t\}}[z], \beta_{\{u_i, t\}}[z], \gamma_{\{u_i, t\}}[z]\}$ for $z \in [\hat{n}]$. 

Let $f$ be a element-agent or bending point agent such that $\{f, w^i_{j'}\} \in E'$. For the chain $A_{\{f, w^i_{j'}\}}$, assume w.l.o.g. that $w^i_{j'} = \gamma_{\{f, w^i_{j'}\}}[\hat{n}]$. We shall create the rooms $\{\alpha_{\{f, w^i_{j'}\}}[z], \beta_{\{f, w^i_{j'}\}}[z], \gamma_{\{f, w^i_j\}}[z]\}$ for $z \in [\hat{n}]$.

See \cref{reducedgridgraphrooms} for an illustration.

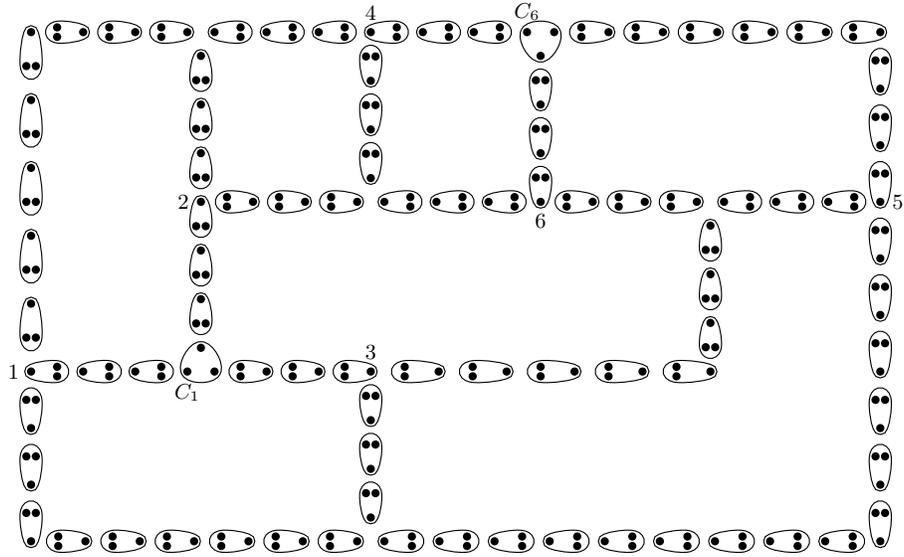
\begin{figure}
	\input{sections/figures/reducedgridgraph}
	\caption{Illustration of the rooms in the bottom layer from \cref{gridgraphfrompcx3cchains} in the reduced outcome, where the solution is $S= \{\{1,2,3\}, \{4,5,6\}\} = \{C_1, C_6\}$.}
	\label{reducedgridgraphrooms}
\end{figure}

\subsubsection{Ascending layer.} 
For the ascending layer, let $u_i$ be an element-agent with corresponding leaf $l_i$ and auxiliary agents $u_i', l_i'$. For chain $A_{\{u_i,u'_i\}}$, w.l.o.g. assume that $u_i = \gamma_{\{u_i,u'_i\}}[0]$ and $u'_i = \gamma_{\{u_i,u'_i\}}[\hat{n}]$. We shall create the rooms $\{\alpha_{\{u_i, u'_i\}}[z], \beta_{\{u_i, u'_i\}}[z], \gamma_{\{u_i, u'_i\}}[z]\} \in \pi_{pp}$ for $z \in [\hat{n}]$. 

For chain $A_{\{u'_i,l'_i\}}$, w.l.o.g. assume that $u'_i = \gamma_{\{u'_i,l'_i\}}[0]$ and $l'_i = \gamma_{\{u'_i,l'_i\}}[\hat{n}]$. We shall create the rooms $\{\alpha_{\{u'_i, l'_i\}}[z], \beta_{\{u'_i, l'_i\}}[z], \gamma_{\{u'_i, l'_i\}}[z]\} \in \pi_{pp}$ for $z \in [\hat{n}]$. 

For chain $A_{\{l'_i,l_i\}}$, w.l.o.g. assume that $l'_i = \gamma_{\{l'_i,l_i\}}[0]$ and $l_i = \gamma_{\{l'_i,l_i\}}[\hat{n}]$. We shall create the rooms $\{\alpha_{\{l'_i, l_i\}}[z], \beta_{\{l'_i, l_i\}}[z], \gamma_{\{l'_i, l_i\}}[z]\} \in \pi_{pp}$ for $z \in [\hat{n}]$. 

See \cref{reducedascendingrooms} for an illustration

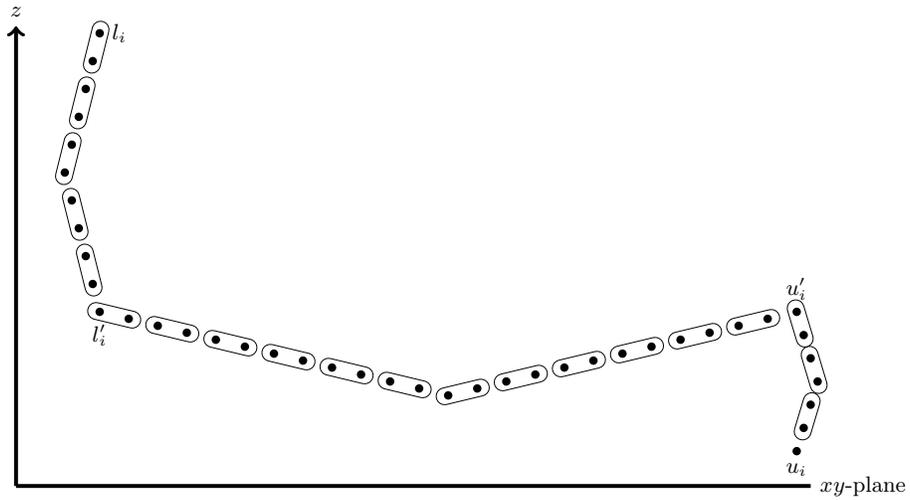
\begin{figure}
	\input{sections/figures/reducedascendingv2}
	\caption{Illustration of the rooms in the ascending layer in the reduced outcome.}
	\label{reducedascendingrooms}
\end{figure}

\subsubsection{Top layer.}
Similar to the permanent popular outcome, the rooms of the top layer in the reduced outcome depends on the parity of $\lfloor k \rfloor$. Let us first consider the case where $\lfloor k \rfloor$ is even.

Let $j \in \{j' \in [k] | j' \mod 2 = 0\}$ be chosen arbitrarily. Let $e$ be an edge with endpoints $d_{i,j}^{l',1}$ and $d_{i,j-1}^{l,p}$, where $l' \in \{2l-1, 2l\}$ and $p \in \{2,3\}$. W.l.o.g. assume that $d_{i,j}^{l',1} = \gamma_e[0]$ and $d_{i,j-1}^{l,p} = \gamma_e[\hat{n}]$. We shall create the rooms $\{\alpha_e[z], \beta_e[z], \gamma_e[z]\} \in \pi_{S}$ for $z \in [\hat{n}]$.

Let $j \in \{j' \in [0, k-1] | j' \mod 2 = 0\}$ be chosen arbitrarily. Let $e$ be an edge with endpoints $d_{i,j}^{l,p}$ and $d_{i,j+1}^{l',1}$, where $l' \in \{2l-1, 2l\}$ and $p \in \{2,3\}$. W.l.o.g. assume that $d_{i,j+1}^{l',1} = \gamma_e[0]$. We shall create the rooms $\{\gamma_e[z-1], \alpha_e[z], \beta_e[z]\} \in \pi_{pp}$ for $z \in [\hat{n}]$ and room $\{d_{i,j}^{l,1}, d_{i,j}^{l,2}, d_{i,j}^{l,3}\}$.

Let us now consider the case where $\lfloor k \rfloor$ is odd. First we create room $\{d_{1,0}^{1,1}, d_{2,0}^{1,1}, d_{3,0}^{1,1}\}$.

Let $j \in \{j' \in [k] | j' \mod 2 = 1\}$ be chosen arbitrarily. Let $e$ be an edge with endpoints $d_{i,j}^{l',1}$ and $d_{i,j-1}^{l,p}$, where $l' \in \{2l-1, 2l\}$ and $p \in \{2,3\}$. W.l.o.g. assume that $d_{i,j-1}^{l,p} = \gamma_e[\hat{n}]$. We shall create the rooms $\{\alpha_e[z], \beta_e[z], \gamma_e[z]\} \in \pi_{S}$ for $z \in [\hat{n}]$.

Let $j \in \{j' \in [k-1] | j' \mod 2 = 1\}$ be chosen arbitrarily. Let $e$ be an edge with endpoints $d_{i,j}^{l,p}$ and $d_{i,j+1}^{l',1}$, where $l' \in \{2l-1, 2l\}$ and $p \in \{2,3\}$. W.l.o.g. assume that $d_{i,j+1}^{l',1} = \gamma_e[0]$. We shall create the rooms $\{\gamma_e[z-1], \alpha_e[z], \beta_e[z]\} \in \pi_{pp}$ for $z \in [\hat{n}]$ and room $\{d_{i,j}^{l,1}, d_{i,j}^{l,2}, d_{i,j}^{l,3}\}$.

See \cref{toplayerreducedoutcome} for an illustration.

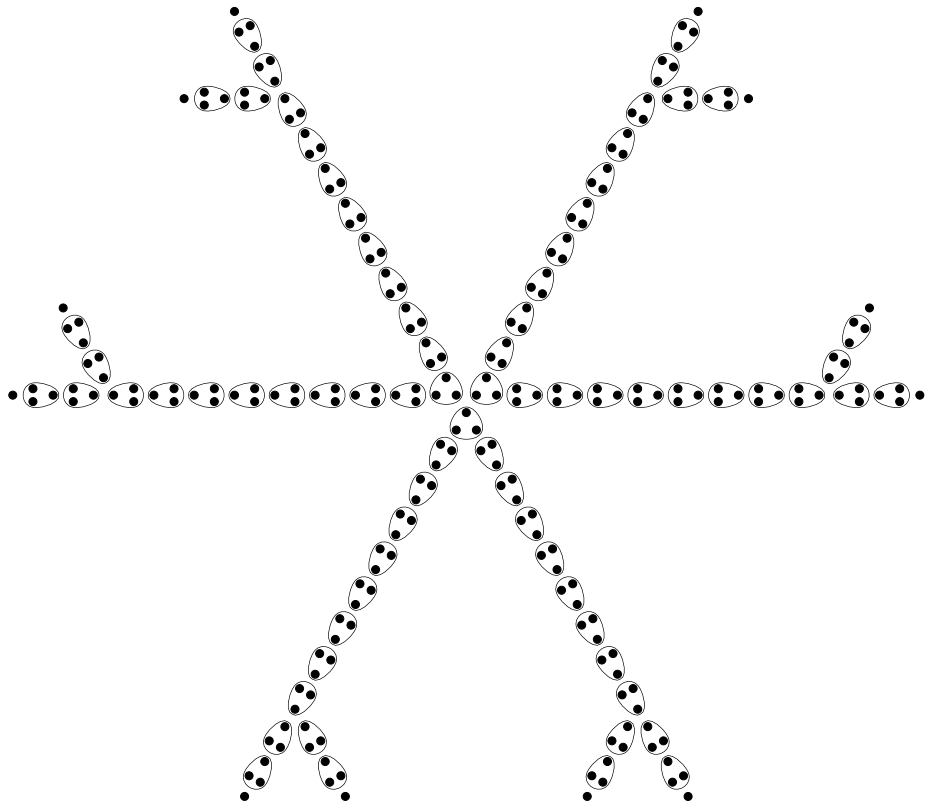
\begin{figure}
	\centering
	\input{sections/figures/snowflakereducedoutcome}
	\caption{Illustration of the rooms in the top layer in the reduced outcomes.}
	\label{toplayerreducedoutcome}
\end{figure}

%% file: sections/figures/reducedgridgraph.tex
\resizebox{\textwidth}{!}{%
	\begin{tikzpicture}[-,auto, semithick]
		\tikzset{
			element/.style={
				fill=black,draw=black,text=black,shape=circle,inner sep=1pt,minimum size=1pt
			}, set/.style={
				fill=black,draw=black,text=black,shape=circle,inner sep=1pt,minimum size=1pt
			}, expand bubble/.style={
				preaction={draw,line width=0.6em},
				white,fill,draw,line width=0.5em,
			},
		}
		\node[element, label=left:{$1$}] 	(1) at (0,-2.5) {};
		
		\node[element] 	(1-2-1) at (2.3/6,-2.5+0.07) {};
		\node[element] 	(1-2-2) at (2.3/6,-2.5-0.07) {};
		\node[element] 	(1-2-3) at (4.6/6,-2.5) {};
		\node[element] 	(1-2-4) at (6.9/6,-2.5+0.07) {};
		\node[element] 	(1-2-5) at (6.9/6,-2.5-0.07) {};
		\node[element] 	(1-2-6) at (9.2/6,-2.5) {};
		\node[element] 	(1-2-7) at (11.5/6,-2.5+0.07) {};
		\node[element] 	(1-2-8) at (11.5/6,-2.5-0.07) {};
		
		\node[element] 	(1-14-1) at (0+0.07,-2.5-2.5/6) {};
		\node[element] 	(1-14-2) at (0-0.07,-2.5-2.5/6) {};
		\node[element] 	(1-14-3) at (0,-2.5-5/6) {};
		\node[element] 	(1-14-4) at (0+0.07,-2.5-7.5/6) {};
		\node[element] 	(1-14-5) at (0-0.07,-2.5-7.5/6) {};
		\node[element] 	(1-14-6) at (0,-2.5-10/6) {};
		\node[element] 	(1-14-7) at (0+0.07,-2.5-12.5/6) {};
		\node[element] 	(1-14-8) at (0-0.07,-2.5-12.5/6) {};
		
		\node[element] 	(1-13-1) at (0+0.07,-2.5+5/10) {};
		\node[element] 	(1-13-2) at (0-0.07,-2.5+5/10) {};
		\node[element] 	(1-13-3) at (0,-2.5+10/10) {};
		\node[element] 	(1-13-4) at (0+0.07,-2.5+15/10) {};
		\node[element] 	(1-13-5) at (0-0.07,-2.5+15/10) {};
		\node[element] 	(1-13-6) at (0,-2.5+20/10) {};
		\node[element] 	(1-13-7) at (0+0.07,-2.5+25/10) {};
		\node[element] 	(1-13-8) at (0-0.07,-2.5+25/10) {};
		\node[element] 	(1-13-9) at (0,-2.5+30/10) {};
		\node[element] 	(1-13-10) at (0+0.07,-2.5+35/10) {};
		\node[element] 	(1-13-11) at (0-0.07,-2.5+35/10) {};
		\node[element] 	(1-13-12) at (0,-2.5+40/10) {};
		\node[element] 	(1-13-13) at (0+0.07,-2.5+45/10) {};
		\node[element] 	(1-13-14) at (0-0.07,-2.5+45/10) {};
		
		\node[set, label= below:{$C_1$}] 		(2-1) at (2.5-0.2,-2.5) {};
		\node[set] 		(2-2) at (2.5+0.2,-2.5) {};
		\node[set] 		(2-3) at (2.5,-2.5+0.35){};
		
		\node[element, label=above:{$3$}] 		(3) at (5,-2.5) {};
		
		\node[element] 		(3-2-1) at (5-2.3/6,-2.5+0.07) {};
		\node[element] 		(3-2-2) at (5-2.3/6,-2.5-0.07) {};
		\node[element] 		(3-2-3) at (5-4.6/6,-2.5) {};
		\node[element] 		(3-2-4) at (5-6.9/6,-2.5+0.07) {};
		\node[element] 		(3-2-5) at (5-6.9/6,-2.5-0.07) {};
		\node[element] 		(3-2-6) at (5-9.2/6,-2.5) {};
		\node[element] 		(3-2-7) at (5-11.5/6,-2.5+0.07) {};
		\node[element] 		(3-2-8) at (5-11.5/6,-2.5-0.07) {};
		
		\node[element] 		(3-4-1) at (5+0.07,-2.5-2.15/6) {};
		\node[element] 		(3-4-2) at (5-0.07,-2.5-2.15/6) {};
		\node[element] 		(3-4-3) at (5,-2.5-4.3/6) {};
		\node[element] 		(3-4-4) at (5+0.07,-2.5-6.45/6) {};
		\node[element] 		(3-4-5) at (5-0.07,-2.5-6.45/6) {};
		\node[element] 		(3-4-6) at (5,-2.5-8.6/6) {};
		\node[element] 		(3-4-7) at (5+0.07,-2.5-10.75/6) {};
		\node[element] 		(3-4-8) at (5-0.07,-2.5-10.75/6) {};
		
		\node[element] 		(3-15-1) at (5+0.5,-2.5+0.07) {};
		\node[element] 		(3-15-2) at (5+0.5,-2.5-0.07) {};
		\node[element] 		(3-15-3) at (5+1,-2.5) {};
		\node[element] 		(3-15-4) at (5+1.5,-2.5+0.07) {};
		\node[element] 		(3-15-5) at (5+1.5,-2.5-0.07) {};
		\node[element] 		(3-15-6) at (5+2,-2.5) {};
		\node[element] 		(3-15-7) at (5+2.5,-2.5+0.07) {};
		\node[element] 		(3-15-8) at (5+2.5,-2.5-0.07) {};
		\node[element] 		(3-15-9) at (5+3,-2.5) {};
		\node[element] 		(3-15-10) at (5+3.5,-2.5+0.07) {};
		\node[element] 		(3-15-11) at (5+3.5,-2.5-0.07) {};
		\node[element] 		(3-15-12) at (5+4,-2.5) {};
		\node[element] 		(3-15-13) at (5+4.5,-2.5+0.07) {};
		\node[element] 		(3-15-14) at (5+4.5,-2.5-0.07) {};
		
		\node[set] 	 (4-1) at (5-0.2,-5) {};
		\node[set] 	 (4-2) at (5+0.2,-5) {};
		\node[set] 	 (4-3) at (5,-5+0.35) {};
		
		\node[element, label=left:{$2$}] 		(5) at (2.5,0) {};
		
		\node[element] 		(5-2-1) at (2.5+0.07,-2.15/6) {};
		\node[element] 		(5-2-2) at (2.5-0.07,-2.15/6) {};
		\node[element] 		(5-2-3) at (2.5,-4.3/6) {};
		\node[element] 		(5-2-4) at (2.5+0.07,-6.45/6) {};
		\node[element] 		(5-2-5) at (2.5-0.07,-6.45/6) {};
		\node[element] 		(5-2-6) at (2.5,-8.6/6) {};
		\node[element] 		(5-2-7) at (2.5+0.07,-10.75/6) {};
		\node[element] 		(5-2-8) at (2.5-0.07,-10.75/6) {};
		
		\node[element] 		(5-6-1) at (2.5+0.07,2.15/6) {};
		\node[element] 		(5-6-2) at (2.5-0.07,2.15/6) {};
		\node[element] 		(5-6-3) at (2.5,4.3/6) {};
		\node[element] 		(5-6-4) at (2.5+0.07,6.45/6) {};
		\node[element] 		(5-6-5) at (2.5-0.07,6.45/6) {};
		\node[element] 		(5-6-6) at (2.5,8.6/6) {};
		\node[element] 		(5-6-7) at (2.5+0.07,10.75/6) {};
		\node[element] 		(5-6-8) at (2.5-0.07,10.75/6) {};
		
		\node[element] 		(5-7-1) at (2.5+2.3/6,0.07) {};
		\node[element] 		(5-7-2) at (2.5+2.3/6,-0.07) {};
		\node[element] 		(5-7-3) at (2.5+4.6/6,0) {};
		\node[element] 		(5-7-4) at (2.5+6.9/6,0.07) {};
		\node[element] 		(5-7-5) at (2.5+6.9/6,-0.07) {};
		\node[element] 		(5-7-6) at (2.5+9.2/6,0) {};
		\node[element] 		(5-7-7) at (2.5+11.5/6,0.07) {};
		\node[element] 		(5-7-8) at (2.5+11.5/6,-0.07) {};
		
		\node[set]		(6-1) at (2.5-0.2,2.5) {};
		\node[set]		(6-2) at (2.5+0.2,2.5) {};
		\node[set]		(6-3) at (2.5,2.5-0.35) {};
		
		\node[set]		(7-1) at (5-0.2,0) {};
		\node[set]		(7-2) at (5+0.2,0) {};
		\node[set]		(7-3) at (5,0+0.35) {};
		
		\node[element, label=above:{$4$}]		(8) at (5,2.5) {};
		
		\node[element]		(8-7-1) at (5+0.07,2.5-2.15/6) {};
		\node[element]		(8-7-2) at (5-0.07,2.5-2.15/6) {};
		\node[element]		(8-7-3) at (5,2.5-4.3/6) {};
		\node[element]		(8-7-4) at (5+0.07,2.5-6.45/6) {};
		\node[element]		(8-7-5) at (5-0.07,2.5-6.45/6) {};
		\node[element]		(8-7-6) at (5,2.5-8.6/6) {};
		\node[element]		(8-7-7) at (5+0.07,2.5-10.75/6) {};
		\node[element]		(8-7-8) at (5-0.07,2.5-10.75/6) {};
		
		\node[element]		(8-6-1) at (5-2.3/6,2.5+0.07) {};
		\node[element]		(8-6-2) at (5-2.3/6,2.5-0.07) {};
		\node[element]		(8-6-3) at (5-4.6/6,2.5) {};
		\node[element]		(8-6-4) at (5-6.9/6,2.5+0.07) {};
		\node[element]		(8-6-5) at (5-6.9/6,2.5-0.07) {};
		\node[element]		(8-6-6) at (5-9.2/6,2.5) {};
		\node[element]		(8-6-7) at (5-11.5/6,2.5+0.07) {};
		\node[element]		(8-6-8) at (5-11.5/6,2.5-0.07) {};
		
		\node[element]		(8-10-1) at (5+2.3/6,2.5+0.07) {};
		\node[element]		(8-10-2) at (5+2.3/6,2.5-0.07) {};
		\node[element]		(8-10-3) at (5+4.6/6,2.5) {};
		\node[element]		(8-10-4) at (5+6.9/6,2.5+0.07) {};
		\node[element]		(8-10-5) at (5+6.9/6,2.5-0.07) {};
		\node[element]		(8-10-6) at (5+9.1/6,2.5) {};
		\node[element]		(8-10-7) at (5+11.4/6,2.5+0.07) {};
		\node[element]		(8-10-8) at (5+11.4/6,2.5-0.07) {};
		
		\node[element, label=below:{$6$}]		(9) at (7.5,0) {};
		
		\node[element]		(9-7-1) at (7.5-2.3/6,0.07) {};
		\node[element]		(9-7-2) at (7.5-2.3/6,-0.07) {};
		\node[element]		(9-7-3) at (7.5-4.6/6,0) {};
		\node[element]		(9-7-4) at (7.5-6.9/6,0.07) {};
		\node[element]		(9-7-5) at (7.5-6.9/6,-0.07) {};
		\node[element]		(9-7-6) at (7.5-9.2/6,0) {};
		\node[element]		(9-7-7) at (7.5-11.5/6,0.07) {};
		\node[element]		(9-7-8) at (7.5-11.5/6,-0.07) {};
		
		\node[element]		(9-10-1) at (7.5+0.07,2.15/6) {};
		\node[element]		(9-10-2) at (7.5-0.07,2.15/6) {};
		\node[element]		(9-10-3) at (7.5,4.3/6) {};
		\node[element]		(9-10-4) at (7.5+0.07,6.45/6) {};
		\node[element]		(9-10-5) at (7.5-0.07,6.45/6) {};
		\node[element]		(9-10-6) at (7.5,8.6/6) {};
		\node[element]		(9-10-7) at (7.5+0.07,10.75/6) {};
		\node[element]		(9-10-8) at (7.5-0.07,10.75/6) {};
		
		\node[element]		(9-11-1) at (7.5+2.3/6,0.07) {};
		\node[element]		(9-11-2) at (7.5+2.3/6,-0.07) {};
		\node[element]		(9-11-3) at (7.5+4.6/6,0) {};
		\node[element]		(9-11-4) at (7.5+6.9/6,0.07) {};
		\node[element]		(9-11-5) at (7.5+6.9/6,-0.07) {};
		\node[element]		(9-11-6) at (7.5+9.2/6,0) {};
		\node[element]		(9-11-7) at (7.5+11.5/6,0.07) {};
		\node[element]		(9-11-8) at (7.5+11.5/6,-0.07) {};
		
		\node[set, label= above:{$C_6$}]		(10-1) at (7.5-0.2,2.5) {};
		\node[set]		(10-2) at (7.5+0.2,2.5) {};
		\node[set]		(10-3) at (7.5,2.5-0.35){};
		
		\node[set]		(11-1) at (10-0.2,0) {};
		\node[set]		(11-2) at (10+0.2,0) {};
		\node[set]		(11-3) at (10,-0.35) {};
		
		\node[element, label=right:{$5$}]		(12) at (12.5,0) {};
		
		\node[element]		(12-11-1) at (12.5-2.3/6,0.07) {};
		\node[element]		(12-11-2) at (12.5-2.3/6,-0.07) {};
		\node[element]		(12-11-3) at (12.5-4.6/6,0) {};
		\node[element]		(12-11-4) at (12.5-6.9/6,0.07) {};
		\node[element]		(12-11-5) at (12.5-6.9/6,-0.07) {};
		\node[element]		(12-11-6) at (12.5-9.2/6,0) {};
		\node[element]		(12-11-7) at (12.5-11.5/6,0.07) {};
		\node[element]		(12-11-8) at (12.5-11.5/6,-0.07) {};
		
		\node[element]		(12-16-1) at (12.5+0.07,-5/12) {};
		\node[element]		(12-16-2) at (12.5-0.07,-5/12) {};
		\node[element]		(12-16-3) at (12.5,-5*2/12) {};
		\node[element]		(12-16-4) at (12.5+0.07,-5*3/12) {};
		\node[element]		(12-16-5) at (12.5-0.07,-5*3/12) {};
		\node[element]		(12-16-6) at (12.5,-5*4/12) {};
		\node[element]		(12-16-7) at (12.5+0.07,-5*5/12) {};
		\node[element]		(12-16-8) at (12.5-0.07,-5*5/12) {};
		\node[element]		(12-16-9) at (12.5,-5*6/12) {};
		\node[element]		(12-16-10) at (12.5+0.07,-5*7/12) {};
		\node[element]		(12-16-11) at (12.5-0.07,-5*7/12) {};
		\node[element]		(12-16-12) at (12.5,-5*8/12) {};
		\node[element]		(12-16-13) at (12.5+0.07,-5*9/12) {};
		\node[element]		(12-16-14) at (12.5-0.07,-5*9/12) {};
		\node[element]		(12-16-15) at (12.5,-5*10/12) {};
		\node[element]		(12-16-16) at (12.5+0.07,-5*11/12) {};
		\node[element]		(12-16-17) at (12.5-0.07,-5*11/12) {};
		
		\node[element]		(12-17-1) at (12.5+0.07,2.5/6) {};
		\node[element]		(12-17-2) at (12.5-0.07,2.5/6) {};
		\node[element]		(12-17-3) at (12.5,2.5*2/6) {};
		\node[element]		(12-17-4) at (12.5+0.07,2.5*3/6) {};
		\node[element]		(12-17-5) at (12.5-0.07,2.5*3/6) {};
		\node[element]		(12-17-6) at (12.5,2.5*4/6) {};
		\node[element]		(12-17-7) at (12.5+0.07,2.5*5/6) {};
		\node[element]		(12-17-8) at (12.5-0.07,2.5*5/6) {};
		
		\node[element]		(13) at (0,2.5) {};
		
		\node[element]		(13-6-1) at (2.3/6,2.5+0.07) {};
		\node[element]		(13-6-2) at (2.3/6,2.5-0.07) {};
		\node[element]		(13-6-3) at (2.3*2/6,2.5) {};
		\node[element]		(13-6-4) at (2.3*3/6,2.5+0.07) {};
		\node[element]		(13-6-5) at (2.3*3/6,2.5-0.07) {};
		\node[element]		(13-6-6) at (2.3*4/6,2.5) {};
		\node[element]		(13-6-7) at (2.3*5/6,2.5+0.07) {};
		\node[element]		(13-6-8) at (2.3*5/6,2.5-0.07) {};
		
		\node[element]		(14) at (0,-5) {};
		
		\node[element]		(14-4-1) at (4.8/12,-5+0.07) {};
		\node[element]		(14-4-2) at (4.8/12,-5-0.07) {};
		\node[element]		(14-4-3) at (4.8*2/12,-5) {};
		\node[element]		(14-4-4) at (4.8*3/12,-5+0.07) {};
		\node[element]		(14-4-5) at (4.8*3/12,-5-0.07) {};
		\node[element]		(14-4-6) at (4.8*4/12,-5) {};
		\node[element]		(14-4-7) at (4.8*5/12,-5+0.07) {};
		\node[element]		(14-4-8) at (4.8*5/12,-5-0.07) {};
		\node[element]		(14-4-9) at (4.8*6/12,-5) {};
		\node[element]		(14-4-10) at (4.8*7/12,-5+0.07) {};
		\node[element]		(14-4-11) at (4.8*7/12,-5-0.07) {};
		\node[element]		(14-4-12) at (4.8*8/12,-5) {};
		\node[element]		(14-4-13) at (4.8*9/12,-5+0.07) {};
		\node[element]		(14-4-14) at (4.8*9/12,-5-0.07) {};
		\node[element]		(14-4-15) at (4.8*10/12,-5) {};
		\node[element]		(14-4-16) at (4.8*11/12,-5+0.07) {};
		\node[element]		(14-4-17) at (4.8*11/12,-5-0.07) {};
		
		\node[element]		(15) at (10,-2.5) {};
		
		\node[element]		(15-11-1) at (10.07,-2.5+2.15/6) {};
		\node[element]		(15-11-2) at (10-0.07,-2.5+2.15/6) {};
		\node[element]		(15-11-3) at (10,-2.5+2.15*2/6) {};
		\node[element]		(15-11-4) at (10.07,-2.5+2.15*3/6) {};
		\node[element]		(15-11-5) at (10-0.07,-2.5+2.15*3/6) {};
		\node[element]		(15-11-6) at (10,-2.5+2.15*4/6) {};
		\node[element]		(15-11-7) at (10.07,-2.5+2.15*5/6) {};
		\node[element]		(15-11-8) at (10-0.07,-2.5+2.15*5/6) {};
		
		\node[element]		(16) at (12.5,-5) {};
		
		\node[element]		(16-4-1) at (12.5-7.3/18,-5+0.07) {};
		\node[element]		(16-4-2) at (12.5-7.3/18,-5-0.07) {};
		\node[element]		(16-4-3) at (12.5-7.3*2/18,-5) {};
		\node[element]		(16-4-4) at (12.5-7.3*3/18,-5+0.07) {};
		\node[element]		(16-4-5) at (12.5-7.3*3/18,-5-0.07) {};
		\node[element]		(16-4-6) at (12.5-7.3*4/18,-5) {};
		\node[element]		(16-4-7) at (12.5-7.3*5/18,-5+0.07) {};
		\node[element]		(16-4-8) at (12.5-7.3*5/18,-5-0.07) {};
		\node[element]		(16-4-9) at (12.5-7.3*6/18,-5) {};
		\node[element]		(16-4-10) at (12.5-7.3*7/18,-5+0.07) {};
		\node[element]		(16-4-11) at (12.5-7.3*7/18,-5-0.07) {};
		\node[element]		(16-4-12) at (12.5-7.3*8/18,-5) {};
		\node[element]		(16-4-13) at (12.5-7.3*9/18,-5+0.07) {};
		\node[element]		(16-4-14) at (12.5-7.3*9/18,-5-0.07) {};
		\node[element]		(16-4-15) at (12.5-7.3*10/18,-5) {};
		\node[element]		(16-4-16) at (12.5-7.3*11/18,-5+0.07) {};
		\node[element]		(16-4-17) at (12.5-7.3*11/18,-5-0.07) {};
		\node[element]		(16-4-18) at (12.5-7.3*12/18,-5) {};
		\node[element]		(16-4-19) at (12.5-7.3*13/18,-5+0.07) {};
		\node[element]		(16-4-20) at (12.5-7.3*13/18,-5-0.07) {};
		\node[element]		(16-4-21) at (12.5-7.3*14/18,-5) {};
		\node[element]		(16-4-22) at (12.5-7.3*15/18,-5+0.07) {};
		\node[element]		(16-4-23) at (12.5-7.3*15/18,-5-0.07) {};
		\node[element]		(16-4-24) at (12.5-7.3*16/18,-5) {};
		\node[element]		(16-4-25) at (12.5-7.3*17/18,-5+0.07) {};
		\node[element]		(16-4-26) at (12.5-7.3*17/18,-5-0.07) {};
		
		\node[element]		(17) at (12.5,2.5) {};
		
		\node[element]		(17-10-1) at (12.5-4.8/12,2.5+0.07) {};
		\node[element]		(17-10-2) at (12.5-4.8/12,2.5-0.07) {};
		\node[element]		(17-10-3) at (12.5-4.8*2/12,2.5) {};
		\node[element]		(17-10-4) at (12.5-4.8*3/12,2.5+0.07) {};
		\node[element]		(17-10-5) at (12.5-4.8*3/12,2.5-0.07) {};
		\node[element]		(17-10-6) at (12.5-4.8*4/12,2.5) {};
		\node[element]		(17-10-7) at (12.5-4.8*5/12,2.5+0.07) {};
		\node[element]		(17-10-8) at (12.5-4.8*5/12,2.5-0.07) {};
		\node[element]		(17-10-9) at (12.5-4.8*6/12,2.5) {};
		\node[element]		(17-10-10) at (12.5-4.8*7/12,2.5+0.07) {};
		\node[element]		(17-10-11) at (12.5-4.8*7/12,2.5-0.07) {};
		\node[element]		(17-10-12) at (12.5-4.8*8/12,2.5) {};
		\node[element]		(17-10-13) at (12.5-4.8*9/12,2.5+0.07) {};
		\node[element]		(17-10-14) at (12.5-4.8*9/12,2.5-0.07) {};
		\node[element]		(17-10-15) at (12.5-4.8*10/12,2.5) {};
		\node[element]		(17-10-16) at (12.5-4.8*11/12,2.5+0.07) {};
		\node[element]		(17-10-17) at (12.5-4.8*11/12,2.5-0.07) {};
		

		\begin{pgfonlayer}{background}
			\path[expand bubble]plot [smooth cycle,tension=1] coordinates {(2-1) (2-2) (2-3)};
			\path[expand bubble]plot [smooth cycle,tension=1] coordinates {(10-1) (10-2) (10-3)};
			
			\path[expand bubble]plot [smooth cycle,tension=1] coordinates {(12-17-1) (12-17-2) (12)};
			\path[expand bubble]plot [smooth cycle,tension=1] coordinates {(12-17-4) (12-17-5) (12-17-3)};
			\path[expand bubble]plot [smooth cycle,tension=1] coordinates {(12-17-7) (12-17-8) (12-17-6)};
			\path[expand bubble]plot [smooth cycle,tension=1] coordinates {(17-10-1) (17-10-2) (17)};
			\path[expand bubble]plot [smooth cycle,tension=1] coordinates {(17-10-4) (17-10-5) (17-10-3)};
			\path[expand bubble]plot [smooth cycle,tension=1] coordinates {(17-10-7) (17-10-8) (17-10-6)};
			\path[expand bubble]plot [smooth cycle,tension=1] coordinates {(17-10-10) (17-10-11) (17-10-9)};
			\path[expand bubble]plot [smooth cycle,tension=1] coordinates {(17-10-13) (17-10-14) (17-10-12)};
			\path[expand bubble]plot [smooth cycle,tension=1] coordinates {(17-10-16) (17-10-17) (17-10-15)};
			
			\path[expand bubble]plot [smooth cycle,tension=1] coordinates {(12-16-1) (12-16-2) (12-16-3)};
			\path[expand bubble]plot [smooth cycle,tension=1] coordinates {(12-16-4) (12-16-5) (12-16-6)};
			\path[expand bubble]plot [smooth cycle,tension=1] coordinates {(12-16-7) (12-16-8) (12-16-9)};
			\path[expand bubble]plot [smooth cycle,tension=1] coordinates {(12-16-10) (12-16-11) (12-16-12)};
			\path[expand bubble]plot [smooth cycle,tension=1] coordinates {(12-16-13) (12-16-14) (12-16-15)};
			\path[expand bubble]plot [smooth cycle,tension=1] coordinates {(12-16-16) (12-16-17) (16)};
			\path[expand bubble]plot [smooth cycle,tension=1] coordinates {(16-4-1) (16-4-2) (16-4-3)};
			\path[expand bubble]plot [smooth cycle,tension=1] coordinates {(16-4-4) (16-4-5) (16-4-6)};
			\path[expand bubble]plot [smooth cycle,tension=1] coordinates {(16-4-7) (16-4-8) (16-4-9)};
			\path[expand bubble]plot [smooth cycle,tension=1] coordinates {(16-4-10) (16-4-11) (16-4-12)};
			\path[expand bubble]plot [smooth cycle,tension=1] coordinates {(16-4-13) (16-4-14) (16-4-15)};
			\path[expand bubble]plot [smooth cycle,tension=1] coordinates {(16-4-16) (16-4-17) (16-4-18)};
			\path[expand bubble]plot [smooth cycle,tension=1] coordinates {(16-4-19) (16-4-20) (16-4-21)};
			\path[expand bubble]plot [smooth cycle,tension=1] coordinates {(16-4-22) (16-4-23) (16-4-24)};
			\path[expand bubble]plot [smooth cycle,tension=1] coordinates {(16-4-25) (16-4-26) (4-2)};
			
			\path[expand bubble]plot [smooth cycle,tension=1] coordinates {(12-11-1) (12-11-2) (12-11-3)};
			\path[expand bubble]plot [smooth cycle,tension=1] coordinates {(12-11-4) (12-11-5) (12-11-6)};
			\path[expand bubble]plot [smooth cycle,tension=1] coordinates {(12-11-7) (12-11-8) (11-2)};
			
			\path[expand bubble]plot [smooth cycle,tension=1] coordinates {(9-11-1) (9-11-2) (9-11-3)};
			\path[expand bubble]plot [smooth cycle,tension=1] coordinates {(9-11-4) (9-11-5) (9-11-6)};
			\path[expand bubble]plot [smooth cycle,tension=1] coordinates {(9-11-7) (9-11-8) (11-1)};
			
			\path[expand bubble]plot [smooth cycle,tension=1] coordinates {(9-10-1) (9-10-2) (9)};
			\path[expand bubble]plot [smooth cycle,tension=1] coordinates {(9-10-4) (9-10-5) (9-10-3)};
			\path[expand bubble]plot [smooth cycle,tension=1] coordinates {(9-10-7) (9-10-8) (9-10-6)};
			
			\path[expand bubble]plot [smooth cycle,tension=1] coordinates {(9-7-1) (9-7-2) (9-7-3)};
			\path[expand bubble]plot [smooth cycle,tension=1] coordinates {(9-7-4) (9-7-5) (9-7-6)};
			\path[expand bubble]plot [smooth cycle,tension=1] coordinates {(9-7-7) (9-7-8) (7-2)};
			
			\path[expand bubble]plot [smooth cycle,tension=1] coordinates {(8-10-1) (8-10-2) (8)};
			\path[expand bubble]plot [smooth cycle,tension=1] coordinates {(8-10-4) (8-10-5) (8-10-3)};
			\path[expand bubble]plot [smooth cycle,tension=1] coordinates {(8-10-7) (8-10-8) (8-10-6)};
			
			\path[expand bubble]plot [smooth cycle,tension=1] coordinates {(8-6-1) (8-6-2) (8-6-3)};
			\path[expand bubble]plot [smooth cycle,tension=1] coordinates {(8-6-4) (8-6-5) (8-6-6)};
			\path[expand bubble]plot [smooth cycle,tension=1] coordinates {(8-6-7) (8-6-8) (6-2)};
			
			\path[expand bubble]plot [smooth cycle,tension=1] coordinates {(8-7-1) (8-7-2) (8-7-3)};
			\path[expand bubble]plot [smooth cycle,tension=1] coordinates {(8-7-4) (8-7-5) (8-7-6)};
			\path[expand bubble]plot [smooth cycle,tension=1] coordinates {(8-7-7) (8-7-8) (7-3)};
			
			\path[expand bubble]plot [smooth cycle,tension=1] coordinates {(5-7-1) (5-7-2) (5-7-3)};
			\path[expand bubble]plot [smooth cycle,tension=1] coordinates {(5-7-4) (5-7-5) (5-7-6)};
			\path[expand bubble]plot [smooth cycle,tension=1] coordinates {(5-7-7) (5-7-8) (7-1)};
			
			\path[expand bubble]plot [smooth cycle,tension=1] coordinates {(5-6-1) (5-6-2) (5-6-3)};
			\path[expand bubble]plot [smooth cycle,tension=1] coordinates {(5-6-4) (5-6-5) (5-6-6)};
			\path[expand bubble]plot [smooth cycle,tension=1] coordinates {(5-6-7) (5-6-8) (6-3)};
			
			\path[expand bubble]plot [smooth cycle,tension=1] coordinates {(5-2-1) (5-2-2) (5)};
			\path[expand bubble]plot [smooth cycle,tension=1] coordinates {(5-2-4) (5-2-5) (5-2-3)};
			\path[expand bubble]plot [smooth cycle,tension=1] coordinates {(5-2-7) (5-2-8) (5-2-6)};
			
			\path[expand bubble]plot [smooth cycle,tension=1] coordinates {(3-15-1) (3-15-2) (3-15-3)};
			\path[expand bubble]plot [smooth cycle,tension=1] coordinates {(3-15-4) (3-15-5) (3-15-6)};
			\path[expand bubble]plot [smooth cycle,tension=1] coordinates {(3-15-7) (3-15-8) (3-15-9)};
			\path[expand bubble]plot [smooth cycle,tension=1] coordinates {(3-15-10) (3-15-11) (3-15-12)};
			\path[expand bubble]plot [smooth cycle,tension=1] coordinates {(3-15-13) (3-15-14) (15)};
			\path[expand bubble]plot [smooth cycle,tension=1] coordinates {(15-11-1) (15-11-2) (15-11-3)};
			\path[expand bubble]plot [smooth cycle,tension=1] coordinates {(15-11-4) (15-11-5) (15-11-6)};
			\path[expand bubble]plot [smooth cycle,tension=1] coordinates {(15-11-7) (15-11-8) (11-3)};
			
			\path[expand bubble]plot [smooth cycle,tension=1] coordinates {(3-4-1) (3-4-2) (3-4-3)};
			\path[expand bubble]plot [smooth cycle,tension=1] coordinates {(3-4-4) (3-4-5) (3-4-6)};
			\path[expand bubble]plot [smooth cycle,tension=1] coordinates {(3-4-7) (3-4-8) (4-3)};
			
			\path[expand bubble]plot [smooth cycle,tension=1] coordinates {(3-2-1) (3-2-2) (3)};
			\path[expand bubble]plot [smooth cycle,tension=1] coordinates {(3-2-4) (3-2-5) (3-2-3)};
			\path[expand bubble]plot [smooth cycle,tension=1] coordinates {(3-2-7) (3-2-8) (3-2-6)};
			
			\path[expand bubble]plot [smooth cycle,tension=1] coordinates {(1) (1-2-1) (1-2-2)};
			\path[expand bubble]plot [smooth cycle,tension=1.1] coordinates {(1-2-4) (1-2-5) (1-2-3)};
			\path[expand bubble]plot [smooth cycle,tension=1.1] coordinates {(1-2-7) (1-2-8) (1-2-6)};
			
			\path[expand bubble]plot [smooth cycle,tension=1] coordinates {(1-14-1) (1-14-2) (1-14-3)};
			\path[expand bubble]plot [smooth cycle,tension=1] coordinates {(1-14-4) (1-14-5) (1-14-6)};
			\path[expand bubble]plot [smooth cycle,tension=1] coordinates {(1-14-7) (1-14-8) (14)};
			\path[expand bubble]plot [smooth cycle,tension=1] coordinates {(14-4-1) (14-4-2) (14-4-3)};
			\path[expand bubble]plot [smooth cycle,tension=1] coordinates {(14-4-4) (14-4-5) (14-4-6)};
			\path[expand bubble]plot [smooth cycle,tension=1] coordinates {(14-4-7) (14-4-8) (14-4-9)};
			\path[expand bubble]plot [smooth cycle,tension=1] coordinates {(14-4-10) (14-4-11) (14-4-12)};
			\path[expand bubble]plot [smooth cycle,tension=1] coordinates {(14-4-13) (14-4-14) (14-4-15)};
			\path[expand bubble]plot [smooth cycle,tension=1] coordinates {(14-4-16) (14-4-17) (4-1)};
			
			\path[expand bubble]plot [smooth cycle,tension=1] coordinates {(1-13-1) (1-13-2) (1-13-3)};
			\path[expand bubble]plot [smooth cycle,tension=1] coordinates {(1-13-4) (1-13-5) (1-13-6)};
			\path[expand bubble]plot [smooth cycle,tension=1] coordinates {(1-13-7) (1-13-8) (1-13-9)};
			\path[expand bubble]plot [smooth cycle,tension=1] coordinates {(1-13-10) (1-13-11) (1-13-12)};
			\path[expand bubble]plot [smooth cycle,tension=1] coordinates {(1-13-13) (1-13-14) (13)};
			\path[expand bubble]plot [smooth cycle,tension=1] coordinates {(13-6-1) (13-6-2) (13-6-3)};
			\path[expand bubble]plot [smooth cycle,tension=1] coordinates {(13-6-4) (13-6-5) (13-6-6)};
			\path[expand bubble]plot [smooth cycle,tension=1] coordinates {(13-6-7) (13-6-8) (6-1)};
		\end{pgfonlayer}
	\end{tikzpicture}
}

%% file: sections/figures/reducedascendingv2.tex
\resizebox{\textwidth}{!}{%
	\begin{tikzpicture}[-,auto, semithick]
		\tikzset{
			element/.style={
				fill=black,draw=black,text=black,shape=circle,inner sep=1pt,minimum size=1pt
			},
			set/.style={
				fill=white,draw=white,text=black,shape=circle
			},halo/.style={line join=round,
				double,line cap=round,double distance=#1,shorten >=-#1/2,shorten <=-#1/2},
			halo/.default=0.7em
		}
		
		\node[element, label= right:{$l_i$}] 	(4) at (0,6) {};
		\node[element, label= below:{$u_i$}] 	(1) at (10,0) {};
		
		\node[element, label= below:{$l'_i$}] 	(3) at (0,2) {};
		\node[element, label= above:{$u'_i$}] 	(2) at (10,2) {};
		
		\node[element] (1-1) at (10.1,1/3) {};
		\node[element] (1-2) at (10.2,2/3) {};
		\node[element] (1-3) at (10.3,3/3) {};
		\node[element] (1-4) at (10.2,4/3) {};
		\node[element] (1-5) at (10.1,5/3) {};
		
		\node[element] (2-23) at (5/12, 2-0.1) {};
		\node[element] (2-22) at (10/12,2-0.2) {};
		\node[element] (2-21) at (15/12,2-0.3) {};
		\node[element] (2-20) at (20/12,2-0.4) {};
		\node[element] (2-19) at (25/12,2-0.5) {};
		\node[element] (2-18) at (30/12,2-0.6) {};
		\node[element] (2-17) at (35/12,2-0.7) {};
		\node[element] (2-16) at (40/12,2-0.8) {};
		\node[element] (2-15) at (45/12,2-0.9) {};
		\node[element] (2-14) at (50/12,2-1) {};
		\node[element] (2-13) at (55/12,2-1.1) {};
		\node[element] (2-12) at (60/12,2-1.2) {};
		\node[element] (2-11) at (65/12,2-1.1) {};
		\node[element] (2-10) at (70/12,2-1) {};
		\node[element] (2-9) at (75/12,2-0.9) {};
		\node[element] (2-8) at (80/12,2-0.8) {};
		\node[element] (2-7) at (85/12,2-0.7) {};
		\node[element] (2-6) at (90/12,2-0.6) {};
		\node[element] (2-5) at (95/12,2-0.5) {};
		\node[element] (2-4) at (100/12,2-0.4) {};
		\node[element] (2-3) at (105/12,2-0.3) {};
		\node[element] (2-2) at (110/12,2-0.2) {};
		\node[element] (2-1) at (115/12,2-0.1) {};
		
		\node[element] (3-1) at (0-0.1,2.4) {};
		\node[element] (3-2) at (0-0.2,2.8) {};
		\node[element] (3-3) at (0-0.3,3.2) {};
		\node[element] (3-4) at (0-0.4,3.6) {};
		\node[element] (3-5) at (0-0.5,4) {};
		\node[element] (3-6) at (0-0.4,4.4) {};
		\node[element] (3-7) at (0-0.3,4.8) {};
		\node[element] (3-8) at (0-0.2,5.2) {};
		\node[element] (3-9) at (0-0.1,5.6) {};
		
		\draw[-,ultra thick] (-1.2,-0.5)--(10.2,-0.5) node[right]{$xy$-plane};
		\draw[->,ultra thick] (-1.2,-0.5)--(-1.2,6.1) node[above]{$z$};
		
		\begin{scope}[on background layer]     
			\draw[halo] (1-1) -- (1-2);
			\draw[halo] (1-4) -- (1-3);
			\draw[halo] (2) -- (1-5);
			
			\draw[halo] (2-1) -- (2-2);
			\draw[halo] (2-4) -- (2-3);
			\draw[halo] (2-6) -- (2-5);
			\draw[halo] (2-8) -- (2-7);
			\draw[halo] (2-10) -- (2-9);
			\draw[halo] (2-12) -- (2-11);
			\draw[halo] (2-14) -- (2-13);
			\draw[halo] (2-16) -- (2-15);
			\draw[halo] (2-18) -- (2-17);
			\draw[halo] (2-20) -- (2-19);
			\draw[halo] (2-22) -- (2-21);
			\draw[halo] (3) -- (2-23);
			
			\draw[halo] (3-1) -- (3-2);
			\draw[halo] (3-4) -- (3-3);
			\draw[halo] (3-6) -- (3-5);
			\draw[halo] (3-8) -- (3-7);
			\draw[halo] (4) -- (3-9);
		\end{scope}
	\end{tikzpicture}
}

%% file: sections/figures/snowflakereducedoutcome.tex
\resizebox{\textwidth}{!}{%
	\begin{tikzpicture}[-,auto, semithick]
		\tikzset{
			element/.style={
				fill=black,draw=black,text=black,shape=circle,inner sep=4pt,minimum size=4pt
			},
			set/.style={
				fill=white,draw=white,text=black,shape=circle
			}, 
			expand bubble/.style={
				preaction={draw,line width=1.8em},
				white,fill,draw,line width=1.6em,
			},
		}
		
		\node[element] (c1) at (0.,-0.8660254037844386) {};
		\node[element] (c2) at (-0.5,0.) {};
		\node[element] (c3) at (0.5,0.) {};
		
		\node[element] (d2200-d2101-c1) at (-1.,0.8660254037844387) {};
		\node[element] (d2200-d2101-a1) at (-1.777263412660235,1.571972701633057) {};
		\node[element] (d2200-d2101-b1) at (-1.222736587339765,1.8921289135046986) {};
		
		\node[element] (d2200-d2101-c2) at (-2.,2.5980762113533165) {};
		\node[element] (d2200-d2101-a2) at (-2.2227365873397558,3.62417972107358) {};
		\node[element] (d2200-d2101-b2) at (-2.7772634126602433,3.3040235092019286) {};
		
		\node[element] (d2200-d2101-c3) at (-3.,4.330127018922194) {};
		\node[element] (d2200-d2101-a3) at (-3.777263412660237,5.036074316770808) {};
		\node[element] (d2200-d2101-b3) at (-3.2227365873397584,5.356230528642454) {};
		
		\node[element] (d2200-d2101-c4) at (-4.,6.062177826491069) {};
		\node[element] (d2200-d2101-a4) at (-4.2227365873397655,7.088281336211331) {};
		\node[element] (d2200-d2101-b4) at (-4.777263412660233,6.768125124339692) {};
		
		\node[element] (d2200-d2101-c5) at (-5.,7.794228634059952) {};
		\node[element] (d2200-d2101-a5) at (-5.777263412660244,8.500175931908561) {};
		\node[element] (d2200-d2101-b5) at (-5.222736587339754,8.820332143780213) {};
		
		\node[element] (d2200-d2101-c6) at (-6.,9.526279441628823) {};
		\node[element] (d2200-d2101-a6) at (-6.222736587339765,10.552382951349072) {};
		\node[element] (d2200-d2101-b6) at (-6.7772634126602185,10.23222673947744) {};
		
		\node[element] (d2200-d2101-c7) at (-7.,11.258330249197689) {};
		\node[element] (d2200-d2101-a7) at (-7.777263412660218,11.964277547046319) {};
		\node[element] (d2200-d2101-b7) at (-7.22273658733977,12.284433758917947) {};
		
		\node[element] (d2200-d2101-c8) at (-8.,12.990381056766578) {};
		\node[element] (d2200-d2101-a8) at (-8.22273658733973,14.016484566486852) {};
		\node[element] (d2200-d2101-b8) at (-8.777263412660258,13.69632835461518) {};
		
		\node[element] (d2101) at (-9.,14.722431864335459) {};
		\node[element] (d2201) at (-9.500000000003242,15.588457268117999) {};
		\node[element] (d2301) at (-10.,14.722431864331696) {};
		
		\node[element] (d2201-d2102-a1) at (-10.27726341266613,16.294404565963696) {};
		\node[element] (d2201-d2102-b1) at (-9.72273658734688,16.61456077783742) {};
		\node[element] (d2201-d2102-c1) at (-10.50000000000976,17.32050807568309) {};
		
		\node[element] (d2201-d2102-a2) at (-11.277263412708495,18.02645537348929) {};
		\node[element] (d2201-d2102-b2) at (-10.722736587301267,18.346611585413807) {};
		\node[element] (d2201-d2102-c2) at (-11.5,19.05255888322) {};
		
		\node[element] (d2301-d2112-a1) at (-11.000000000001197,15.042588076199516) {};
		\node[element] (d2301-d2112-b1) at (-10.999999999998787,14.402275652456348) {};
		\node[element] (d2301-d2112-c1) at (-12.,14.722431864324168) {};
		
		\node[element] (d2301-d2112-a2) at (-12.999999999984906,15.042588076242883) {};
		\node[element] (d2301-d2112-b2) at (-12.999999999982496,14.40227565239792) {};
		\node[element] (d2301-d2112-c2) at (-13.999999999967393,14.722431864316636) {};
		
		\node[element] (d2310-d2111-c1) at (-1.5,0.) {};
		\node[element] (d2310-d2111-a1) at (-2.5,0.32015621187164245) {};
		\node[element] (d2310-d2111-b1) at (-2.5,-0.3201562118716425) {};
		
		\node[element] (d2310-d2111-c2) at (-3.5,0.) {};
		\node[element] (d2310-d2111-a2) at (-4.5,0.3201562118716411) {};
		\node[element] (d2310-d2111-b2) at (-4.5,-0.3201562118716411) {};
		
		\node[element] (d2310-d2111-c3) at (-5.5,0.) {};
		\node[element] (d2310-d2111-a3) at (-6.5,0.3201562118716411) {};
		\node[element] (d2310-d2111-b3) at (-6.5,-0.3201562118716411) {};
		
		\node[element] (d2310-d2111-c4) at (-7.5,0.) {};
		\node[element] (d2310-d2111-a4) at (-8.5,0.32015621187164384) {};
		\node[element] (d2310-d2111-b4) at (-8.5,-0.3201562118716439) {};
		
		\node[element] (d2310-d2111-c5) at (-9.5,0.) {};
		\node[element] (d2310-d2111-a5) at (-10.5,0.3201562118716328) {};
		\node[element] (d2310-d2111-b5) at (-10.5,-0.3201562118716328) {};
		
		\node[element] (d2310-d2111-c6) at (-11.5,0.) {};
		\node[element] (d2310-d2111-a6) at (-12.5,0.32015621187163) {};
		\node[element] (d2310-d2111-b6) at (-12.5,-0.32015621187163) {};
		
		\node[element] (d2310-d2111-c7) at (-13.5,0.) {};
		\node[element] (d2310-d2111-a7) at (-14.5,0.32015621187163) {};
		\node[element] (d2310-d2111-b7) at (-14.5,-0.32015621187163) {};
		
		\node[element] (d2310-d2111-c8) at (-15.5,0.) {};
		\node[element] (d2310-d2111-a8) at (-16.5,0.3201562118716466) {};
		\node[element] (d2310-d2111-b8) at (-16.5,-0.32015621187164667) {};
		
		\node[element] (d2111) at (-17.5,0.) {};
		\node[element] (d2211) at (-18.000000000000192,0.8660254037843291) {};
		\node[element] (d2311) at (-18.5,0.) {};
		
		\node[element] (d2311-d2132-a1) at (-19.499999999999886,-0.32015621187200327) {};
		\node[element] (d2311-d2132-b1) at (-19.5000000000001,0.3201562118713456) {};
		\node[element] (d2311-d2132-c1) at (-20.5,0.) {};
		
		\node[element] (d2311-d2132-a2) at (-21.499999999999105,0.32015621187367704) {};
		\node[element] (d2311-d2132-b2) at (-21.499999999998963,-0.32015621187543103) {};
		\node[element] (d2311-d2132-c2) at (-22.499999999998103,0.) {};
		
		\node[element] (d2211-d2122-a1) at (-18.777263412660524,1.5719727016327898) {};
		\node[element] (d2211-d2122-b1) at (-18.22273658734022,1.8921289135044983) {};
		\node[element] (d2211-d2122-c1) at (-19.00000000000056,2.598076211352972) {};
		
		\node[element] (d2211-d2122-a2) at (-19.77726341266323,3.3040235091990144) {};
		\node[element] (d2211-d2122-b2) at (-19.22273658733733,3.624179721073954) {};
		\node[element] (d2211-d2122-c2) at (-20.,4.33012701892) {};
		
		\node[element] (d3200-d3101-c1) at (1.5,0.) {};
		\node[element] (d3200-d3101-a1) at (2.5,0.3201562118716411) {};
		\node[element] (d3200-d3101-b1) at (2.5,-0.3201562118716411) {};
		
		\node[element] (d3200-d3101-c2) at (3.5,0.) {};
		\node[element] (d3200-d3101-a2) at (4.5,0.3201562118716411) {};
		\node[element] (d3200-d3101-b2) at (4.5,-0.3201562118716411) {};
		
		\node[element] (d3200-d3101-c3) at (5.5,0.) {};
		\node[element] (d3200-d3101-a3) at (6.5,0.3201562118716411) {};
		\node[element] (d3200-d3101-b3) at (6.5,-0.3201562118716411) {};
		
		\node[element] (d3200-d3101-c4) at (7.5,0.) {};
		\node[element] (d3200-d3101-a4) at (8.5,0.3201562118716494) {};
		\node[element] (d3200-d3101-b4) at (8.5,-0.32015621187164944) {};
		
		\node[element] (d3200-d3101-c5) at (9.5,0.) {};
		\node[element] (d3200-d3101-a5) at (10.5,0.32015621187164384) {};
		\node[element] (d3200-d3101-b5) at (10.5,-0.3201562118716439) {};
		
		\node[element] (d3200-d3101-c6) at (11.5,0.) {};
		\node[element] (d3200-d3101-a6) at (12.5,0.32015621187163) {};
		\node[element] (d3200-d3101-b6) at (12.5,-0.32015621187163) {};
		
		\node[element] (d3200-d3101-c7) at (13.5,0.) {};
		\node[element] (d3200-d3101-a7) at (14.5,0.32015621187163) {};
		\node[element] (d3200-d3101-b7) at (14.5,-0.32015621187163) {};
		
		\node[element] (d3200-d3101-c8) at (15.5,0.) {};
		\node[element] (d3200-d3101-a8) at (16.5,0.3201562118716688) {};
		\node[element] (d3200-d3101-b8) at (16.5,-0.32015621187166887) {};
		
		\node[element] (d3101) at (17.5,0.) {};
		\node[element] (d3301) at (18.,0.8660254037844387) {};
		\node[element] (d3201) at (18.5,0.) {};
		
		\node[element] (d3201-d3102-a1) at (19.5,-0.3201562118716744) {};
		\node[element] (d3201-d3102-b1) at (19.5,0.32015621187167437) {};
		\node[element] (d3201-d3102-c1) at (20.5,0.) {};
		
		\node[element] (d3201-d3102-a2) at (21.5,0.32015621187167437) {};
		\node[element] (d3201-d3102-b2) at (21.5,-0.3201562118716744) {};
		\node[element] (d3201-d3102-c2) at (22.5,0.) {};
		
		\node[element] (d3301-d3112-a1) at (18.22273658733974,1.8921289135047132) {};
		\node[element] (d3301-d3112-b1) at (18.777263412660265,1.5719727016330403) {};
		\node[element] (d3301-d3112-c1) at (19.,2.5980762113533156) {};
		
		\node[element] (d3301-d3112-a2) at (19.777263412660265,3.3040235092019152) {};
		\node[element] (d3301-d3112-b2) at (19.22273658733974,3.624179721073588) {};
		\node[element] (d3301-d3112-c2) at (20.,4.330127018922194) {};
		
		\node[element] (d3310-d3111-c1) at (1.,0.8660254037844386) {};
		\node[element] (d3310-d3111-a1) at (1.2227365873397644,1.8921289135046986) {};
		\node[element] (d3310-d3111-b1) at (1.7772634126602358,1.5719727016330558) {};
		
		\node[element] (d3310-d3111-c2) at (2.,2.5980762113533156) {};
		\node[element] (d3310-d3111-a2) at (2.777263412660237,3.3040235092019326) {};
		\node[element] (d3310-d3111-b2) at (2.222736587339764,3.6241797210735767) {};
		
		\node[element] (d3310-d3111-c3) at (3.,4.330127018922194) {};
		\node[element] (d3310-d3111-a3) at (3.2227365873397615,5.356230528642455) {};
		\node[element] (d3310-d3111-b3) at (3.7772634126602402,5.036074316770809) {};
		
		\node[element] (d3310-d3111-c4) at (4.,6.062177826491069) {};
		\node[element] (d3310-d3111-a4) at (4.777263412660231,6.76812512433969) {};
		\node[element] (d3310-d3111-b4) at (4.222736587339771,7.088281336211326) {};
		
		\node[element] (d3310-d3111-c5) at (5.,7.794228634059945) {};
		\node[element] (d3310-d3111-a5) at (5.777263412660234,8.500175931908561) {};
		\node[element] (d3310-d3111-b5) at (5.222736587339765,8.820332143780202) {};
		
		\node[element] (d3310-d3111-c6) at (6.,9.526279441628821) {};
		\node[element] (d3310-d3111-a6) at (6.222736587339767,10.552382951349085) {};
		\node[element] (d3310-d3111-b6) at (6.777263412660235,10.232226739477444) {};
		
		\node[element] (d3310-d3111-c7) at (7.,11.2583302491977) {};
		\node[element] (d3310-d3111-a7) at (7.777263412660203,11.964277547046333) {};
		\node[element] (d3310-d3111-b7) at (7.222736587339793,12.28443375891794) {};
		
		\node[element] (d3310-d3111-c8) at (8.,12.990381056766577) {};
		\node[element] (d3310-d3111-a8) at (8.222736587339822,14.016484566486817) {};
		\node[element] (d3310-d3111-b8) at (8.777263412660194,13.696328354615233) {};

		\node[element] (d3111) at (9.,14.722431864335459) {};
		\node[element] (d3211) at (10.,14.722431864328366) {};
		\node[element] (d3311) at (9.500000000006144,15.588457268116352) {};
		
		\node[element] (d3211-d3222-a1) at (10.999999999997726,14.402275652449552) {};
		\node[element] (d3211-d3222-b1) at (11.000000000002267,15.042588076192999) {};
		\node[element] (d3211-d3222-c1) at (12.,14.722431864314183) {};
		
		\node[element] (d3211-d3222-a2) at (12.999999999997716,14.402275652435467) {};
		\node[element] (d3211-d3222-b2) at (13.000000000002256,15.042588076178719) {};
		\node[element] (d3211-d3222-c2) at (14.,14.7224318643) {};

		\node[element] (d3211-d3232-a1) at (9.722736587353138,16.614560777835027) {};
		\node[element] (d3211-d3232-b1) at (10.277263412671406,16.29440456595941) {};
		\node[element] (d3211-d3232-c1) at (10.500000000018408,17.320508075678102) {};
		
		\node[element] (d3211-d3232-a2) at (11.277263412683617,18.026455373521273) {};
		\node[element] (d3211-d3232-b2) at (10.722736587365524,18.346611585396786) {};
		\node[element] (d3211-d3232-c2) at (11.50000000003071,19.05255888323992) {};
		
		\node[element] (d1300-d1111-c1) at (0.500101214362474,-1.7319923635428225) {};
		\node[element] (d1300-d1111-a1) at (1.2774471301671986,-2.437848812958505) {};
		\node[element] (d1300-d1111-b1) at (0.7229577272826965,-2.758069833643906) {};
		
		\node[element] (d1300-d1111-c2) at (1.5003036430874208,-3.4639262830595876) {};
		\node[element] (d1300-d1111-a2) at (2.2776495588921404,-4.169782732475276) {};
		\node[element] (d1300-d1111-b2) at (1.7231601560076508,-4.490003753160669) {};
		
		\node[element] (d1300-d1111-c3) at (2.50050607181237,-5.195860202576356) {};
		\node[element] (d1300-d1111-a3) at (3.27785198761708,-5.90171665199205) {};
		\node[element] (d1300-d1111-b3) at (2.7233625847326084,-6.2219376726774325) {};
		
		\node[element] (d1300-d1111-c4) at (3.5007085005373213,-6.92779412209313) {};
		\node[element] (d1300-d1111-a4) at (4.278054416342035,-7.633650571508832) {};
		\node[element] (d1300-d1111-b4) at (3.7235650134575655,-7.953871592194213) {};
		
		\node[element] (d1300-d1111-c5) at (4.500910929262278,-8.659728041609911) {};
		\node[element] (d1300-d1111-b5) at (5.278256845066982,-9.365584491025619) {};
		\node[element] (d1300-d1111-a5) at (4.723767442182533,-9.685805511710988) {};
		
		\node[element] (d1300-d1111-c6) at (5.501113357987235,-10.391661961126694) {};
		\node[element] (d1300-d1111-a6) at (6.2784592737919445,-11.097518410542383) {};
		\node[element] (d1300-d1111-b6) at (5.72396987090747,-11.417739431227767) {};
		
		\node[element] (d1300-d1111-c7) at (6.501315786712181,-12.123595880643457) {};
		\node[element] (d1300-d1111-a7) at (6.724172299632444,-13.149673350744504) {};
		\node[element] (d1300-d1111-b7) at (7.2786617025168505,-12.829452330059159) {};
		
		\node[element] (d1300-d1111-c8) at (7.501518215437113,-13.855529800160198) {};
		\node[element] (d1300-d1111-a8) at (7.729934586918504,-14.880383917212601) {};
		\node[element] (d1300-d1111-b8) at (8.275025729780358,-14.565590449157538) {};

		\node[element] (d1200-d1100-c1) at (-0.4998987788075001,-1.7321092397649804) {};
		\node[element] (d1200-d1100-a1) at (-1.2770796777057918,-2.438147376403253) {};
		\node[element] (d1200-d1100-b1) at (-0.7225154375242089,-2.7582387750877917) {};
		
		\node[element] (d1200-d1100-c2) at (-1.4996963364225004,-3.464276911726065) {};
		\node[element] (d1200-d1100-a2) at (-2.2768772353207907,-4.1703150483643405) {};
		\node[element] (d1200-d1100-b2) at (-1.7223129951392115,-4.490406447048876) {};
		
		\node[element] (d1200-d1100-c3) at (-2.499493894037501,-5.196444583687151) {};
		\node[element] (d1200-d1100-a3) at (-3.276674792935787,-5.902482720325423) {};
		\node[element] (d1200-d1100-b3) at (-2.722110552754211,-6.222574119009958) {};
		
		\node[element] (d1200-d1100-c4) at (-3.4992914516524998,-6.928612255648232) {};
		\node[element] (d1200-d1100-a4) at (-4.27647235055079,-7.63465039228651) {};
		\node[element] (d1200-d1100-b4) at (-3.7219081103692124,-7.954741790971045) {};
		
		\node[element] (d1200-d1100-c5) at (-4.499089009267505,-8.660779927609324) {};
		\node[element] (d1200-d1100-a5) at (-5.276269908165813,-9.366818064247578) {};
		\node[element] (d1200-d1100-b5) at (-4.721705667984187,-9.68690946293214) {};
		
		\node[element] (d1200-d1100-c6) at (-5.498886566882493,-10.392947599570387) {};
		\node[element] (d1200-d1100-a6) at (-6.2760674657807956,-11.09898573620863) {};
		\node[element] (d1200-d1100-b6) at (-5.72150322559917,-11.419077134893193) {};
		
		\node[element] (d1200-d1100-c7) at (-6.498684124497479,-12.125115271531447) {};
		\node[element] (d1200-d1100-a7) at (-7.275865023395801,-12.831153408169708) {};
		\node[element] (d1200-d1100-b7) at (-6.721300783214163,-13.15124480685428) {};
		
		\node[element] (d1200-d1100-c8) at (-7.498481682112478,-13.85728294349253) {};
		\node[element] (d1200-d1100-a8) at (-8.271823196912203,-14.567524382361928) {};
		\node[element] (d1200-d1100-b8) at (-7.726658485200276,-14.882190425430602) {};
		
		\node[element] (d1111) at (8.503442101261754,-15.590444566209953) {};
		\node[element] (d1211) at (8.003098167399914,-16.45627130860572) {};
		\node[element] (d1311) at (9.003098088521511,-16.456668494761615) {};
		
		\node[element] (d1311-d1132-a1) at (9.225427104183659,-17.482860391433327) {};
		\node[element] (d1311-d1132-b1) at (9.780081047378895,-17.16292445519325) {};
		\node[element] (d1311-d1132-c1) at (10.002410063041049,-18.189116351864975) {};
		
		\node[element] (d1311-d1132-a2) at (10.775552145354403,-18.899574880250796) {};
		\node[element] (d1311-d1132-b2) at (10.2303000189465,-19.214087652925347) {};
		\node[element] (d1311-d1132-c2) at (11.003442101259846,-19.92454618131115) {};
		
		\node[element] (d1311-d1122-a1) at (7.225554423555656,-17.16190983250353) {};
		\node[element] (d1311-d1122-b1) at (7.779954043520483,-17.482286269499618) {};
		\node[element] (d1311-d1122-c1) at (7.002410299676228,-18.187924793397414) {};
		
		\node[element] (d1311-d1122-a2) at (6.228704092760568,-18.897768935008887) {};
		\node[element] (d1311-d1122-b2) at (6.773706206915663,-19.21271474160852) {};
		\node[element] (d1311-d1122-c2) at (6.,-19.92255888322) {};
		
		\node[element] (d1101) at (-8.5,-15.5924318643) {};
		\node[element] (d1201) at (-9.00000000000021,-16.458457268084317) {};
		\node[element] (d1301) at (-8.0000000000002,-16.458457268084526) {};
		
		\node[element] (d1301-d1112-a1) at (-7.777263412660744,-17.484560777804862) {};
		\node[element] (d1301-d1112-b1) at (-7.222736587340055,-17.164404565933257) {};
		\node[element] (d1301-d1112-c1) at (-7.000000000000606,-18.190508075653582) {};
		
		\node[element] (d1301-d1112-a2) at (-6.777263412663967,-19.216611585374512) {};
		\node[element] (d1301-d1112-b2) at (-6.22273658733854,-18.89645537350017) {};
		\node[element] (d1301-d1112-c2) at (-6.000000000001899,-19.9225588832211) {};
		
		\node[element] (d1301-d1102-a1) at (-9.777263412660474,-17.164404565932884) {};
		\node[element] (d1301-d1102-b1) at (-9.222736587340343,-17.484560777804482) {};
		\node[element] (d1301-d1102-c1) at (-10.000000000000606,-18.19050807565305) {};
		
		\node[element] (d1301-d1102-a2) at (-10.777263412663114,-18.896455373499144) {};
		\node[element] (d1301-d1102-b2) at (-10.222736587337499,-19.21661158537392) {};
		\node[element] (d1301-d1102-c2) at (-11.,-19.92255888322) {};

		\begin{pgfonlayer}{background}
			
			\path[expand bubble]plot [smooth cycle,tension=1] coordinates {(c1) (d1200-d1100-c1) (d1300-d1111-c1)};
			
			\path[expand bubble]plot [smooth cycle,tension=1] coordinates {(d1300-d1111-a1) (d1300-d1111-b1) (d1300-d1111-c2)};
			\path[expand bubble]plot [smooth cycle,tension=1] coordinates {(d1300-d1111-a2) (d1300-d1111-b2) (d1300-d1111-c3)};
			\path[expand bubble]plot [smooth cycle,tension=1] coordinates {(d1300-d1111-a3) (d1300-d1111-b3) (d1300-d1111-c4)};
			\path[expand bubble]plot [smooth cycle,tension=1] coordinates {(d1300-d1111-a4) (d1300-d1111-b4) (d1300-d1111-c5)};
			\path[expand bubble]plot [smooth cycle,tension=1] coordinates {(d1300-d1111-a5) (d1300-d1111-b5) (d1300-d1111-c6)};
			\path[expand bubble]plot [smooth cycle,tension=1] coordinates {(d1300-d1111-a6) (d1300-d1111-b6) (d1300-d1111-c7)};
			\path[expand bubble]plot [smooth cycle,tension=1] coordinates {(d1300-d1111-a7) (d1300-d1111-b7) (d1300-d1111-c8)};
			\path[expand bubble]plot [smooth cycle,tension=1] coordinates {(d1300-d1111-a8) (d1300-d1111-b8) (d1111)};
			
			\path[expand bubble]plot [smooth cycle,tension=1] coordinates {(d1200-d1100-c2) (d1200-d1100-a1) (d1200-d1100-b1)};
			\path[expand bubble]plot [smooth cycle,tension=1] coordinates {(d1200-d1100-c3) (d1200-d1100-a2) (d1200-d1100-b2)};
			\path[expand bubble]plot [smooth cycle,tension=1] coordinates {(d1200-d1100-c4) (d1200-d1100-a3) (d1200-d1100-b3)};
			\path[expand bubble]plot [smooth cycle,tension=1] coordinates {(d1200-d1100-c5) (d1200-d1100-a4) (d1200-d1100-b4)};
			\path[expand bubble]plot [smooth cycle,tension=1] coordinates {(d1200-d1100-c6) (d1200-d1100-a5) (d1200-d1100-b5)};
			\path[expand bubble]plot [smooth cycle,tension=1] coordinates {(d1200-d1100-c7) (d1200-d1100-a6) (d1200-d1100-b6)};
			\path[expand bubble]plot [smooth cycle,tension=1] coordinates {(d1200-d1100-c8) (d1200-d1100-a7) (d1200-d1100-b7)};
			\path[expand bubble]plot [smooth cycle,tension=1] coordinates {(d1101) (d1200-d1100-a8) (d1200-d1100-b8)};
			
			\path[expand bubble]plot [smooth cycle,tension=1] coordinates {(d1311) (d1311-d1132-a1) (d1311-d1132-b1)};
			\path[expand bubble]plot [smooth cycle,tension=1] coordinates {(d1311-d1132-c1) (d1311-d1132-a2) (d1311-d1132-b2)};
			
			\path[expand bubble]plot [smooth cycle,tension=1] coordinates {(d1211) (d1311-d1122-a1) (d1311-d1122-b1)};
			\path[expand bubble]plot [smooth cycle,tension=1] coordinates {(d1311-d1122-c1) (d1311-d1122-a2) (d1311-d1122-b2)};
			
			\path[expand bubble]plot [smooth cycle,tension=1] coordinates {(d1301-d1112-a1) (d1301-d1112-b1) (d1301)};
			\path[expand bubble]plot [smooth cycle,tension=1] coordinates {(d1301-d1112-a2) (d1301-d1112-b2) (d1301-d1112-c1)};
			
			\path[expand bubble]plot [smooth cycle,tension=1] coordinates {(d1301-d1102-a1) (d1301-d1102-b1) (d1201)};
			\path[expand bubble]plot [smooth cycle,tension=1] coordinates {(d1301-d1102-a2) (d1301-d1102-b2) (d1301-d1102-c1)};
			
			\path[expand bubble]plot [smooth cycle,tension=1] coordinates {(d3200-d3101-c1) (d3310-d3111-c1) (c3)};
			
			\path[expand bubble]plot [smooth cycle,tension=1] coordinates {(d3200-d3101-c2) (d3200-d3101-a1) (d3200-d3101-b1)};
			\path[expand bubble]plot [smooth cycle,tension=1] coordinates {(d3200-d3101-c3) (d3200-d3101-a2) (d3200-d3101-b2)};
			\path[expand bubble]plot [smooth cycle,tension=1] coordinates {(d3200-d3101-c4) (d3200-d3101-a3) (d3200-d3101-b3)};
			\path[expand bubble]plot [smooth cycle,tension=1] coordinates {(d3200-d3101-c5) (d3200-d3101-a4) (d3200-d3101-b4)};
			\path[expand bubble]plot [smooth cycle,tension=1] coordinates {(d3200-d3101-c6) (d3200-d3101-a5) (d3200-d3101-b5)};
			\path[expand bubble]plot [smooth cycle,tension=1] coordinates {(d3200-d3101-c7) (d3200-d3101-a6) (d3200-d3101-b6)};
			\path[expand bubble]plot [smooth cycle,tension=1] coordinates {(d3200-d3101-c8) (d3200-d3101-a7) (d3200-d3101-b7)};
			\path[expand bubble]plot [smooth cycle,tension=1] coordinates {(d3101) (d3200-d3101-a8) (d3200-d3101-b8)};
			
			\path[expand bubble]plot [smooth cycle,tension=1] coordinates {(d3201-d3102-a1) (d3201-d3102-b1) (d3201)};
			\path[expand bubble]plot [smooth cycle,tension=1] coordinates {(d3201-d3102-a2) (d3201-d3102-b2) (d3201-d3102-c1)};
			
			\path[expand bubble]plot [smooth cycle,tension=1] coordinates {(d3301-d3112-a1) (d3301-d3112-b1) (d3301)};
			\path[expand bubble]plot [smooth cycle,tension=1] coordinates {(d3301-d3112-a2) (d3301-d3112-b2) (d3301-d3112-c1)};
			
			\path[expand bubble]plot [smooth cycle,tension=1] coordinates {(d2310-d2111-c1) (d2200-d2101-c1) (c2)};
			
			\path[expand bubble]plot [smooth cycle,tension=1] coordinates {(d2310-d2111-c2) (d2310-d2111-a1) (d2310-d2111-b1)};
			\path[expand bubble]plot [smooth cycle,tension=1] coordinates {(d2310-d2111-c3) (d2310-d2111-a2) (d2310-d2111-b2)};
			\path[expand bubble]plot [smooth cycle,tension=1] coordinates {(d2310-d2111-c4) (d2310-d2111-a3) (d2310-d2111-b3)};
			\path[expand bubble]plot [smooth cycle,tension=1] coordinates {(d2310-d2111-c5) (d2310-d2111-a4) (d2310-d2111-b4)};
			\path[expand bubble]plot [smooth cycle,tension=1] coordinates {(d2310-d2111-c6) (d2310-d2111-a5) (d2310-d2111-b5)};
			\path[expand bubble]plot [smooth cycle,tension=1] coordinates {(d2310-d2111-c7) (d2310-d2111-a6) (d2310-d2111-b6)};
			\path[expand bubble]plot [smooth cycle,tension=1] coordinates {(d2310-d2111-c8) (d2310-d2111-a7) (d2310-d2111-b7)};
			\path[expand bubble]plot [smooth cycle,tension=1] coordinates {(d2111) (d2310-d2111-a8) (d2310-d2111-b8)};
			
			\path[expand bubble]plot [smooth cycle,tension=1] coordinates {(d2311-d2132-a1) (d2311-d2132-b1) (d2311)};
			\path[expand bubble]plot [smooth cycle,tension=1] coordinates {(d2311-d2132-a2) (d2311-d2132-b2) (d2311-d2132-c1)};
			
			\path[expand bubble]plot [smooth cycle,tension=1] coordinates {(d2211-d2122-a1) (d2211-d2122-b1) (d2211)};
			\path[expand bubble]plot [smooth cycle,tension=1] coordinates {(d2211-d2122-a2) (d2211-d2122-b2) (d2211-d2122-c1)};
			
			\path[expand bubble]plot [smooth cycle,tension=1] coordinates {(d2200-d2101-c2) (d2200-d2101-a1) (d2200-d2101-b1)};
			\path[expand bubble]plot [smooth cycle,tension=1] coordinates {(d2200-d2101-c3) (d2200-d2101-a2) (d2200-d2101-b2)};
			\path[expand bubble]plot [smooth cycle,tension=1] coordinates {(d2200-d2101-c4) (d2200-d2101-a3) (d2200-d2101-b3)};
			\path[expand bubble]plot [smooth cycle,tension=1] coordinates {(d2200-d2101-c5) (d2200-d2101-a4) (d2200-d2101-b4)};
			\path[expand bubble]plot [smooth cycle,tension=1] coordinates {(d2200-d2101-c6) (d2200-d2101-a5) (d2200-d2101-b5)};
			\path[expand bubble]plot [smooth cycle,tension=1] coordinates {(d2200-d2101-c7) (d2200-d2101-a6) (d2200-d2101-b6)};
			\path[expand bubble]plot [smooth cycle,tension=1] coordinates {(d2200-d2101-c8) (d2200-d2101-a7) (d2200-d2101-b7)};
			\path[expand bubble]plot [smooth cycle,tension=1] coordinates {(d2101) (d2200-d2101-a8) (d2200-d2101-b8)};
			
			\path[expand bubble]plot [smooth cycle,tension=1] coordinates {(d2301-d2112-a1) (d2301-d2112-b1) (d2301)};
			\path[expand bubble]plot [smooth cycle,tension=1] coordinates {(d2301-d2112-a2) (d2301-d2112-b2) (d2301-d2112-c1)};
			
			\path[expand bubble]plot [smooth cycle,tension=1] coordinates {(d2201-d2102-a1) (d2201-d2102-b1) (d2201)};
			\path[expand bubble]plot [smooth cycle,tension=1] coordinates {(d2201-d2102-a2) (d2201-d2102-b2) (d2201-d2102-c1)};
			
			\path[expand bubble]plot [smooth cycle,tension=1] coordinates {(d3310-d3111-a1) (d3310-d3111-b1) (d3310-d3111-c2)};
			\path[expand bubble]plot [smooth cycle,tension=1] coordinates {(d3310-d3111-a2) (d3310-d3111-b2) (d3310-d3111-c3)};
			\path[expand bubble]plot [smooth cycle,tension=1] coordinates {(d3310-d3111-a3) (d3310-d3111-b3) (d3310-d3111-c4)};
			\path[expand bubble]plot [smooth cycle,tension=1] coordinates {(d3310-d3111-a4) (d3310-d3111-b4) (d3310-d3111-c5)};
			\path[expand bubble]plot [smooth cycle,tension=1] coordinates {(d3310-d3111-a5) (d3310-d3111-b5) (d3310-d3111-c6)};
			\path[expand bubble]plot [smooth cycle,tension=1] coordinates {(d3310-d3111-a6) (d3310-d3111-b6) (d3310-d3111-c7)};
			\path[expand bubble]plot [smooth cycle,tension=1] coordinates {(d3310-d3111-a7) (d3310-d3111-b7) (d3310-d3111-c8)};
			\path[expand bubble]plot [smooth cycle,tension=1] coordinates {(d3310-d3111-a8) (d3310-d3111-b8) (d3111)};
			
			\path[expand bubble]plot [smooth cycle,tension=1] coordinates {(d3211-d3222-a1) (d3211-d3222-b1) (d3211)};
			\path[expand bubble]plot [smooth cycle,tension=1] coordinates {(d3211-d3222-a2) (d3211-d3222-b2) (d3211-d3222-c1)};
			
			\path[expand bubble]plot [smooth cycle,tension=1] coordinates {(d3211-d3232-a1) (d3211-d3232-b1) (d3311)};
			\path[expand bubble]plot [smooth cycle,tension=1] coordinates {(d3211-d3232-a2) (d3211-d3232-b2) (d3211-d3232-c1)};
		\end{pgfonlayer}
		(d3211-d3232-a2)
	\end{tikzpicture}
}

%% file: sections/strictpophardness.tex
We shall show that determining the existence of a strictly popular outcome for a 3D-EuclidSR game $G$ is co-NP-hard. The co-NP-hardness is demonstrated by proving that the original PC-X3C instance $I = (X,C)$ has a solution $S \subseteq C$ if and only if $G$ has no strictly popular outcome.

\begin{remark}
	Any edge in the construction in \cref{strictpopconstr} has length at least 10.
\end{remark}

\begin{remark}
	Let $\pi$ be an arbitrary outcome of $G$. For an agent $a \in \{\alpha_e[z], \beta_e[z] | z\in [\hat{n}]\}$, we have that $\delta(a, \pi(a)) \geq 1 + \epsilon$. This follows from the construction in \cref{strictpopconstr}
\end{remark}
\begin{remark}
	For any agent $a \notin \{\alpha_e[z], \beta_e[z] | z\in [\hat{n}]\}$, we have that $\delta(a, \pi(a)) \geq 2$. This follows from the construction in \cref{strictpopconstr}
\end{remark}
\begin{remark}
	Let $A_e$ be an arbitrary chain from the construction in \cref{strictpopconstr} and let $\pi$ be an arbitrary outcome of $G$. For an agent $a \in \{\alpha_e[z], \beta_e[z] | z\in [\hat{n}]\}$, we have that $\delta(a, \pi(a)) = 1 + \epsilon$ if and only if $\pi(a) = \{\alpha_e[z], \beta_e[z], \gamma_e[z]\}$ or $\pi(a) = \{\alpha_e[z], \beta_e[z], \gamma_e[z-1]\}$. 
	
	That is, the most preferred room assignment of $\alpha_e[z]$ (resp. $\beta_e[z]$) is to be put together with $\beta_e[z]$ (resp. $\alpha_e[z]$) and $\gamma_e[z']$, where $z' = z$ or $z' = z-1$.
\end{remark}
\begin{remark}
	Let $A_e$ be an arbitrary chain from the construction in \cref{strictpopconstr} and let $\pi$ be an arbitrary outcome of $G$. For any agent $a \notin \{\alpha_e[z], \beta_e[z] | z\in [\hat{n}]\}$ we have that $\delta(a, \pi(a)) = 2$ if and only if $\pi(a) = \{\alpha_e[z], \beta_e[z], \gamma_e[z]\}$ or $\pi(a) = \{\alpha_e[z], \beta_e[z], \gamma_e[z-1]\}$ or $\pi(a) = \{w_j^{i_1}, w_j^{i_2}, w_j^{i_3}\}$ or $\pi(a) = \{d^{l,1}_{i,j}, d^{l,2}_{i,j}, d^{l,3}_{i,j}\}$ or $\pi(a) = \{d^{1,1}_{1,0}, d^{1,2}_{1,0}, d^{1,3}_{1,0}\}$.
	
	That is, the most preferred room assignment of $\gamma_e[z]$ is to be put together with $\alpha_e[z]$ and $\beta_e[z]$ or, if $z >1$, with $\alpha_e[z]$ and $\beta_e[z-1]$. 
	
	In the bottom layer, we have an additional most preferred room in the case where $\gamma_e[z]$ is a set-agent. This room consists of the set-agents that form an equilateral triangle, i.e., $\{w_j^{i_1}, w_j^{i_2}, w_j^{i_3}\}$.
	
	In the top layer, we also have an additional most preferred room for $\gamma_e[z]$. In the case that $\gamma_e[z]$ is an agent from the internal vertex gadget, the additional most preferred room consist of the agents that form an equilateral triangle with $\gamma_e[z]$, i.e., $\{d^{l,1}_{i,j}, d^{l,2}_{i,j}, d^{l,3}_{i,j}\}$.
	
	Finally, the most preferred rooms of a center agent $a$ consists of the other center agents with which it forms an equilateral triangle, i.e., $\{d^{1,1}_{1,0}, d^{1,2}_{1,0}, d^{1,3}_{1,0}\}$, or the vertices from the internal vertex gadget with which it also forms an equilateral triangle, i.e., $\{d^{l,1}_{i,j}, d^{l,2}_{i,j}, d^{l,3}_{i,j}\}$.
\end{remark}

\begin{lemma}\label{lemkeven}
	Let $\pi$ be a outcome of $G$ that assigns all agents to its most preferred room and assume that $\lfloor k \rfloor$ is even. If $\{d^{1,1}_{1,0}, d^{1,2}_{1,0}, d^{1,3}_{1,0}\} \in \pi$, then $\pi = \pi_{pp}$. Otherwise $\pi = \pi_S$ for some solution $S$ of $I$.
\end{lemma}
\begin{proof}
	This follows from the above remarks.
\end{proof}

\begin{lemma}\label{lemkodd}
	Let $\pi$ be a outcome of $G$ that assigns all agents to its most preferred room and assume that $\lfloor k \rfloor$ is odd. If $\{d^{1,1}_{1,0}, d^{1,2}_{1,0}, d^{1,3}_{1,0}\} \in \pi$, then $\pi = \pi_{S}$ for some solution $S$ of $I$. Otherwise $\pi = \pi_{pp}$.
\end{lemma}
\begin{proof}
This follows from the above remarks.
\end{proof}

From the above remarks, we have that if an outcome $\pi$ assigns all agents to its most preferred room, then $\pi= \pi_{pp}$ or $\pi= \pi_{S}$. Note that an outcome that assigns all agents to its most preferred room is popular. 

Since $\pi_{pp}$ exists independent of $I$ having a solution, any popular outcome of $G$ must assign all agents to its most preferred room.

Thus, if $I$ has no solution, then $\pi_{pp}$ is the only popular outcome and therefore a strictly popular outcome. Otherwise, we have multiple popular outcomes and therefore no strictly popular outcome. Therefore, the 3D-EuclidSR problem with the room size fixed to 3 is co-NP-hard.

\begin{theorem}
	The 3D-EuclidSR problem with the room size fixed to 3 is co-NP-hard.
\end{theorem}
\begin{proof}
	Follow from \cref{lemkeven,,lemkodd}.
\end{proof}

%% file: sections/conclusion.tex
\section{Conclusion}
We have demonstrated that determining the existence of a strictly popular outcome for 3D-EuclidSR game is co-NP-hard when the room size is fixed to 3. 

We believe that our construction can be modified to show that the problem remains hard, when to room size is greater than 3. Additionally, we also believe that our construction can be modified to show that determining the existence of a popular outcome is co-NP-hard and that, under the assumption that P$\neq$NP, a mixed popular outcome cannot be computed in polynomial time. These are left for future research.

The complexity of computing and determining the existence of a popular outcome, when the room size is fixed to 2 is also yet unknown. We believe that the problem becomes tractable in this case.

Finally, it would be interesting to investigate whether it would be possible to obtain a stronger result by showing that determining the existence of a strictly popular outcome remains co-NP-hard, when the agents are assigned a point in 2-dimensional euclidean space, i.e., popularity on the Euclid-$d$-SR problem. It is also possible that the problem becomes tractable in this case.